\documentclass[11pt,reqno]{amsart}
\usepackage[foot]{amsaddr}
\usepackage{graphicx,verbatim,lineno,titletoc}
\usepackage{amssymb,mathrsfs}
\usepackage{amsmath}
\usepackage{adjustbox}
\usepackage{calc}
\usepackage{ytableau}
\usepackage{array}
\usepackage{tikz}
\usetikzlibrary{shapes.multipart,patterns,arrows}
\usepackage[font=footnotesize]{caption}
\usepackage[utf8]{inputenc} 
\usepackage[T1]{fontenc} 
\usepackage[english]{babel} 
\usepackage[all]{xy}
\usepackage{enumerate}
\usepackage[english]{babel}
\usepackage{multirow}
\usepackage{float}
\usepackage[caption = false]{subfig}
\usepackage{enumitem}
\usepackage{accents}
\usepackage{hyperref}
\hypersetup{colorlinks}
\usepackage{algorithm}
\usepackage{algpseudocode}
\usepackage{mathtools}
\mathtoolsset{showonlyrefs}
\usepackage{bbold}
\usepackage[toc,page,titletoc]{appendix}
\usepackage{soul}
\usepackage{pifont}
\usepackage{wrapfig}

\usepackage[backend=bibtex,style=nature, sorting=none]{biblatex}

\newcommand{\xmark}{\ding{55}}%

\usepackage[top=1.1in, left=1.1in, right=1.1in, bottom=1.1in]{geometry}
\hypersetup{colorlinks,linkcolor={red},citecolor={olive},urlcolor={red}}

\definecolor{joshcolor}{rgb}{0.2, 0.8, 0.6}

\definecolor{yacoubcolor}{rgb}{0.98, 0.27, 0.62}

\definecolor{hancolor}{rgb}{0.0 0.0, 1.0}

\definecolor{hancolor1}{rgb}{1, 0, 0.0}

\definecolor{hancolor1}{rgb}{1, 0, 0.0}

\definecolor{hanRMKcolor}{rgb}{1.0 0.0, 0.0}

\newcolumntype{M}[1]{>{\centering\arraybackslash}m{#1}}
\newcolumntype{N}{@{}m{0pt}@{}}

\makeatletter
\newcommand{\thickhline}{%
	\noalign {\ifnum 0=`}\fi \hrule height 1pt
	\futurelet \reserved@a \@xhline
}
\newcolumntype{"}{@{\hskip\tabcolsep\vrule width 1pt\hskip\tabcolsep}}

\makeatletter
\newcommand*{\inlineequation}[2][]{%
	\begingroup
	\refstepcounter{equation}%
	\ifx\\#1\\%
	\else
	\label{#1}%
	\fi
	\relpenalty=10000 %
	\binoppenalty=10000 %
	\ensuremath{%
		#2%
	}%
	~\@eqnnum
	\endgroup
}
\makeatother

\newtheorem{theorem}{Theorem}[section]
\newtheorem{prop}[theorem]{Proposition}

\newtheorem{remark}[theorem]{Remark}

\def\one{\mathbb{1}}

\def\R{\mathbb{R}}

\def\E{\mathbb{E}}
\def\P{\mathbb{P}}

\def\x{\mathbf{x}}
\def\y{\mathbf{y}}
\def\z{\mathbf{z}}

\def\G{\mathcal{G}}

\def\param{\boldsymbol{{\mathbf{v}}}}
\def\Param{\boldsymbol{\Theta}}

\newcommand\dataset[1]{\textsc{\texttt{#1}}}
\DeclareTextFontCommand{\textbfit}{\bfseries\itshape}

\newcommand{\tr}{\textup{tr}}

\DeclareMathOperator*{\argmin}{arg\,min}

\newcommand{\JD}{\textup{JD}}

\makeatletter
\let\orgdescriptionlabel\descriptionlabel
\renewcommand*{\descriptionlabel}[1]{%
	\let\orglabel\label
	\let\label\@gobble
	\phantomsection
	\edef\@currentlabel{#1}%
	\let\label\orglabel
	\orgdescriptionlabel{#1}%
}
\makeatother

\allowdisplaybreaks

\renewcommand{\descriptionlabel}[1]{%
	\hspace\labelsep \upshape\bfseries #1%
}

\makeatletter
\newcommand{\addresseshere}{%
	\enddoc@text\let\enddoc@text\relax
}
\makeatother

\begin{document}

	\title[Learning low-rank latent mesoscale structures in networks]{Learning low-rank latent mesoscale structures in networks}

	\author{Hanbaek Lyu$^{\dagger}$}
	\address{$\dagger$ Department of Mathematics, University of Wisconsin-Madison, WI 53706, USA}
	\author{Yacoub H. Kureh$^{*}$}
	\author{Joshua Vendrow$^{*}$}
	\author{Mason A. Porter$^{*,**}$}
	
	\address{$*$ Department of Mathematics, University of California, Los Angeles, CA 90095, USA}
	\address{$**$ Santa Fe Institute, Sante FE, NM 87501, USA}
	
	\email{ \texttt{hlyu@math.wisc.edu}, \texttt{\{ykureh, jvendrow, mason\}@math.ucla.edu}}

	\thanks{Our code for our algorithms and simulations is available at \url{https://github.com/HanbaekLyu/NDL_paper}. At \url{https://github.com/jvendrow/Network-Dictionary-Learning}, we provide a user-friendly version as a {\sc Python} package {\textsc{ndlearn}}.}


	\begin{abstract}	
		
		It is common to use networks to encode the architecture of interactions between entities in complex systems {in} the physical, biological, social, and information sciences. To study the large-scale behavior of complex systems, it is useful to {examine} mesoscale structures in networks as building blocks that influence such behavior~\cite{onnela2012taxonomies, khambhati2018modeling}. 
		We present a new approach for describing low-rank mesoscale structures in networks, and we illustrate our approach using several synthetic network models and empirical friendship, collaboration, and protein--protein interaction (PPI) networks. We find that these networks possess 
		a relatively small number of `latent motifs' that together can successfully approximate most subgraphs of a network at a fixed mesoscale. We use an algorithm for `network dictionary learning' (NDL)~\cite{lyu2020online}, which combines a network-sampling method~\cite{lyu2023sampling} and 
		nonnegative matrix factorization~\cite{lee1999learning, lyu2020online}, to learn the latent motifs of a given network. 
		The ability to encode a network using a set of
		latent motifs has a wide variety of applications to network-analysis tasks, such as comparison, denoising, and edge inference.  
		Additionally, using a new network denoising and reconstruction (NDR) algorithm, we demonstrate how to denoise a 
		corrupted network by using only the latent motifs that one learns directly from the corrupted network. 
	\end{abstract}

	
	\maketitle

	


	It is often insightful to examine structures in networks~\cite{newman2018} at intermediate scales (i.e., at `mesoscales') that lie between the microscale of nodes and edges {and the} macroscale distributions of local network properties. Researchers have considered subgraph patterns (i.e., the connection patterns of subsets of nodes) as building blocks of network structure at various mesoscales~\cite{schwarze2020motifs}. 
	In many studies of networks, researchers identify $k$-node (where $k$ is typically between $3$ and $5$) subgraph patterns of a network that are unexpectedly common in comparison to some random-graph null model as `motifs' of that network~\cite{milo2002network}. In the past two decades, the study of motifs has been important for the analysis of
	networked systems in many areas, including biology~\cite{conant2003convergent, rip2010experimental, sporns2004motifs, ristl2014complex, alon2007network}, sociology~\cite{hong2014social, juszczyszyn2009temporal}, and economics~\cite{ohnishi2010network, takes2018multiplex}. However, to the best of our knowledge, researchers have not examined how to use such motifs (or related mesoscale structures), after their discovery, as building blocks to reconstruct a network. In {the present} paper, we provide this missing computational framework to bridge inferred subgraph-based mesoscale structures and the global structure of networks. To do this, we propose (1) a `network dictionary learning' (NDL) algorithm that learns `latent motifs' from samples of certain random $k$-node subgraphs and (2) a complementary algorithm for `network denoising and reconstruction' (NDR) that constructs a best `mesoscale linear approximation' of a given network using the learned latent motifs. We also provide a rigorous theoretical analysis of the proposed algorithms. This analysis includes a novel result in which we prove that one can accurately reconstruct an entire network if one has a dictionary of latent motifs that can accurately approximate mesoscale structures of the network. We compare {our approach} to related prior work~\cite{lyu2020online} in the `Methods' section and in our Supplementary Information (SI).

	Using our approach, we find that various real-world networks (such as Facebook friendship networks, Coronavirus and \textit{Homo sapiens} protein--protein interaction (PPI) networks, and an arXiv collaboration network) have low-rank subgraph patterns, in the sense that one can successfully approximate their $k$-node subgraph patterns by a weighted sum of a small number of latent motifs. The latent motifs of these networks thereby reveal low-rank mesoscale structures of these networks. Our claim of the low-rank nature of such mesoscale structures
	concerns the space of certain subgraph patterns, rather than the embedding of an entire network into a low-dimensional Euclidean space (as considered in spectral-embedding and graph-embedding methods~\cite{perozzi2014deepwalk, grover2016node2vec}). It is impossible to obtain such a low-dimensional graph embedding for networks with small mean degrees and large clustering coefficients~\cite{seshadhri2020impossibility}. Additionally, as we demonstrate in this paper, the ability to encode a network using a set of latent motifs has a wide variety of applications in network analysis. {These} applications include network comparison, denoising, and edge inference.
	
	
	\subsection*{Motivating application: Anomalous-subgraph detection}

	A common problem 
	in network analysis is the detection of anomalous subgraphs of a network (see Figure \ref{fig:anomaly_detection}) \cite{akoglu2015graph}. The connection pattern of an anomalous subgraph distinguishes it from the rest of a network.
	This anomalous-subgraph-detection problem has numerous high-impact applications, including in security, finance, healthcare, and law enforcement \cite{noble2003graph,  miller2015spectral}. Various approaches, including both classical techniques \cite{akoglu2015graph} and modern deep-neural-network techniques \cite{ma2021comprehensive}, have been proposed to detect anomalous subgraphs.

		
		

	\begin{figure*}[h]
		\centering
		\includegraphics[width=1 \linewidth]{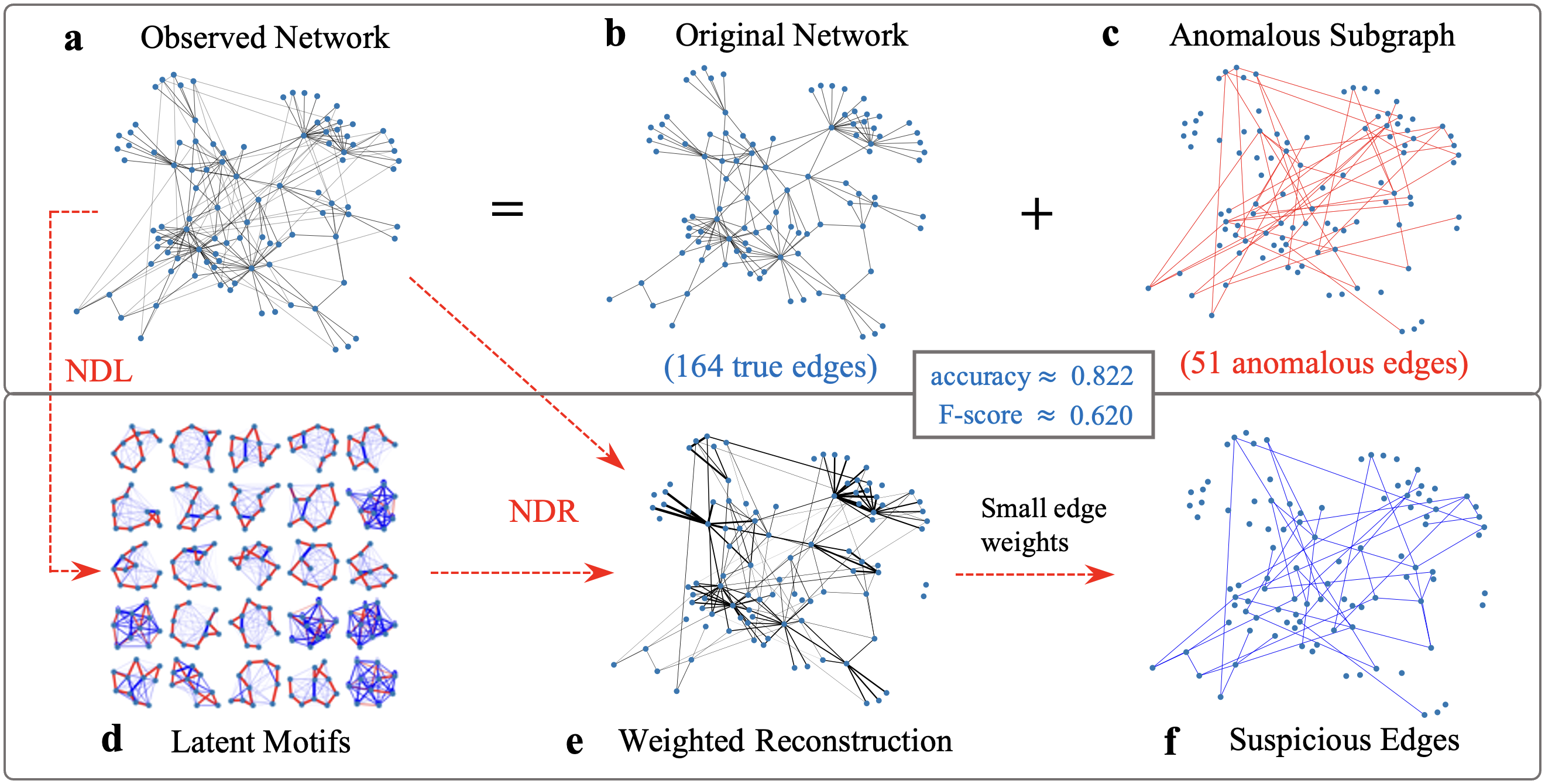}
		\caption{Illustration of anomalous-subgraph detection using 
		network reconstruction. 
		(\textbf{a}) The observed network consists of (\textbf{b}) the original network and (\textbf{c}) anomalous edges, and we seek to detect the anomalous edges in the observed network. In our approach, we first 
		(\textbf{d}) determine a set of latent motifs and then (\textbf{e}) use them to reconstruct the observed network. (\textbf{f}) In the weighted reconstruction of the network, we identify the edges with positive but small weights as suspicious edges. 
		We compute the accuracy and the F-score for inferring the anomalous edges in \textbf{c} as the suspicious edges in \textbf{f}, 
		where the F-score is the harmonic mean of the precision and recall scores.
		}
		\label{fig:anomaly_detection}
	\end{figure*}

		Consider the following simple conceptual framework for anomalous-subgraph detection. 
		\begin{itemize}
		\item{We learn ``normal subgraph patterns'' in an observed network and then seek to detect subgraphs in the observed network that deviate significantly from them.}
		\end{itemize}
		 By studying low-rank mesoscale structures in networks, we can turn this high-level idea for anomalous-subgraph detection into a concrete approach, which we now briefly summarize.
		First, we compute latent motifs (see Figure \ref{fig:anomaly_detection}\textbf{d}) of an observed network (see Figure \ref{fig:anomaly_detection}\textbf{a}) that can successfully approximate the $k$-node subgraphs of the observed network. A key observation is that these subgraphs should also describe the normal subgraph patterns of the observed
		network (see Figure \ref{fig:anomaly_detection}\textbf{b}). The rationale that underlies this observation is that the $k$-node subgraphs of the observed
		network likely form a low-rank space, so we expect the latent motifs to be robust with respect to the addition of
		anomalous edges (see Figure \ref{fig:anomaly_detection}\textbf{c}). 
		Consequently, reconstructing the observed network using its latent motifs yields a weighted network (see Figure \ref{fig:anomaly_detection}\textbf{e}) in which edges with positive and small weights deviate significantly from the normal subgraph patterns, which are  
		captured 
		by the latent motifs. 
		Therefore, such edges are likely to be anomalous.
		The suspicious edges (see Figure \ref{fig:anomaly_detection}\textbf{f}) are the edges in the weighted reconstruction that have positive weights that are less than a threshold. One can determine the threshold using a small set of known true edges and known anomalous edges. 
		The suspicious edges match well with the
		anomalous edges in Figure \ref{fig:anomaly_detection}\textbf{c}. 
		See the SI for more details.

	In the remainder of our paper, we carefully develop the three key components of our approach:
	 (1) effective sampling of $k$-node subgraphs; (2) reconstructing observed networks using candidate
	 latent motifs; and (3) computing latent motifs from observed networks. 
	 The key idea of our work is to approximate sampled subgraphs by latent motifs and then combine these approximations to construct a weighted reconstructed network. We illustrate this procedure in Figure \ref{fig:recons_illustration}.
	 We also present a variety of supporting numerical experiences using several synthetic and real-world networks.
	
	
	\subsection*{$k$-path motif sampling and latent motifs}

		Computing all $k$-node subgraphs of a network is computationally expensive and is the main computational bottleneck of traditional motif analysis~\cite{milo2002network}. Our approach, which bypasses this issue, is to learn latent motifs by drawing random samples of a particular class of $k$-node connected subgraphs. We consider random $k$-node subgraphs that we obtain by uniformly randomly sampling a `$k$-path' from a network and including all edges between the sampled nodes of the network. A sequence $\x=(x_{1},\ldots,x_{k})$ of $k$ (not necessarily distinct) nodes is a \textit{$k$-walk} if $x_{i}$ and $x_{i+1}$ are adjacent for all $i \in \{1, \ldots, k-1\}$. A $k$-walk is a \textit{$k$-path} if all nodes in the walk are distinct (see Figure~\ref{fig:subgraphs}). 
		Sampling a $k$-path serves two purposes: (1) {it} ensures that the sampled $k$-node induced subgraph is connected with the minimum number of imposed edges; and (2) it induces a natural node ordering of the $k$-node induced subgraph. (Such an ordering is important for computations that involve subgraphs.) 
		By using the `$k$-walk' motif-sampling algorithm in~\cite{lyu2023sampling} in conjunction with rejection sampling, {one} can sample a large number of $k$-paths and obtain their associated induced subgraphs.

		\begin{figure*}[h]
			\centering
			\includegraphics[width=1 \linewidth]{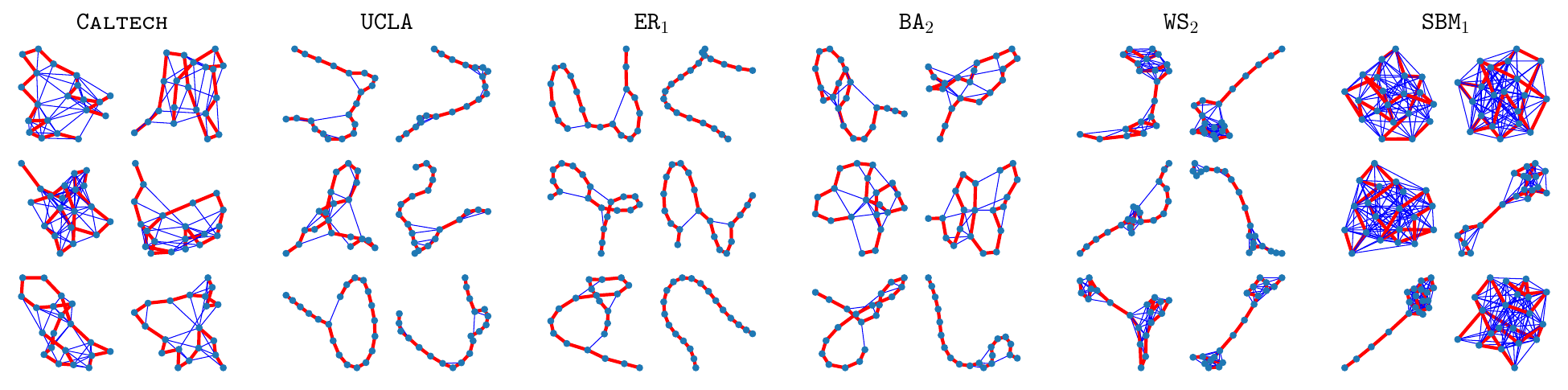}
			\vspace{-0.5cm}
			\caption{Six subgraphs that are induced by uniformly sampled $k$-paths with $k = 20$ (red edges) from {the} \dataset{Caltech} and \dataset{UCLA} Facebook networks, an Erd\H{o}s--R\'{e}nyi (ER) random graph (\dataset{ER}$_{1}$), a Barab\'{a}si--Albert (BA) random graph (\dataset{BA}$_{2}$), a Watts--Strogatz (WS) small-world network (\dataset{WS}$_{2}$), and a stochastic-block-model (SBM) network (\dataset{SBM}$_{1}$). See the `Methods' section for more details about these networks.}
			\label{fig:subgraphs}
		\end{figure*}

		The $k$-node subgraphs that are induced by uniformly randomly sampling $k$-paths from a network
		are the mesoscale structures that we consider in the present paper. 
		We use the term `on-chain edges' for the edges of these subgraphs between nodes $x_{i}$ and $x_{i+1}$ for $i\in \{1,\ldots,k-1\}$, and we use the term `off-chain edges' for all other edges. It is the off-chain edges that can differ across subgraphs and hence encode meaningful information about the network. For $k=2$, the subgraphs are all isomorphic to a 2-path and hence have no off-chain edges. For $k=3$, the subgraphs can have a single off-chain edge, so they are isomorphic either to a 2-path or to a 3-clique (i.e., a graph with three nodes and all three possible edges between them).
		For larger values of $k$, the subgraphs can have diverse connection patterns (see Figure~\ref{fig:subgraphs}), depending on the architecture of the original network).

		We study the connection patterns of {a} random $k$-node subgraph by decomposing it as a weighted sum of more elementary subgraph patterns (possibly with continuous-valued edge weights), which we call \textit{latent motifs} (see Figures \ref{fig:recons_illustration}\textbf{a1}--\textbf{a3}).
		To study mesoscale structures in networks, we investigate several questions. How many distinct latent motifs does one need to successfully approximate all {of these} $k$-node subgraph patterns? What do they look like? How do these latent motifs differ for different networks?

		\begin{figure*}[h]
			\centering
			\includegraphics[width=1 \linewidth]{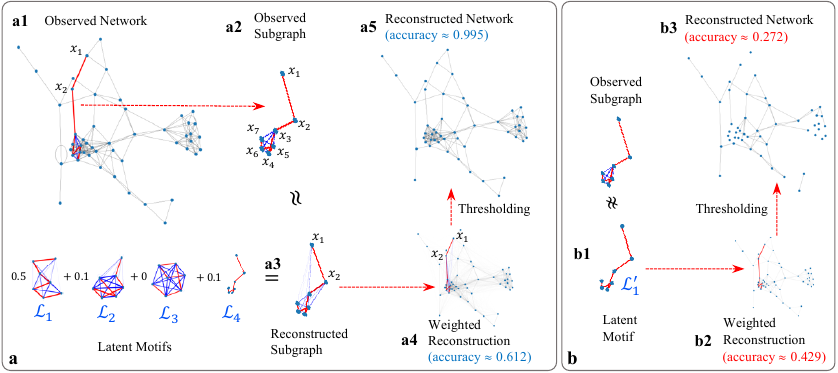}
			\caption{An illustration of our low-rank network-reconstruction process using latent motifs. Given (\textbf{a1}) an observed network and a set of latent motifs $\mathcal{L}_{1},\ldots,\mathcal{L}_{r}$, we repeatedly sample a $k$-path and approximate the induced subgraph {(\textbf{a2})} of the nodes in that path by (\textbf{a3}) a nonnegative linear combination of the latent motifs $\mathcal{L}_{1},\ldots,\mathcal{L}_{r}$.
			{We then compute the weighted reconstruction in \textbf{a4} by taking the edge weight between each unordered node pair $\{x,y\}$ to be the mean of the reconstructed weights of $\{x,y\}$ from all sampled subgraphs that include both $x$ and $y$.}
			 We measure the accuracy of the reconstruction of the weighted network by calculating $1$ minus the Jaccard distance in \eqref{eq:JD} in the SI. We then obtain (\textbf{a5}) an unweighted (i.e., binary) reconstructed network by thresholding the edge weights in \textbf{a4} with a threshold 0.5. That is, we retain edges whose weights are at least
			 0.5, and we remove all other edges.
			  We measure the accuracy of the reconstruction of the binary network by calculating the Jaccard index between the original network's edge set and the associated reconstructed network's edge set. The same
			  network-reconstruction process using (\textbf{b1}) a single latent motif $\mathcal{L}_{1}'$, which is the $k$-path, yields (\textbf{b2}) weighted and (\textbf{b3}) binary reconstructions with lower accuracies than those in \textbf{a4} and \textbf{a5}, respectively.
			}
			\label{fig:recons_illustration}
		\end{figure*}

		
		\subsection*{Low-rank network reconstruction using latent motifs}

	Suppose that we have a network $G = (V,E)$ and two collections,
		 $W=\{ \mathcal{L}_{1},\dots,\mathcal{L}_{r} \}$ and $W'=\{ \mathcal{L}_{1}',\dots,\mathcal{L}_{r'}' \}$, of latent motifs. 
		 How can one determine which of the collections better describes the mesoscale structure of $G$? One can sample a large number of $k$-node subgraphs $A_i$ of $G$ and, for each $A_i$, independently, determine the 
		 nonnegative linear combination of latent motifs that yields the closest approximation $\hat{A}_i$. 
		 By comparing the subgraphs $A_i$ with their corresponding approximations $\hat{A}_i$, one can demonstrate how well the
		  latent motifs 
		  in the `network dictionary'
		  $W$ 
		  approximates
		  the $k$-node subgraph patterns of
		  $G$ (see Figures \ref{fig:recons_illustration}\textbf{a1}--\textbf{a3}).

		For applications such as anomalous-subgraph detection, it is helpful to construct a weighted network $G_{\textup{recons}}$ with the same node set $V$ that gives the `best approximation' of $G$ using the network dictionary $W$. 
	We regard the network $G_{\textup{recons}}$ as a `rank-$r$ mesoscale reconstruction' of $G$. If $G_{\textup{recons}}$ is close to $G$, we conclude that the latent motifs in $W$ successfully capture the structure of $k$-node subgraphs of $G$ and that $G$ has a rank-$r$ subgraph 
	patterns that is prescribed by the latent motifs in $W$. 
	We interpret the edge weights in $G_{\textup{recons}}$ as measures of confidence in the corresponding edges in $G$ with respect to $W$. For example, if an edge $e$ has the smallest weight in $G_{\textup{recons}}$, we interpret it as the most `outlying' edge with respect to the latent motifs in $W$ (see Figures \ref{fig:anomaly_detection}\textbf{e},\textbf{f}). We can threshold the weighted edges of $G_{\textup{recons}}$ at some fixed value $\theta \in [0,1]$ to obtain an undirected reconstructed network $G_{\textup{recons}}(\theta)$ with binary edge weights (which are either $0$ or $1$). We can then directly compare $G_{\textup{recons}}(\theta)$ to the original unweighted network $G$. 

		Our \textit{network denoising and reconstruction (NDR) algorithm (see Algorithm~\ref{alg:network_reconstruction} in the SI) works as follows.} We seek to build a weighted network $G_{\textup{recons}}$ using the node set $V$ and a weighted adjacency matrix $A_{\textup{recons}}: V^{2}\rightarrow \R$. This network best approximates the observed network $G$, whose subgraphs are generated by the latent motifs in $W$. First, we uniformly randomly sample a large number $T$ of $k$-paths $\x_{1},\ldots,\x_{T}:{1,\ldots,k} \rightarrow V$ in $G$. We then determine the $k\times k$ unweighted matrices $A_{\x_{1}},\ldots,A_{\x_{T}}$ with entries $A_{\x_{t}}(i,j)=A(\x_{t}(i), \x_{t}(j))$, which equals 1 if nodes $\x_{t}(i)$ and $\x_{t}(j)$ are adjacent in the network and equals 0 otherwise. These are the adjacency matrices of the induced subgraphs of the
		nodes of the $k$-paths that we sampled (see Figure \ref{fig:recons_illustration}\textbf{a2}). 
		We then approximate each $A_{\x_{t}}$ by a nonnegative linear combination $\hat{A}_{\x_{t}}$ of the latent motifs in $W$ (see Figure \ref{fig:recons_illustration}\textbf{a3}). 
		We then we compute $A_{\textup{recons}}(x,y)$ for each $x,y\in V$ as the mean of $\hat{A}_{\x_{t}}(a,b)$ over all $t\in {1,\ldots, T}$ and all $a,b\in {1,\ldots,k}$ such that $\x_{t}(a) = x$ and $\x_{t}(b) = y$ (see Figure \ref{fig:recons_illustration}\textbf{a4}). {We also provide theoretical guarantees and error bounds for our NDR algorithm in the SI {(see Algorithm~\ref{alg:network_reconstruction})}.

		Consider reconstructing a network $G$ using a single latent motif $\mathcal{L}_{1}'$ that is a $k$-path. We begin with the case $k=2$, such that each subgraph that we sample is a $2$-path. 
		{The sampled $2$-paths} are approximated perfectly by
		$\mathcal{L}_{1}'$ (see Figure \ref{fig:recons_illustration}\textbf{b1}). A $2$-path that one chooses uniformly at random has an equal probability of sampling each edge of $G$,
		so $G_{\textup{recons}} = G$.
		 {Therefore, we conclude that, at scale $k=2$, one can perfectly reconstruct $G$ by using the $2$-path latent motif $\mathcal{L}_{1}'$.}
			However, for $k\ge 3$, the graph $G_{\textup{recons}}$ can differ significantly from $G$, as approximating the observed subgraphs by a single $k$-path misses all of the off-chain edges (see Figures \ref{fig:recons_illustration}\textbf{b1}--\textbf{b3}). 
			Therefore, to properly describe the $k$-node subgraph  
			patterns of $G$, one may need more than one latent motif with  
			off-chain edges (see Figure \ref{fig:recons_illustration}\textbf{a3}). 
			We give more details in Appendix~\ref{subsection:NR} of the SI.



		\subsection*{Dictionary learning and latent motifs}
		
		{How does one compute latent motifs from a given network?}
		\textit{Dictionary-learning} algorithms are machine-learning techniques that learn interpretable latent structures of complex data sets. They are employed regularly in the data analysis of text and images~\cite{elad2006image, mairal2007sparse, peyre2009sparse}. Dictionary-learning algorithms usually consist of two steps. First, one samples a large number of structured subsets of a data set (e.g., square patches of an image or collections of a few sentences of a text); we refer to such a subset as a \textit{mesoscale patch} of a data set. Second, one finds a set of basis elements such that taking a nonnegative linear combination of them can successfully approximate each of the sampled mesoscale patches. Such a set of basis elements is called a \textit{dictionary}, and one can interpret each basis element as a latent structure of the data set.

		\begin{figure*}[h]
			\centering
			\includegraphics[width=1 \linewidth]{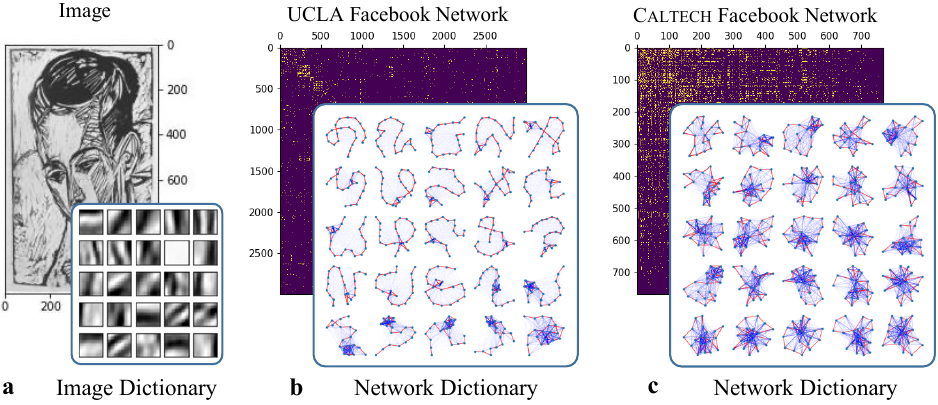}
			\vspace{-0.5cm}
			\caption{Illustration of mesoscale structures that we learn from (\textbf{a}) images and (\textbf{b},\textbf{c}) networks. In each experiment in this figure, we form a matrix $X$ of size $d \times n$ by sampling $n$ mesoscale patches of size $d = 21 \times 21$ from the corresponding object. 
				For the image in panel {\textbf{a}}, the columns of $X$ are square patches of $21\times 21$ pixels. In panels \textbf{b} and \textbf{c}, we show portions of the associated adjacency matrices of the two networks. We take the columns of $X$ to be the $k\times k$ adjacency matrices of the connected subgraphs that are induced by a path of $k = 21$ nodes, where a $k$-node path 
				consists of $k$ distinct nodes $x_{1}, \ldots, x_{k}$ such that $x_{i}$ and $x_{i+1}$ are adjacent for all $i \in \{1, \ldots, k-1\}$.
				Using nonnegative matrix factorization (NMF), we compute an approximate factorization $X \approx WH$ into nonnegative matrices $W$ and $H$,
				where $W$ is called a `dictionary' and has $r = 25$ columns. Because of this factorization, we can approximate any sampled mesoscale patches (i.e., the columns of $X$) of 
				an object by a nonnegative linear combination of the columns of $W$, which we interpret as latent shapes for the images and latent motifs (i.e., subgraphs) for the networks, respectively.  The columns of $H$ give the coefficients in these linear combinations. The network dictionaries of latent motifs that we learn from the (\textbf{b}) \dataset{UCLA} and (\textbf{c}) \dataset{Caltech} Facebook networks have distinctive social structures. In the adjacency matrix of the UCLA network, we show only the first 3000 nodes (according to the node labeling in the data set). The image in panel \textbf{a} is from the collection \dataset{Die Graphik Ernst Ludwig Kirchners bis 1924, von Gustav Schiefler Band I bis 1916} (Accession Number 2007.141.9, Ernst Ludwig Kirchner, 1926). We use the image with permission from the National Gallery of Art in Washington, DC, USA.]
			}
			\label{fig:img_ntwk_dict}
		\end{figure*}

		As an example, consider the artwork image in Figure~\ref{fig:img_ntwk_dict}\textbf{a}. We first uniformly randomly sample 10,000 square patches of $21 \times 21$ pixels and vectorize them to {obtain} a $21^2 \times \text{10,000}$ matrix $X$. 
		The choice of vectorization $\R^{k\times k}\rightarrow \R^{k^{2}}$ is arbitrary; we use the column-wise vectorization in Algorithm~\ref{alg:vectorize} {in the SI}.	We then use a nonnegative matrix factorization (NMF)~\cite{lee1999learning} algorithm to find an approximate factorization $X \approx WH$, where $W$ and $H$ are nonnegative matrices of sizes $21^2 \times 25$ and $25 \times \text{10,000}$, respectively. Reshaping the columns of $W$ into $21 \times 21$ square images yields an image dictionary that describes `latent shapes' {of} the image. 
		
		Our \textit{network dictionary learning} (NDL) algorithm to compute
		a `network dictionary' that consists of latent motifs is based on a similar idea. As mesoscale patches of a network, we use the $k \times k$ binary (i.e., unweighted) matrices that encode connection
		patterns between the nodes that form a uniformly random $k$-path. After obtaining sufficiently many mesoscale patches of a network (e.g., by using a motif-sampling algorithm~\cite{lyu2023sampling} with rejection sampling), we apply a dictionary-learning algorithm (e.g., NMF~\cite{lee1999learning}) to obtain latent motifs of the network. A latent motif is a $k$-node weighted network with nodes
		$\{1,\ldots,k\}$ and edges that have weights between $0$ and $1$. 
		We use the term `on-chain edges' for the edges of a latent motif between nodes $i$ and $i+1$ for $i\in \{1,\ldots,k-1\}$; we use the term `off-chain edges' for all other edges. We give more background about our NDL algorithm in the `Methods' section and provide a complete implementation of our approach in Algorithm~\ref{alg:NDL} {in} the SI. We give theoretical guarantees for Algorithm~\ref{alg:NDL} in Theorems~\ref{thm:NDL} and~\ref{thm:NDL2} in the SI.

		{In Figure~\ref{fig:img_ntwk_dict}, we compare 25 latent motifs with $k = 21$ nodes of Facebook friendship networks (which were collected on one day in fall 2005) from UCLA (`\dataset{UCLA}') and Caltech (`\dataset{Caltech}')~\cite{Traud2011,traud2012social}. Each node in one of these networks is a Facebook account of an individual, and each edge encodes a Facebook friendship between two individuals. The latent motifs reveal striking differences between these networks in the connection patterns of the subgraphs that are induced by $k$-paths with $k = 21$.
			For example, the latent motifs in \dataset{UCLA}'s dictionary (see Figure~\ref{fig:img_ntwk_dict}\textbf{b}) have sparse off-chain connections with a few clusters, whereas \dataset{Caltech}'s dictionary (see Figure~\ref{fig:img_ntwk_dict}\textbf{c}) has relatively dense off-chain connections. 
			Most of \dataset{Caltech}'s latent motifs have `hub' nodes (which are adjacent to many other nodes in the latent motif) or communities~\cite{porter2009communities,fortunato2016} 
			with six or more nodes. (See Figure~\ref{fig:Figure_boxcompare} in the SI for community-size statistics.)  An important property of $k$-node latent motifs is that any network structure (e.g., hub nodes, communities, and so on) in the latent motifs must also exist in actual $k$-node subgraphs. 
			We observe both hubs and communities in the subgraphs samples from \dataset{Caltech} in Figure~\ref{fig:subgraphs}. By contrast, most of \dataset{UCLA}'s latent motifs do not have such structures, as is also the case for the subgraph samples from \dataset{UCLA} in Figure~\ref{fig:subgraphs}. 
			
			Because $k$-node latent motifs encode basic 
			{connection}
			patterns of $k$ nodes that are at most $k - 1$ edges apart, one can interpret $k$ as a scale parameter. Latent motifs {that one learns} from the same network {for} different values of $k$ {reveal} different mesoscale structures. See Figure \ref{fig:latent_motifs_multiscale_1} in the SI for more details.

		\subsection*{Example networks}
		
		We demonstrate our approach using 16 example networks; 8 of them are real-world networks and 8 of them synthetic networks. The 8 real-world networks are \dataset{Coronavirus PPI} (for which we use the shorthand \dataset{Coronavirus})~\cite{oughtred2019biogrid, oughtred2019biogrid_covid, gordon2020sars} and \dataset{Homo sapiens PPI} (for which we use the shorthand \dataset{H. sapiens})~\cite{oughtred2019biogrid, grover2016node2vec}; Facebook networks from \dataset{Caltech}, \dataset{UCLA}, \dataset{Harvard}, and \dataset{MIT}~\cite{Traud2011,traud2012social}; \dataset{SNAP Facebook} (for which we use the shorthand \dataset{SNAP FB})
		\cite{leskovec2012learning, grover2016node2vec}; and \dataset{arXiv ASTRO-PH} (for which we use the shorthand \dataset{arXiv})~\cite{Leskovec2014SNAP,grover2016node2vec}. The first network is a protein--protein interaction (PPI) network of proteins that are related to the coronaviruses that cause Coronavirus disease 2019 (COVID-19), Severe Acute Respiratory Syndrome (SARS), and Middle Eastern Respiratory Syndrome (MERS)~\cite{oughtred2019biogrid_covid}. The second network is a PPI network of proteins that are related to \emph{Homo sapiens}~\cite{oughtred2019biogrid}. The third network is a 2012 Facebook network that was collected from participants in a survey~\cite{leskovec2012learning}. 
		The fourth network is a collaboration network {from coauthorships} of preprints that were posted in the astrophysics category of the arXiv preprint server. The last four real-world networks are 2005 Facebook networks from four universities from the {\sc Facebook100} data set~\cite{traud2012social}. In each Facebook network, nodes represent accounts and edges encode Facebook `friendships' between these accounts.
		
		For the eight synthetic networks, we generate two instantiations each of Erd\H{o}s--R\'{e}nyi (ER) $G(N,p)$ networks~\cite{erdds1959random}, Watts--Strogatz (WS) networks~\cite{watts1998collective},  Barab\'{a}si--Albert (BA) networks~\cite{barabasi1999emergence}, and stochastic-block-model (SBM) networks~\cite{holland1983stochastic}.	
		These four random-graph models are well-studied and are common choices for testing new network methods and models~\cite{newman2018}. Each of the ER networks has 5,000 nodes, and we independently connect each pair of nodes with probabilities of $p = 0.01$ (in the network that we call \dataset{ER$_{1}$}) and $p = 0.02$ (in \dataset{ER$_{2}$}). For the WS networks, we use rewiring probabilities of $p = 0.05$ (in \dataset{WS$_{1}$}) and $p = 0.1$ (in \dataset{WS$_{2}$}) starting from a 5,000-node ring network in which each node is adjacent to its $50$ nearest neighbors. For the BA networks, we use $m = 25$ (in \dataset{BA$_{1}$}) and $m = 50$ (in \dataset{BA$_{2}$}), where $m$ denotes the number of edges of each new node when it connects (via linear preferential attachment) to the existing network, which we grow from an initial network of $m$ isolated nodes (i.e., none of them are adjacent to any other node) until it has 5,000 nodes. The SBM networks \dataset{SBM}$_{1}$ and \dataset{SBM}$_{2}$ have three planted 1,000-node communities{; two} nodes in the $i_{0}$th and the $j_{0}$th communities are connected by an edge independently with probability $0.5$ if $i_{0} = j_{0}$ (i.e., if they are in the same community) and $0.001$ for \dataset{SBM}$_{1}$ and $0.1$ for \dataset{SBM}$_{2}$ if $i_{0} \neq j_{0}$ (i.e., if they are in different communities). See the `Methods' section for more details.


		\subsection*{Network-reconstruction experiments}
	
		An important observation is that one can reconstruct a given network using an arbitrary network dictionary, including ones that one learns from an entirely different network.
		Such a `cross-reconstruction' allows one to quantitatively compare the learned mesoscale structures of different networks.
		In Figure~\ref{fig:network_recons}, we show the results of several network-reconstruction experiments using a variety of real-world networks and synthetic networks. 
		We label each subplot of Figure~\ref{fig:network_recons} with  
		$X \leftarrow Y$
		to indicate that we are reconstructing network  
		$X$
		by approximating mesoscale patches of  
		$X$
		using a network dictionary that we learn from network 
		$Y$.
		We perform these experiments for various values of the edge threshold $\theta \in [0,1]$ and $r \in \{9,16,25,36,49,81,100\}$ latent motifs in a single dictionary. Each network dictionary in Figure~\ref{fig:network_recons} has $k = 21$ nodes, for which the dimension of the space of all possible mesoscale patches (i.e., the adjacency matrices of the induced subgraphs) is $\binom{21}{2} - 20 = 190$. 
		We measure the reconstruction accuracy by calculating the Jaccard index between the original network's edge set and the reconstructed network's edge set. That is, to measure the similarity of two edge sets, we calculate the number of edges in the intersection of these sets divided by the number of edges in the union of these sets. This gives a 
		measure of reconstruction accuracy; 
		if the Jaccard index equals $1$, the
		reconstructed network is precisely the same as	the original {network}. 
		We obtain the same qualitative results as in Figure~\ref{fig:network_recons} if we instead measure similarity using the Rand index~\cite{rand1971objective}).

		\begin{figure*}[h]
			\centering
			\includegraphics[width=1 \linewidth]{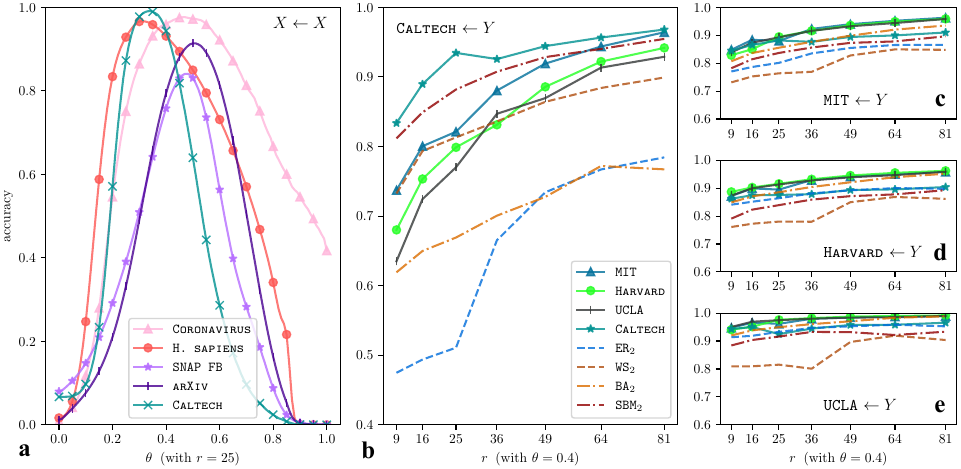}
			\caption{The self-reconstruction and cross-reconstruction accuracies of several real-world and synthetic networks 
				versus the edge threshold $\theta$ and the number $r$ of latent motifs in a network dictionary. The label $X \leftarrow Y$ indicates that we reconstruct network $X$ using a network dictionary that we learn from network $Y$. The reconstruction process produces a weighted network that we turn into an unweighted network
				by thresholding the edge weights at a threshold value $\theta$; we keep only edges whose weights are strictly larger than $\theta$.
				We measure reconstruction accuracy by calculating the Jaccard index of an original network's edge set and an associated reconstructed network's edge set. 
				In panel \textbf{a}, we plot accuracies versus $\theta$ (with the number of latent motifs fixed at $r = 25$), where $X$ is one of five real-world networks (two PPI networks, two Facebook networks, and one collaboration network). 
				In panels \textbf{b}--\textbf{e}, we reconstruct each of the four Facebook networks using network dictionaries with $r \in \{9,16,25,36,64,81,100\}$ latent motifs that we learn from one of eight networks (with the threshold value fixed at $\theta = 0.4$). 
			}
			\label{fig:network_recons}
		\end{figure*}

		In Figure~\ref{fig:network_recons}\textbf{a}, we plot the accuracy of the `self-reconstruction' $X \leftarrow X$ versus the threshold $\theta$ (with $r = 25$ latent motifs), where $X$ is one of the real-world networks \dataset{Coronavirus}, \dataset{H. sapiens}, \dataset{SNAP FB}, \dataset{Caltech}, and \dataset{arXiv}.
		 The accuracies for {\dataset{Coronavirus}}, \dataset{H. sapiens}, and \dataset{Caltech} peak above $95\%$ when $\theta \approx 0.4$; the accuracies for \dataset{arXiv} and \dataset{SNAP FB} peak above $88\%$ and $70\%$, respectively, for $\theta \approx 0.6$. We choose $\theta = 0.4$ for the cross-reconstruction experiments for the Facebook networks \dataset{Caltech}, \dataset{Harvard}, \dataset{MIT}, and \dataset{UCLA} in Figures~\ref{fig:network_recons}\textbf{b},\textbf{c}. These four Facebook networks have self-reconstruction accuracies above $80\%$ for $r = 25$ motifs with a threshold of $\theta = 0.4$. The total number of dimensions when using mesoscale patches at scale $k = 21$ is $190$, so this result suggests that all eight of these real-world networks have low-rank mesoscale structures at scale $k = 21$.

		\begin{figure*}[h]
			\centering
			\includegraphics[width=1 \linewidth]{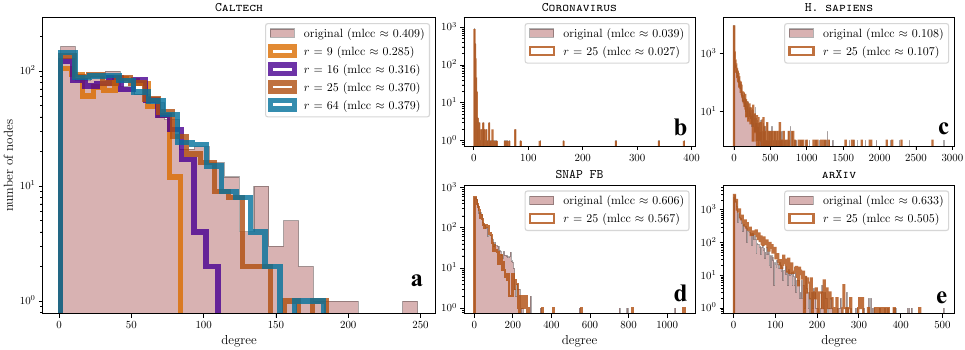}
			\caption{A comparison of the degree distributions (which we show as histograms) and the mean local clustering coefficients (which we write in the legends as `mlcc') of the original networks and the reconstructed networks using $r$ latent motifs at scale $k = 21$ {for the five networks in Figure~\ref{fig:network_recons}\textbf{a}. 
					In panel \textbf{a}, we use the four unweighted reconstructed networks for \dataset{Caltech} with $r \in \{9,16,25,64\}$ latent motifs that we used to compute the self-reconstruction accuracies in Figure~\ref{fig:network_recons}\textbf{b}. As we increase $r$, 
					the mean local clustering coefficient increase towards its value in the original networks and map{the} degree distributions of the reconstructed networks converge to that of the original network. 
					By increasing $r$, we are able to include nodes with progressively larger degrees in the latent motifs.
					In panels \textbf{b}--\textbf{e}, we show the mean local clustering coefficients and degree distributions of the unweighted reconstructed networks for \dataset{Coronavirus}, \dataset{H. sapiens}, \dataset{SNAP FB}, and \dataset{arXiv} with $r = 25$ latent motifs that we used to compute the self-reconstruction accuracies in Figures~\ref{fig:network_recons}\textbf{b}--\textbf{e}.}
			}
			\label{fig:deg_dist_recons_plot}
		\end{figure*}

		We gain further insights into our self-reconstruction experiments by comparing the degree distributions and the mean local clustering coefficients of the original and the unweighted reconstructed networks with threshold $\theta=0.4$ (see Figure~\ref{fig:deg_dist_recons_plot}). The mean local clustering coefficients of all reconstructed networks are similar to those of the corresponding original networks. In Figure~\ref{fig:deg_dist_recons_plot}\textbf{a}, we show that the degree distributions of the reconstructed networks for \dataset{Caltech} with $r\in \{9,16,25,64\}$ latent motifs converge toward that of the original network as we increase $r$. Reconstructing \dataset{Caltech} with larger values of $r$ appears to increase accuracy by including nodes with a larger degree than is possible for smaller values of $r$. In other words, `low-rank'  reconstructions (i.e., those with small values of $r$) of \dataset{Caltech} seem to recover only a small number of the edges of each node, even though it is able to achieve a large reconstruction accuracy (e.g., over $81$\% for $r = 9$).

		We now consider cross-reconstruction accuracies 
		{$X\leftarrow Y$}
		 in Figures~\ref{fig:network_recons}\textbf{b},\textbf{c}, where $Y$ is one of the Facebook networks \dataset{Caltech}, \dataset{Harvard}, \dataset{MIT}, and \dataset{UCLA} and $X$ (with $X \neq Y$) is one of the these four networks or one of the four synthetic networks \dataset{ER$_2$}, \dataset{WS$_2$}, \dataset{BA$_2$}, and \dataset{SBM}$_{2}$. From the cross-reconstruction accuracies and examining the network structures of the the latent motifs (see {Appendix}~\ref{subsection:NDL_formulation} of the SI) in Figures~\ref{fig:img_ntwk_dict} and~\ref{fig:latent_motifs_multiscale_1} (also see Figures~\ref{fig:all_dictionaries_1},~\ref{fig:all_dictionaries_3}, and~\ref{fig:all_dictionaries_4} in the SI), we draw a few conclusions at scale $k = 21$. First, the mesoscale {structure} of \dataset{Caltech} is distinct from those of \dataset{Harvard}, \dataset{UCLA}, and \dataset{MIT}. This is consistent with prior studies of these networks (see, e.g.,~\cite{traud2012social,jeub2015think}). Second, \dataset{Caltech}'s mesoscale structures at scale $k = 21$ are higher-dimensional
		than those of the other three universities' Facebook networks. 	 
		Third, \dataset{Caltech} has a lot more communities with at least $10$ nodes than the other three universities' Facebook networks (also see Figure~\ref{fig:Figure_boxcompare}). 
		Fourth, the BA network \dataset{BA$_2$} captures the mesoscale {structure} of \dataset{MIT}, \dataset{Harvard}, and \dataset{UCLA} at scale $k = 21$ better than the synthetic networks that we generate from the ER, WS, and SBM models. However, for all $r \in \{ 9,15,25,49\}$, the network \dataset{SBM}$_{2}$ captures the mesoscale structures of \dataset{Caltech} better than all other networks in Figure~\ref{fig:network_recons}\textbf{b} except for \dataset{Caltech} itself.
		See Appendix~\ref{subsection:Fig3_SI} of the SI for further discussion.  
		

		\begin{figure*}[h]
			\centering
			\includegraphics[width=1 \linewidth]{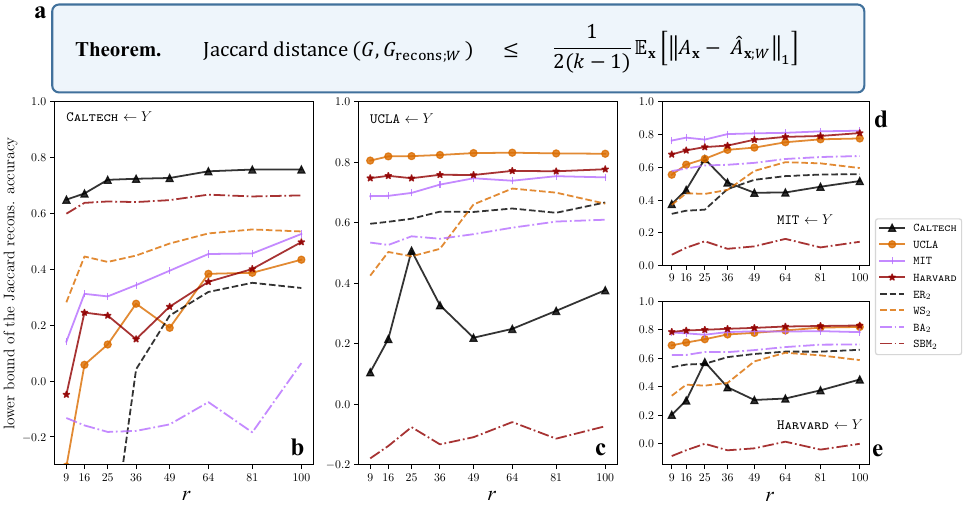}
			\caption{ 
					We prove that the Jaccard reconstruction error of the weighted reconstructed networks (i.e., without thresholding edge weights as in Figure~\ref{fig:network_recons}) is upper-bounded by the mean approximation error of {the} $k$-node subgraphs by $k$-path latent motifs divided by $2(k-1)$. {In (\textbf{a}), we state this result concisely.} Given a network $G$ and a network dictionary $W$ of $k$-path latent motifs, let $G_{\textup{recons};W}$ denote the weighted reconstructed network that we obtain using our NDR algorithm{. (See Algorithm~\ref{alg:network_reconstruction} in the SI.)} We define the Jaccard distance in \eqref{eq:JD} in the SI. On the right-hand side of the inequality in \textbf{a}, 
					the subgraph $\x$ is a uniformly random $k$-path of $G$, the matrix $A_{\x}$ is the $k\times k$ adjacency matrix of the subgraph that is induced by the node set of $\x$, and $\hat{A}_{\x;W}$ is the best nonnegative linear approximation of 
					$A_{\x}$ that we obtain using the latent motifs in the network dictionary $W$.
					See {Appendix}~\ref{sec:SI_NDR} and Theorem~\ref{thm:NDL} in the SI for precise statements of relevant definitions and the mathematical result.
					Subtracting both sides of the  
					inequality in \textbf{a} from $1$ yields a lower bound for the  
					Jaccard reconstruction accuracy. This quantity also measures the accuracy of reconstructing mesoscale patches of $G$ using latent motifs in $W$. {In (\textbf{b})--(\textbf{e}), we plot this lower bound} for the same parameters as in Figures~\ref{fig:network_recons}\textbf{b}--\textbf{e} (i.e., $k = 21$ and $r \in \{6,16,25,36,49,64,81,100 \}$).
			}
			\label{fig:recons_bd_plot}
		\end{figure*}

			We also comment briefly about the cross-reconstruction experiments in Figure~\ref{fig:network_recons} that use latent motifs that we learn from ER networks. For instance, when reconstructing \dataset{MIT}, \dataset{Harvard}, and \dataset{UCLA} using latent motifs that we learn from \dataset{ER$_2$}, we obtain a reconstruction accuracy of at least 72\%.
			This may seem unreasonable at first glance because the latent motifs that we learn from \dataset{ER$_2$} should not have any information about the Facebook networks. However, all of these networks are sparse (with edge densities of at most $0.02$) and we are sampling subgraphs using $k$-paths. The $k$-node subgraphs that are induced by uniformly random $k$-paths in these sparse networks have only a few off-chain edges (see Figure~\ref{fig:subgraphs}). For example, the $k$-node subgraphs that we sample from the sparse ER network \dataset{ER}$_{2}$ tend to have $k$-node paths and a few extra off-chain edges (see Figure~\ref{fig:subgraphs}). A similarly sparse or sparser network, such as \dataset{UCLA} (whose edge density is about $0.0036$) has similar subgraph patterns (despite the fact that, unlike the subgraphs of \dataset{ER}$_{2}$, the off-chain edges are not independent). This is the reason that we can reconstruct some networks with high accuracy by using latent motifs that we learn from a completely unrelated network.

			One learns latent motifs by maximizing the accuracy of reconstructions of mesoscale patches using them,
			rather than
			by maximizing the
			network-reconstruction accuracy. 
			In Figure \ref{fig:recons_bd_plot}\textbf{b}--\textbf{e}, we 	illustrate that self-reconstruction of mesoscale patches is more accurate than cross-reconstructions of mesoscale patches. However, because network reconstruction involves taking the mean of the reconstructed weights of an edge
			from multiple mesoscale patches that include that edge, an accurate reconstruction of mesoscale patches need not always entail accurate network reconstruction.
			In Figure \ref{fig:network_recons}, we see that  the self-reconstruction $X\leftarrow X$ is more accurate than the cross-reconstructions $X\leftarrow Y$ for $Y\ne X$ for almost all choices of networks $X$ and $Y$ and the parameter $r$. The two exceptions are $(X,Y,r) = (\dataset{Harvard}, \dataset{MIT}, 25)$ and $(X,Y,r) = (\dataset{Harvard}, \dataset{Caltech}, 25)$,
			although the cross-reconstruction accuracies in these cases are at most $2\%$ larger than the self-reconstruction accuracy.

			The above discussion suggests an important question: 
			If a network dictionary is effective at approximating the mesoscale patches of a network, what reconstruction accuracy does one expect? In the present paper, we state and prove a novel theorem that answers this question. Specifically, we prove mathematically that a {Jaccard reconstruction error} of a weighted reconstructed network (i.e., without thresholding edge weights as in Figure~\ref{fig:network_recons}) is upper-bounded by the mean approximation error of the mesoscale patches (i.e., $k$-node subgraphs) of a network by the $k$-path latent motifs divided by $2(k - 1)$. 
			See {Appendix}~\ref{sec:SI_NDR} and Theorem~\ref{thm:NDL} in the SI for precise statements of the relevant definitions and the mathematical result.

	

	Our NDL algorithm (see Algorithm~\ref{alg:NDL}) finds a network dictionary that approximately minimizes the upper bound in the inequality in Figure~\ref{fig:recons_bd_plot}\textbf{a}. {(See Theorem~\ref{thm:NDL} in the SI.)} Using an arbitrary network dictionary is likely to yield larger values of the upper bound. Therefore, according to the theorem in Figure~\ref{fig:recons_bd_plot}\textbf{a}, it is likely to yield a less accurate weighted reconstructed network. For instance, {for} a network dictionary that consists of a single $k$-path, the aforementioned upper bound is the mean number of off-chain edges in {the} mesoscale patches of a network divided by $k - 1$. 
	(See the inequality in Figure~\ref{fig:recons_bd_plot}\textbf{a}.) For an ER network with expected edge density $p$, this upper bound
	equals $kp/2$ in expectation. At scale $k = 20$ for \dataset{ER}$_{2}$, this value is $0.5$. Consequently, we expect a reconstruction accuracy of at least $50\%$ when reconstructing \dataset{ER}$_{2}$ using latent motifs (such as the ones for \dataset{UCLA} in Figure~\ref{fig:img_ntwk_dict}) that have large on-chain entries and small off-chain entries. Substituting the edge density of \dataset{UCLA} for $p$, we expect a reconstruction accuracy of at least 94\%. However, according to the first inequality in Figure~\ref{fig:recons_bd_plot}\textbf{a}, the lower bound of the reconstruction accuracy that we obtain using latent motifs of \dataset{UCLA} is about 80\%. Therefore, for \dataset{UCLA}, we expect to obtain many more off-chain edges in mesoscale patches than what we expect from an ER network with the same edge density. We plot the lower bound of the reconstruction accuracy in Figures~\ref{fig:recons_bd_plot}\textbf{b}--\textbf{e} for the {parameters in Figures~\ref{fig:network_recons}\textbf{b}--\textbf{e} (i.e., $k = 21$ and $r \in  \{6,16,25,36,49,64,81,100 \}$).} The lower bounds for the self-reconstructions is not too far from the actual reconstruction accuracies for unweighted reconstructed networks in Figures~\ref{fig:network_recons}\textbf{b}--\textbf{e} (it is within $20$\% for \dataset{Caltech} and \dataset{UCLA} and within $10$\% for \dataset{MIT} and \dataset{Harvard} for all $r$), but we {observe}
	much larger accuracy gaps for the cross-reconstruction experiments. (For example, there is at least a $50$\% difference for \dataset{UCLA}$\leftarrow$\dataset{Caltech}.) This indicates that, even if one uses latent motifs that are not very efficient {at approximating mesoscale patches}, one can obtain unweighted 	{reconstructions} that are significantly more accurate than what is guaranteed by the theoretically proven bounds.


		\subsection*{Network-denoising experiments}

		We consider the following `network-denoising' problem (which is closely related to the anomalous-subgraph-detection problem in Figure \ref{fig:anomaly_detection}). Suppose that we are given an observed network $G_{\textup{obs}}=(V,E_{\textup{obs}})$ with a node set $V$ and edge set $E_{\textup{obs}}$ and that we are asked to find an unknown 
		{network} $G_{\textup{true}}=(V,E_{\textup{true}})$ with the same node set $V$ but a possibly different edge set $E_{\textup{true}}$. We interpret $G_{\textup{obs}}$ as a corrupted version of a `true network' $G_{\textup{true}}$ that we observe with some uncertainty. To simplify the setting, we consider two types of network denoising. In the first type of network denoising, we consider \textit{additive noise~\cite{noble2003graph, parikh2008understanding, miller2015spectral, correia2019handling}. We} suppose that $G_{\textup{obs}}$ is a corrupted version of $G_{\textup{true}}$ that includes false edges (i.e., $E_{\textup{obs}}\supseteq E_{\textup{true}}$), and we seek to classify all edges in $G_{\textup{obs}}$ into `positives' (i.e., edges in $G_{\textup{true}}$) and `negatives' (i.e., {false edges in $G_{\textup{obs}}$ or equivalently} nonedges in $G_{\textup{true}}$). 
		 	 This network-denoising setting is identical to the anomalous-subgraph-detection problem in Figure \ref{fig:anomaly_detection}, except that now we label the false edges 
		 as 
		 negatives. 
		 We interpreted them 
		  as 
		  positives 
		  when we computed the F-score (i.e., the harmonic mean of the precision and recall scores) in Figure \ref{fig:anomaly_detection}\textbf{f}.
		  In the second type of network denoising, we consider \textit{subtractive noise} (which is often called edge `prediction'~\cite{zhou2021, liben2007link, menon2011link, kovacs2019network, guimera2020one}). We assume that $G_{\textup{obs}}$ is a partially observed version of $G_{\textup{true}}$ (i.e., $E_{\textup{obs}}\subsetneq E_{\textup{true}}$), and we seek to classify nonedges in $G_{\textup{obs}}$ into {positives} (i.e., nonedges in $G_{\textup{true}}$) and {negatives} (i.e., edges in $G_{\textup{true}}$). There are many more {positives} than {negatives} 
		  because 
		  $G_{\textup{true}}$ is sparse (i.e., the edge density is low), so we restrict the classification task to a subset $E_{\textup{nonedge}}$ {of $E_{\textup{obs}}$} that includes all negatives  
		  and an equal number of positives.
		  We will discuss shortly 
		  how we choose $E_{\textup{nonedge}}$.
	
		Given a true network $G_{\textup{true}} = (V, E_{\textup{true}})$, we generate an observed {(i.e., corrupted)} network $G_{\textup{obs}} = (V, E_{\textup{obs}})$ as follows.  In the additive-noise setting, we create two types of corrupted networks. We create the first type of corrupted network by adding false edges in a structured way by generating them using the WS model. We select 100 nodes for four of the networks (the exception is that we use 500 nodes for \dataset{H. sapiens}) uniformly at random and generate 1,000 new edges (we generate 30,000 new edges for \dataset{H. sapiens}) according to the WS model. In this corrupting WS network, each node in a ring of 100 nodes is adjacent to its 20 nearest neighbors and we uniformly randomly choose 30\% of the edges to rewire. When rewiring an edge, we choose between its two ends with equal probability of each, and we attach this end to a node in the network that we choose uniformly at random. We then add these newly generated edges to the original network. We refer to this noise type as `${+}\textup{WS}$'. We create the second type of corrupted network by first choosing 5\% of the nodes uniformly at random and adding an edge between each pair of chosen nodes with independent probability $0.3$. We refer to this noise type as `${+}\textup{ER}$'.} {In the subtractive-noise setting, we obtain $G_{\textup{obs}}$ from $G_{\textup{true}}$ by removing half of the existing edges, which we choose uniformly at random, such that the remaining network is connected. We refer to this noise type as `$-\textup{ER}$'. 
	
				For each observed network $G_{\textup{obs}}$, we apply NDL {at scale $k = 21$ with $r \in \{2, 25\}$} to learn a network dictionary $W_{\textup{obs}}$. We construct another network dictionary $\bar{W}_{\textup{obs}}$ by removing the on-chain edges from all of the latent motifs in $W_{\textup{obs}}$. (See the `Methods' section for further discussion.) This gives a total of four network dictionaries, corresponding to the two values of $r$ and whether or not we keep the on-chain edges of the latent motifs. With each of the network dictionaries, use NDR to reconstruct a network $G_{\textup{recons}}$ by approximating mesoscale patches of $G_{\textup{obs}}$ using latent motifs in $W_{\textup{obs}}$. (We compute $G_{\textup{recons}}$ without using any information on $G_{\textup{true}}$.) We anticipate that the reconstructed network $G_{\textup{recons}}$ is similar to its corresponding original (i.e., uncorrupted) network $G_{\textup{true}}$. The reconstruction algorithms output a weighted network $G_{\textup{recons}}$, where the weight of each edge is our confidence that the edge is a true edge of that network. For denoising subtractive (respectively, additive) noise, we classify each nonedge (respectively, each edge) in a corrupted network as `positive' if its weight in $G_{\textup{recons}}$ is strictly larger than some threshold $\theta$ and as `negative' otherwise. By varying $\theta$, we construct a receiver-operating characteristic (ROC) curve that consists of points whose horizontal and vertical coordinates are the false-positive rates and true-positive rates, respectively. For denoising the ${-}\textup{ER}$ (respectively, ${+}\textup{ER}$ and ${+}\textup{WS}$) noise, one can also infer an optimal value of $\theta$ for a 50\% training set of nonedges (respectively, edges) of $G$ with known labels and then use this value of $\theta$ to compute classification measures such as accuracy and precision.

			\begin{figure*}[h]
				\centering
				\includegraphics[width=1 \linewidth]{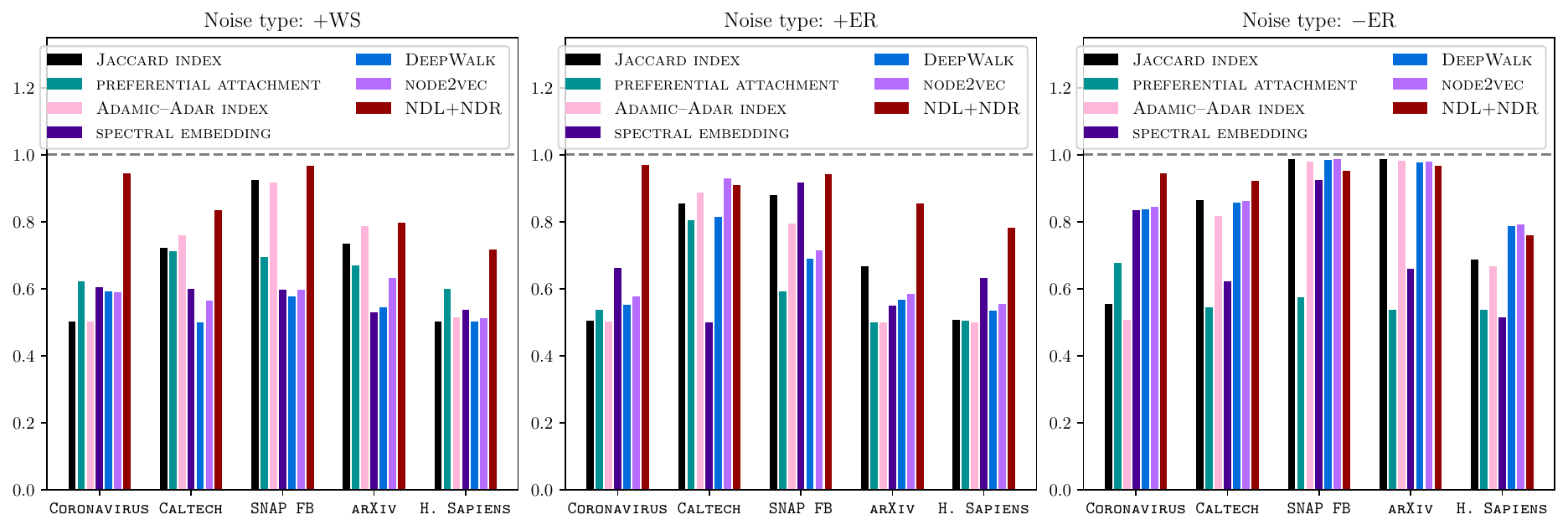}
				\caption{Applications of our NDL and NDR algorithms to network denoising with additive and subtractive noise on a variety of real-world networks. In our experiments with subtractive noise, we corrupt a network by removing $50$\% of its edges uniformly at random. We seek to classify 
					the nonedges in the corrupted network as true edges (i.e., removed edges) and false edges (i.e., nonedges in the original network), respectively. In our experiments with additive noise, we corrupt a network by uniformly randomly adding $50$\% of the number of its edges (i.e., 1,000 random edges) for all but one network (we add 30,000 random edges for \dataset{H. sapiens}) that we generate using the WS model. We seek to classify the edges in the resulting corrupted network as true edges (i.e., original edges) and false edges (i.e., added edges).
					To perform classification in a network, we first use NDL to learn latent motifs from a corrupted network and then reconstruct the networks using NDR to assign a confidence value to each potential edge. We then use these confidence values to infer the correct labeling of potential edges in the uncorrupted network. 
					Importantly, we never use information from the original networks to denoise the corrupted networks.
					For each network, we report the areas under the curves (AUCs) of the receiver-operating character (ROC) curves, which plot false-positive rates on the horizontal axis and true-positive rates on the vertical axis.
					See Figures~\ref{fig:Figure4__PRC_REC}, \ref{fig:Figure_PR_curve}, and \ref{fig:Figure_PR_curve_flipped} in the SI for the values of other binary-classification measures. 
				}
				\label{fig:Figure4}
			\end{figure*}

				In Figure~\ref{fig:Figure4}, we compare the performance of our network-denoising approach to the performance of several existing approaches using the real-world networks \dataset{Caltech}, \dataset{SNAP FB}, \dataset{arXiv}, \dataset{Coronavirus}, and \dataset{H. sapiens}. We use four classical approaches (the {\sc Jaccard index}, {\sc preferential attachment}, the {\sc Adamic--Adar index}, and a {\sc spectral embedding})~\cite{liben2007link, hasan2011survey} and two more recent methods ({\sc DeepWalk}~\cite{perozzi2014deepwalk} and {\sc node2vec}~\cite{grover2016node2vec}) that are based on network embeddings. Let $N(x)$ denote the set of neighbors of node $x$ of a network. For the {\sc Jaccard index}, {\sc preferential attachment}, and the {\sc Adamic--Adar index}, the confidence score (which plays the same role as an edge weight in a reconstructed network) that the nodes $x$ and $y$ are adjacent via a true edge is $|N(x)\cap N(y)|/ | N(x)\cup N(y)|$, $ |N(x)| \cdot |N(y)| $, and $\sum_{z\in N(x)\cap N(y)} 1 / \ln | N(z) |$, respectively. 

				{We now discuss how we choose the set $E_{\textup{nonedge}}$ of nonedges of $G_{\textup{obs}}$ for our subtractive-noise experiments.	First, we note that it is unlikely that many deleted edges in $E_{\textup{deleted}}$ are between two small-degree nodes.
				If we simply choose $E_{\textup{nonedge}}$ as a uniformly random subset of the set of all nonedges of $G_{\textup{obs}}$ with a given size $|E_{\textup{deleted}}|$, then it is likely that we will choose many nonedges between small-degree nodes. Consequently, the resulting classification problem is easy
				for existing methods, such as the {\sc Jaccard index} and {\sc preferential attachment}, that are based on node degrees. (For example, consider a star network with five leaves (i.e., degree-1 nodes). In this network, a uniformly randomly chosen nonedge is always attached to two degree-1 nodes, but a uniformly randomly chosen edge is always attached to one degree-5 node (i.e., the center node) and one degree-1 node.) 
				To reduce the size-biasing of node degrees, we choose each nonedge of $E_{\textup{nonedge}}$ with a probability that is proportional to the product of the degrees of the two associated nodes.}

				We show results in the form of means of the areas under the curves (AUCs) of the ROC curves
				 for five independent runs of each approach.
				In Figure \ref{fig:Figure4}, we see that our approach performs competitively in all of our experiments,
				particularly for denoising additive noise (i.e., anomalous-subgraph detection). For example, when we add 1,000 false edges that we generate from the WS model to \dataset{Coronavirus} (which has 2,463 true edges), our approach yields an AUC of $0.94$. We obtain the second best AUC (of only $0.61$) using {\sc preferential attachment}. For noise of type $+\textup{ER}$, we add $804$ false edges to \dataset{Coronavirus}; our approach achieves the best AUC (of $0.97$) and {\sc spectral embedding} achieves the second best AUC (of $0.66$).
				
{In Figure~\ref{fig:network_recons}\textbf{a}}, we saw that we can use a small number of latent motifs to reconstruct the social and PPI networks that we use for our denoising experiments in Figure~\ref{fig:Figure4}. Because NDL learns a small number of latent motifs that are able to successfully give an approximate basis for all mesoscale patches, they should not be affected significantly by false edges between nodes in a small subset of the entire node set. Consequently, the latent motifs in $W_{\textup{obs}}$ that we learn from the observed network $G_{\textup{obs}}$ may still be effective at approximating mesoscale patches of the true network $G_{\textup{true}}$, so the  network $G_{\textup{recons}}$ that we reconstruct using $G_{\textup{obs}}$ and $W_{\textup{obs}}$ may be similar to $G_{\textup{true}}$.

			
			
			\subsection*{Conclusions and outlook}
			
			We introduced a mesoscale network structure, which we call \textit{latent motifs}, that consists of $k$-node subgraphs that are building blocks of all connected $k$-node subgraphs of a network. In contrast to ordinary motifs~\cite{milo2002network}, which refer to
			overrepresented $k$-node subgraphs (especially for small $k$) of a network, nonnegative linear combinations of our latent motifs approximate $k$-node subgraphs that are induced by uniformly random $k$-paths in a network. We also established algorithmically and theoretically that one can approximate a network accurately if one has a dictionary of latent motifs that can accurately approximate mesoscale structures in the network.
			
			Our computational experiments in Figures~\ref{fig:img_ntwk_dict} 
			and~\ref{fig:latent_motifs_multiscale_1} {demonstrated} 
			that latent motifs can have distinctive network structures.
			 Our computational experiments in Figures~\ref{fig:network_recons},~\ref{fig:denoising_caltech_comparison}, and~\ref{fig:Figure4} illustrated that various social, collaboration, and PPI networks have low-rank~\cite{markovsky2012low} mesoscale structures, in the sense that a few latent motifs (e.g., $r = 25$ of them, but see Figure~\ref{fig:network_recons} for other choices of $r$) that we learn using NDL are able to reconstruct, infer, and denoise the edges of a network using our NDR algorithm.
			We hypothesize that such low-rank mesoscale structures are a common feature of networks beyond the social, collaboration, and PPI networks that we examined. As we have illustrated in this paper, one can leverage mesoscale structures to perform important tasks like network denoising, so it is important in future studies to explore the level of generality of our insights.

			In our work, we examined latent motifs in ordinary graphs.
			However, notions of motifs have been developed for several more general types of networks, including temporal networks (in which nodes, edges, and edge weights can change with time) \cite{paranjape2017motifs} and multilayer networks (in which, e.g., nodes can be adjacent via multiple types of relationships) \cite{battiston2017multilayer}.
			We did not consider latent motifs in such
			network structures, and it is worthwhile to extend our approach and algorithms to these situations.


			
			\subsection*{Limitations and further discussion}
			
			In the next few paragraphs, we briefly discuss several salient points about our work.

			First, it is possible for two sets of latent motifs to be equally effective at reconstructing the same network. Therefore, although one can interpret the structures 
			in latent motifs 
			as mesoscale structures of a network, one cannot conclude that other mesoscale structures (which not in a given set of latent motifs) do not also occur in the network.

			Second, our NDL algorithm approximately computes a `best' network dictionary to reconstruct the mesoscale patches of a network, rather than one to reconstruct the network itself.
			Although our theoretical bound on the reconstruction error (see Figure \ref{fig:recons_bd_plot}\textbf{a}) implies that such a network dictionary should also be effective at reconstructing a network, it is still necessary to empirically verify the actual efficacy of doing so.

			Third, the same theoretical bound on the reconstruction error illustrates that it is possible to successfully reconstruct a very sparse network using latent motifs that one learns from a radically different but similarly sparse network at a given scale $k$ (see Figure \ref{fig:network_recons}\textbf{e}). To better distinguish distinct 
			sparse networks from each other, one can use a scale $k$ that is large enough so that $k$-node mesoscale patches have many off-chain edges and latent motifs at that scale are sufficiently different in different networks. For example, see the latent motifs at scale $k = 51$ in Figures~\ref{fig:all_dictionaries_1}--\ref{fig:all_dictionaries_4} in the SI. Naturally, using a larger scale $k$ increases the computational cost of our approach.

			Fourth, although our method for network denoising is competitive --- especially for the anomalous-subgraph-detection problem --- it does not always outperform all existing methods, and some of those methods are much simpler than ours. For instance, for edge-prediction tasks, it seems that our method is often more conservative than the other examined methods at detecting unobserved edges. 
			 (See Figures~\ref{fig:Figure4__PRC_REC}, \ref{fig:Figure_PR_curve}, and \ref{fig:Figure_PR_curve_flipped} in the SI.) Therefore, we recommend using our method in conjunction with 
			 existing methods for such tasks.

			
			\section*{Acknowledgements}
			
			HL was supported by the National Science Foundation through grants 2206296 and 2010035. JV was supported by the National Science Foundation through grant 1740325.
		

			\vspace{0.3cm}

			\pagestyle{myheadings}
			\markright{\uppercase{Learning low-rank latent mesoscale structures in networks}}
			\markleft{\uppercase{Learning low-rank latent mesoscale structures in networks}}
			
			\printbibliography[keyword = {main}]
			
			
			\section*{Methods}\label{sec:methods}

			{We briefly discuss our algorithms for network dictionary learning (NDL) and network denoising and reconstruction (NDR). We also provide a detailed description of our real-world and synthetic networks.
				
				We restrict our present discussion to networks that one can represent as a graph $G=(V,E)$ with a node set $V$ and an edge set $E$ without directed edges or multi-edges (but possibly with {self-edges}).
				 In the SI, we give an extended discussion that applies to more general types of networks. Specifically, in that discussion, we no longer restrict edges to have binary weights; instead, the weights {can have continuous nonnegative values}. 
				 See {Appendix}~\ref{subsection:Definition} of the SI.
				
				
				\subsection*{Motif sampling and mesoscale patches of networks}
				
				The connected $k$-node subgraphs of a network are natural candidates for the network's mesoscale patches.
				These subgraphs have $k$ nodes that inherit their adjacency structures from the original networks from which we obtain them. It is convenient to consider the $k \times k$ adjacency matrices of these subgraphs, as it allows us to perform computations on the space of subgraphs. However, to do this, we need to address two issues. First, because the same $k$-node subgraph can have multiple (specifically, $k!$) different representations as an adjacency matrix (depending on the ordering of its nodes), we need an unambiguous way to choose an ordering of its nodes. Second, because most real-world networks are sparse~\cite{newman2018}, independently choosing a set of $k$ nodes from a network may yield only a few edges and may thus often result in a disconnected subgraph. Therefore, we need an efficient sampling algorithm to guarantee that we obtain connected $k$-node subgraphs when we sample from sparse networks.}
						
			We employ an approach that is based on \textit{motif sampling}~\cite{lyu2023sampling} both to choose an ordering of the nodes of a $k$-node subgraph and to ensure that we sample connected subgraphs from sparse networks. The key idea is to consider the random $k$-node subgraph that we obtain by sampling a copy of a `template' subgraph uniformly at random from a network. (We sample the nodes uniformly at random and include all of the network's edges between those sampled nodes.)
			We take such a template to be a `$k$-path'. A sequence $\x = (x_{1},\ldots,x_{k})$ of $k$ (not necessarily distinct) nodes is a \textit{$k$-walk} if $x_{i}$ and $x_{i+1}$ are adjacent for all $i \in \{1, \ldots, k-1\}$. A $k$-walk is a \textit{$k$-path} if all nodes in the walk are distinct (see Figure~\ref{fig:subgraphs}). For each $k$-path $\x=(x_{1},\ldots,x_{k})$, we define the corresponding mesoscale patch of a network to be the $k \times k$ matrix $A_{\x}$ such that $A_{\x}(i,j) = 1$ if nodes $x_{i}$ and $x_{j}$ are adjacent and $A_{\x}(i,j) = 0$ if they are not adjacent. This is the adjacency matrix of the $k$-node subgraph of the network with nodes $x_{1},\ldots,x_{k}$.   {One can use one of the} Markov-chain Monte Carlo (MCMC) algorithms for motif sampling from~\cite{lyu2023sampling} {to efficiently and uniformly randomly sample a} $k$-walk from a sparse network. By only accepting samples in which the $k$-walk has $k$ distinct nodes (i.e., so that it is $k$-path), we efficiently sample a uniformly random $k$-path from a network, as long as $k$ is not too large. If $k$ is too large, one has to `reject' too many samples of $k$-walks that are not $k$-paths. The expected number of rejected samples is approximately the number of $k$-walks divided by the number of $k$-paths. The number of $k$-walks in a network grows monotonically with $k$, but the number of $k$-paths can decrease with $k$. 
			(See {Appendix}~\ref{section:MCMC} of the SI for a detailed discussion.) Consequently, by repeatedly sampling $k$-{paths} $\x$, we obtain a data set of mesoscale patches $A_{\x}$ of a network.


			\subsection*{Algorithm for network dictionary learning (NDL)}
					
			We now present the basic structure of the algorithm that we employ for \textit{network dictionary learning} (NDL)~\cite{lyu2020online}. Suppose that we compute all possible $k$-paths $\x_{1},\dots,\x_{M}$ and their corresponding mesoscale patches $A_{\x_{t}}$ (which are $k \times k$ binary matrices), with $t \in \{1, \ldots, M\}$, of a network. We column-wise vectorize (i.e., we place the second column underneath the first column and so on; see Algorithm~\ref{alg:vectorize} in the SI) each of these $k \times k$ mesoscale patches to obtain a $k^{2} \times M$ data matrix $X$. We then apply nonnegative matrix factorization (NMF)~\cite{lee1999learning} to obtain a $k^{2} \times r$ nonnegative matrix $W$ for some fixed integer $r \ge 1$ to yield an approximate factorization $X \approx WH$ for some nonnegative matrix $H$. From this procedure, we approximate each column of $X$ by the nonnegative linear combination of the $r$ columns of $W$; its coefficients are the entries of the corresponding column of $H$. If we let $\mathcal{L}_{i}$ be the $k \times k$ matrix that we obtain by reshaping the $i^{\textup{th}}$ column of $W$ (using Algorithm~\ref{alg:reshape} in the SI), then $\mathcal{L}_{1},\ldots,\mathcal{L}_{r}$ are the learned latent motifs; they form a network dictionary. The set of these latent motifs is an approximate basis (but not a subset) of the set $\{ A_{\x_{1}},\dots,A_{\x_{M}} \}$ 
			of mesoscale patches. For instance, latent motifs have entries that take continuous values between $0$ and $1$, but mesoscale patches have binary entries. We can regard each $\mathcal{L}_{i}$ as the $k$-node weighted network with node set $\{1,\ldots,k\}$ and weighted adjacency matrix $\mathcal{L}_{i}$. See Figure~\ref{fig:img_ntwk_dict} for an illustration of latent motifs as weighted networks.

			The scheme in the paragraph above requires us to store all possible mesoscale patches of a network, entailing a memory requirement that is at least of order $k^{2}M$, where $M$ denotes the total number of mesoscale patches of a network. Because $M$ scales with the size (i.e., the number of nodes) of the network from which we sample subgraphs, we need unbounded memory to handle arbitrarily large networks. To address this issue, Algorithm~\ref{alg:NDL} implements the above scheme in the setting of `online learning', where subsets (so-called `minibatches') of data arrive in a sequential manner and one does not store previous subsets of the data before processing new subsets. Specifically, at each iteration $t \in \{1,2,\ldots,T\}$, we  process a sample matrix $X_{t}$ that is smaller than the full matrix $X$ and includes only $N \ll M$ mesoscale patches, where one can take $N$ to be independent of the network size. 
			Instead of using a standard NMF algorithm for a fixed matrix~\cite{lee2001algorithms}, we employ an `online' NMF algorithm~\cite{mairal2010online, lyu2020online} that one can use on sequences of matrices, where the intermediate dictionary matrices $W_{t}$ that we obtain by factoring the sample matrix $X_{t}$ typically improve as we iterate
			(see~\cite{mairal2010online, lyu2020online}).
			In Algorithm~\ref{alg:NDL} in the SI, we give a complete implementation of the NDL algorithm.	
			
			
			\subsection*{Algorithm for network denoising and reconstruction (NDR)}

				Suppose that we have an image $\gamma$ of size $k \times k$ pixels and a set of `basis images'  {$\beta_{1},\ldots,\beta_{r}$} of the same size. {We} can `reconstruct' the image {$\gamma$} using the basis images {$\beta_{1},\ldots,\beta_{r}$} by finding nonnegative coefficients $a_{1},\ldots,a_{r}$ such that the linear combination {$\hat{\gamma} = a_{1}\beta_{1} + \cdots + a_{r}\beta_{r}$} is as close as possible to {$\gamma$}. The basis images determine what shapes and colors of the original image to capture in the reconstruction {$\hat{\gamma}$}. In the standard pipeline for image denoising and reconstruction~\cite{elad2006image, mairal2008sparse, mairal2009non}, one assumes that the size $k \times k$ of the square `patches' is much smaller than the size of the full image $\gamma$. One can then sample a large number of $k \times k$ overlapping patches $\gamma_{1},\ldots,\gamma_{M}$ of the image $\gamma$ and obtain the best linear approximations $\hat{\gamma}_{1},\ldots,\hat{\gamma}_{M}$ of them using the basis images $\beta_{1},\ldots,\beta_{r}$. Because the $k \times k$ patches $\gamma_{1},\ldots,\gamma_{M}$ overlap, each pixel $(I,J)$ of $\gamma$ can occur in multiple instances of $\gamma_{1},\ldots,\gamma_{M}$. Therefore, we take the mean of the corresponding values in the mesoscale reconstructions $\hat{\gamma}_{1},\dots,\hat{\gamma}_{M}$ as the value of the pixel $(I,J)$ in the 
				{reconstruction} $\hat{\gamma}$.

			\begin{figure}[h]
				\centering
				\vspace{-0.3cm}
				\includegraphics[width=1 \linewidth]{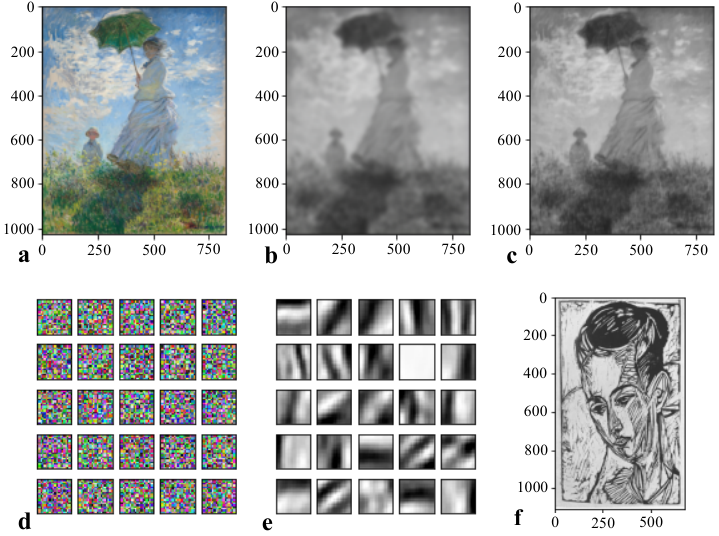}
				\vspace{-0.8cm}
				\caption{(\textbf{a}) The image \dataset{Woman with a Parasol - Madame Monet and Her Son} (Claude Monet, 1875). In (\textbf{b},\textbf{c}), we show reconstructions of the image.
				The image in \textbf{b} is a reconstruction of the image in \textbf{a} using the dictionary with 25 basis images of size $21 \times 21$ pixels in panel \textbf{d}. We uniformly randomly choose the color of each pixel from all possible colors (which we represent as vectors in $[0,256]^{3}$ for red--green--blue (RGB) weights). The image in \textbf{c} is a reconstruction of the image in panel \textbf{a} using the dictionary with 25 basis images of size $21 \times 21$ pixels in panel \textbf{e}. We learn this basis from the image in panel \textbf{f} using NMF. The image in \textbf{f} is from the collection \dataset{Die Graphik Ernst Ludwig Kirchners bis 1924, von Gustav Schiefler Band I bis 1916} (Accession Number 2007.141.9, Ernst Ludwig Kirchner, 1926). [We use the images in panels \textbf{a} and \textbf{f} with permission from the National Gallery of Art in Washington, DC, USA.]
				}
				\label{fig:img_recons_ex}
			\end{figure}

				As an illustration, we reconstruct the color image in Figure~\ref{fig:img_recons_ex}\textbf{a} in two ways, which yield the images in Figures~\ref{fig:img_recons_ex}\textbf{b},\textbf{c}. In Figure~\ref{fig:img_recons_ex}\textbf{d}, we show a dictionary with 25 basis images of size $21 \times 21$ pixels. We uniformly randomly choose the color of each pixel from all possible colors (which we represent as vectors in $[0,256]^{3}$ for red--green--blue (RGB) weights). The basis images do not include any information about the original image in Figure~\ref{fig:img_recons_ex}\textbf{a}, so the linear approximation of the $21 \times 21$ mesoscale patches of the image in
				Figure~\ref{fig:img_recons_ex}\textbf{a} using the basis images in Figure~\ref{fig:img_recons_ex}\textbf{d} may be inaccurate. However, when we reconstruct the entire image from Figure~\ref{fig:img_recons_ex}\textbf{a} using the basis images in Figure~\ref{fig:img_recons_ex}\textbf{d}, we do observe some basic geometric information from the original image. In Figure~\ref{fig:img_recons_ex}\textbf{b}, we show the image that results from this reconstruction.
				Importantly, the image reconstruction in Figure~\ref{fig:img_recons_ex}\textbf{b} uses both the basis images and the original image that one seeks to reconstruct. Unfortunately, the colors have averaged out to be neutral, so the reconstructed image is monochrome. Using a smaller (e.g., $5 \times 5$) randomly generated (with each pixel again taking an independently and uniformly chosen color) basis-image set for reconstruction results in a monochrome (but sharper) reconstructed image. 
				Notably, one can learn the basis images in Figure~\ref{fig:img_recons_ex}\textbf{e} from the image in Figure~\ref{fig:img_recons_ex}\textbf{f} using nonnegative matrix factorization (NMF)~\cite{lee1999learning}. The image in Figure~\ref{fig:img_recons_ex}\textbf{f} is in black and white, so the images in Figure~\ref{fig:img_recons_ex}\textbf{e} are also in black and white. The corresponding reconstruction in Figure~\ref{fig:img_recons_ex}\textbf{c} has nicely captured shapes from the original image, although we have lost the color information in the original image in Figure~\ref{fig:img_recons_ex}\textbf{a} and the reconstruction in Figure~\ref{fig:img_recons_ex}\textbf{c} is thus in black and white.

			A network analog of the above patch-based image reconstruction proceeds as follows. (See Figure \ref{fig:recons_illustration} of the main manuscript for an illustration.)
			Given a network $G = (V,E)$ and latent motifs $(\mathcal{L}_{1},\ldots,\mathcal{L}_{r})$ (which we do not necessarily compute from $G$; see Figure~\ref{fig:img_recons_ex}), we 
			obtain a weighted network $G_{\textup{recons}}$ using the same node set $V$ and 
			a weighted adjacency matrix $A_{\textup{recons}}: V^{2}\rightarrow \R$. To do this, {we first use the MCMC motif-sampling algorithm from~\cite{lyu2023sampling} with rejection sampling to} sample a large number $T$ of $k$-paths $\x_{1},\ldots,\x_{T}:\{1,\ldots,k\} \rightarrow V$ of $G$. (For details, see Algorithm \ref{alg:motif_inj} in the SI.) We then determine the corresponding mesoscale patches $A_{\x_{1}},\ldots,A_{\x_{T}}$ of $G$. We then approximate each mesoscale patch $A_{\x_{t}}$, which is a  $k\times k$ unweighted matrix, by a nonnegative linear combination $\hat{A}_{\x_{t}}$ of the latent motifs $\mathcal{L}_{i}$. We seek
			to `replace' each $A_{\x_{t}}$ by $\hat{A}_{\x_{t}}$ to construct the weighted adjacency matrix $A_{\textup{recons}}$. To do this, we define
			$A_{\textup{recons}}(x,y)$ for each $x,y\in V$ as the mean of $\hat{A}_{\x_{t}}(a,b)$ over all $t\in \{1,\ldots, T\}$ and all $a,b\in \{1,\ldots,k\}$ such that $\x_{t}(a)={x}$ and $\x_{t}(b)={y}$. We state this network-reconstruction algorithm precisely in Algorithm~\ref{alg:network_reconstruction}.
			See Appendix~\ref{subsection:NR} of the SI for more details.


			\subsection*{A comparison of our work with the prior research in~\cite{lyu2020online}}

				Recently, Lyu et al.~\cite{lyu2020online} proposed a preliminary approach for the algorithms that we study in the present work --- the NDL algorithm with $k$-walk sampling and the {NDR algorithm for network-reconstruction tasks} --- as an application to showcase a theoretical result about the convergence of online NMF for  
				data samples that are not independently and identically distributed (IID)~\cite[Thm. 1]{lyu2020online}. 
				A notable limitation of the NDL algorithm in~\cite {lyu2020online} is that one cannot interpret the elements of a network dictionary as latent motifs and one thus cannot associate them directly with mesoscale structures in a network.
				Additionally, Lyu et al.~\cite{lyu2020online} did not include any theoretical analysis of either the convergence or the correctness of network reconstruction, so it is unclear from that work whether or not one can reconstruct a network using the low-rank mesoscale {structures} that are encoded in a network dictionary. 
				Moreover, one cannot use the {NDR} algorithm that was proposed in~\cite{lyu2020online} to denoise additive noise unless one knows in advance that the noise is additive (rather than subtractive) before denoising (see~\cite[Rmk. 4.]{lyu2020online}). In the present paper, we build substantially on the research in~\cite{lyu2020online} and provide a much more complete computational and theoretical framework to analyze low-rank mesoscale structures in networks. In particular, we overcome all of the aforementioned limitations. In Table~\ref{table:NDL_comparison}, we summarize the key differences between the present work and \cite{lyu2020online}.

			\begin{table}[htbp]
				\centering
				\begin{tabular}{c|cccc}
					\thickhline 
					\textbf{NDL} & Sampling & Latent motifs & Convergence & Efficient MCMC \\ 
					\hline
					\rule{0pt}{1.1\normalbaselineskip} Lyu et al.~\cite{lyu2020online} & $k$-walks & \xmark & Non-bipartite networks & \xmark \\[0.2cm]
					\rule{0pt}{1.1\normalbaselineskip} The present work & $\begin{matrix} \textup{$k$-paths} \\ \textup{(Alg.~\ref{alg:motif_inj}) } \end{matrix}$& \checkmark & $\begin{matrix} \textup{Non-bipartite and} \\ \textup{{bipartite networks}}  \\ \textup{(Thm.~\ref{thm:NDL}, Thm.~\ref{thm:NDL2}) } \end{matrix}$ & $\begin{matrix} \checkmark\\ \textup{(Prop.~\ref{prop:approximate_pivot}) } \end{matrix}$ \\[0.1cm]
					\hline
					\thickhline 
				\end{tabular}%
				\\[0.2cm]
				\begin{tabular}{c|ccccc}
					\thickhline 
					\textbf{NDR} & Sampling & Reconstruction & Denoising & Convergence & Error bound \\ 
					\hline 
					\rule{0pt}{1.2\normalbaselineskip} Lyu et al.~\cite{lyu2020online} & $k$-walks & \checkmark & \xmark & \xmark & \xmark \\[0.2cm]
					The present work & $\begin{matrix} \textup{$k$-walks} \\ \textup{$k$-paths} \end{matrix}$  & \checkmark  & \checkmark & $\begin{matrix} \checkmark\\ \textup{ (Thm.~\ref{thm:NR}) } \end{matrix}$ & $\begin{matrix} \checkmark\\ \textup{ (Thm.~\ref{thm:NR}) } \end{matrix}$  \\ 
					\hline
					\thickhline 
				\end{tabular}%
				\caption{A comparison of the contributions of the present work with those of Lyu et al.~\cite{lyu2020online}. In the SI, we give the statements and proofs of the proposition and theorems in this table.}
				\label{table:NDL_comparison}
			\end{table}

		The most significant theoretical advance of the present paper concerns the relationship between the reconstruction error and the error from approximating mesoscale patches by latent motifs, with an explicit dependence on the number $k$ of nodes in subgraphs at the mesoscale.
		We state this result in Theorem~\ref{thm:NR}\textbf{(iii)} in the SI. Informally, Theorem~\ref{thm:NR}\textbf{(iii)} states that one can accurately reconstruct a network
				if one has a dictionary of latent motifs that can accurately approximate the mesoscale patches of a network. In Figure~\ref{fig:recons_bd_plot}, we illustrate this theoretical result with supporting experiments. A crucial part of our proof of Theorem~{\ref{thm:NR}\textbf{(iii)}} is that the sequence of weighted adjacency matrices
				of the reconstructed networks converges as the number of iterations that one uses for network reconstruction tends to infinity and that this limiting weighted adjacency matrix
				has an explicit formula. We state these results, which are also novel contributions of our paper, in Theorem~\ref{thm:NR}{\textbf{(i)}},\textbf{(ii)}.

				\begin{figure*}[h]
					\centering
					\includegraphics[width=1 \linewidth]{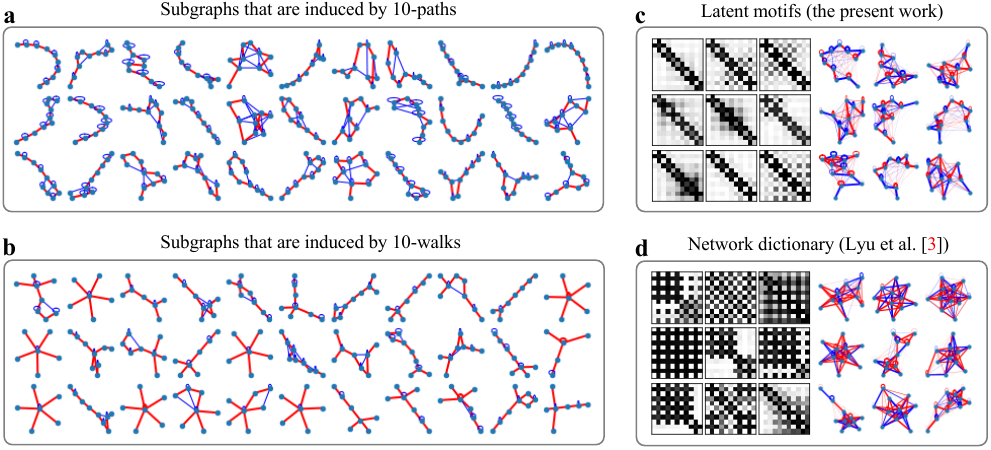}
					\caption{A comparison of subgraphs of \dataset{Coronavirus PPI} that are induced by node sets that we sample using (\textbf{a}) uniformly random $k$-paths and (\textbf{b}) uniformly random $k$-walks with $k = 10$. We also compare the network dictionary with $r = 9$ latent motifs of \dataset{Coronavirus PPI} that we determine using (\textbf{c}) the NDL algorithm (see Algorithm~\ref{alg:NDL}) in the present work to (\textbf{d}) the network dictionary that we determine using the NDL algorithm from~\cite{lyu2020online}. We also show the weighted adjacency matrices of the latent motifs. The $10$-walks in the network tend to visit the same nodes many times. Consequently, one cannot regard the $10 \times 10$ mesoscale patches that correspond to those walks as the adjacency matrices of $k$-node subgraphs of the network. Additionally, the networks in the network dictionary in \textbf{d} have clusters of several large-degree nodes, even though the original network does not possess such mesoscale structures.
					}
					\label{fig:covid_dict_comparison}
				\end{figure*}

				We now elaborate on the use of $k$-path sampling in our new NDL algorithm to ensure that one can interpret the network-dictionary elements as latent motifs. The NDL algorithm in~\cite{lyu2020online} used the $k$-walk motif-sampling algorithm of~\cite{lyu2023sampling}. That algorithm samples a sequence of $k$ nodes (which are not necessarily distinct) in which the $i$th node is adjacent to the $(i+1)$th node for all $i\in \{1,\ldots, k-1\}$. The $k$-walks that sample $k \times k$ subgraph adjacency matrices can have overlapping nodes, so some of the $k \times k$ adjacency matrices can correspond to subgraphs with fewer than $k$ nodes. If a network has a large number of such subgraphs, then the $k$-node latent motifs that one learns from the set of subgraph adjacency matrices can have misleading patterns that may not exist in any $k$-node subgraph of the network. This situation occurs in the network \dataset{CORONAVIRUS PPI}, where one obtains clusters of large-degree nodes from the learned latent motifs if one uses $k$-walk sampling. This misleading result arises from the $k$-walk visiting the same large-degree node many times, rather than because $k$ distinct nodes of the network actually have this type of subgraph patterns 
				(see Figure~\ref{fig:covid_dict_comparison}). 
				To resolve this issue,
				 during the dictionary-learning phase, we combine 
				MCMC $k$-walk sampling with rejection sampling so that we use only $k$-walks with $k$ distinct nodes (i.e., we use $k$-paths). Consequently, we now learn $k$-node latent motifs only from $k \times k$ adjacency matrices that correspond to $k$-node subgraphs of a network. This guarantees that any network structure (e.g., large-degree nodes, communities, and so on) in the latent motifs must also exist in the network at scale $k$.

			\begin{figure*}[h]
				\centering
				\includegraphics[width=1 \linewidth]{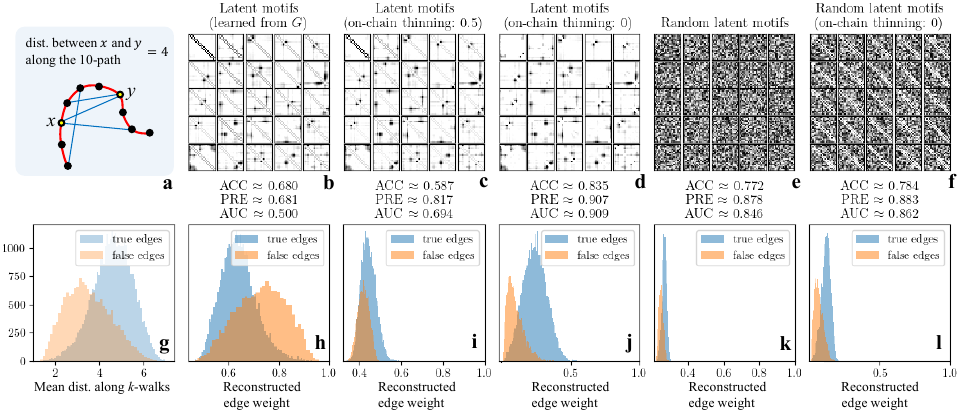}
				\caption{Denoising \dataset{Caltech} corrupted by $50\%$ additive noise of type ER. There are 16,656 true (i.e., original) edges and 7,854 false (i.e., added) edges to correctly classify. As we illustrate in (\textbf{a}), when a $k$-path connects two nodes $x$ and $y$, we define the distance between $x$ and $y$ along the $k$-path to be the shortest-path distance between $x$ and $y$. During reconstruction, we sample a sequence of $k$-paths in a network using a Markov-chain Monte Carlo (MCMC) algorithm (see Algorithm~\ref{alg:pivot} in the SI). Suppose that this sequence is $\x_{1},\ldots,\x_{T}$. We compute the mean of the distances between $x$ and $y$ along the $k$-path $\x_{t}$ for all $t\in \{1,\ldots, T\}$ such that $\x_{t}$ connects $x$ and $y$. In (\textbf{b})--(\textbf{f}), we show the weighted adjacency matrices of five sets of $21$-node latent motifs.
					We learn the 25 latent motifs in panel {\textbf{b}} from the corrupted network. We multiply the on-chain edge weights of these latent motifs by a thinning parameter $\xi = 0.5$ and $\xi = 0$ to obtain the matrices in panels {\textbf{c}} and {\textbf{d}}, respectively. We randomly choose the latent motifs in panel {\textbf{e}} by drawing each entry of its $k \times k$ weighted adjacency matrix independently and uniformly from $[0,1]$. Setting the on-chain entries {of these random latent motifs to $0$} gives the matrices in panel \textbf{f}. 
					In (\textbf{g}), we show histograms of the mean distances along the $k$-paths $\x_{1},\ldots,\x_{T}$ between the two ends of true edges and the two ends of false edges. In ({\textbf{h}})--(\textbf{l}), we show histograms of the edge weights for various reconstructions of the corrupted network using latent motifs in panels \textbf{b}--\textbf{f}.
					The classification accuracy (ACC) and precision (PRE) use the {best} threshold for truncating weighted edges in the reconstruction that we compute from a uniformly randomly chosen training set of edges (with $50\%$ of the edges of the observed network). 
					The AUC refers to the area under the ROC curve, which consists of points whose horizontal and vertical coordinates are the false-positive rates and true-positive rates, respectively.
				}
				\label{fig:denoising_caltech_comparison}
			\end{figure*}

			
			Based on our experiments, the network-reconstruction algorithm that was proposed in~\cite{lyu2020online} 
			seems to be effective at denoising subtractive noise. However, when denoising additive noise (see, e.g., the ${+}\textup{ER}$ and ${+}\textup{WS}$ results in Figure \ref{fig:Figure4} of the main manuscript), the edge weights in a reconstructed network can result in ROC curves with an AUC that is as small as 0.5. We demonstrate this issue more concretely in Figure
				 \ref{fig:denoising_caltech_comparison}. As we show in the histogram in Figure~\ref{fig:denoising_caltech_comparison}{\textbf{h}},
			 when we denoise \dataset{Caltech} with ${+}\textup{ER}$ noise, the false edges are assigned weights that are significantly larger than those of the true edges when we reconstruct the observed network using latent motifs that we learn from the corrupted network (see Figure~\ref{fig:denoising_caltech_comparison}{\textbf{b}}). 
			Consequently, we obtain an AUC of 0.5 for our classification. There is a simple explanation of this outcome. 
			A uniformly random $k$-path,
			which we use throughout the denoising process, tends to connect false edges using a smaller number of edges than it uses to connect true edges. In other words, if there is an edge between nodes $x$ and $y$ in an additively corrupted network and we uniformly randomly sample a $k$-path that uses both $x$ and $y$, then the number of edges between these two nodes along the sampled $k$-path tends to be small if the edge between $x$ and $y$ is false and tends to be large if it is true (see Figure~\ref{fig:denoising_caltech_comparison}{\textbf{g}}). 
			This indicates that there are not many ways to connect the two ends of a false edge using a $k$-path that avoids using that false edge. Consequently, of the edges between the nodes in a uniformly sampled $k$-path of an additively corrupted network, false edges are more likely to appear as on-chain edges than as off-chain edges. Consequently, as we observe in Figures~\ref{fig:denoising_caltech_comparison}\textbf{k},\textbf{l}, we can reasonably successfully denoise the network \dataset{Caltech} with additive noise of type {${+}\textup{ER}$} using randomized latent motifs in which each we draw each entry of their associated weighted adjacency matrices independently and uniformly from $[0,1]$.
			When we consider subtractive noise, an analogous observation holds for true nonedges and false nonedges.

			To ensure effective network denoising for both additive and subtractive noise, we modify our network-reconstruction algorithm
			by thinning out the on-chain edge weights of latent motifs in $W_{\textup{obs}}$ prior to network reconstruction. To do this, we multiply the weights of on-chain edges in the latent motifs and in all sampled mesoscale patches by a scalar chain-edge `thinning parameter' $\xi \in [0,1]$. For instance, the latent motifs in Figures~\ref{fig:denoising_caltech_comparison}\textbf{i},\textbf{j}
			use $\xi = 0.5$ and $\xi = 0$, respectively. As we see in the histograms in 
			Figures~\ref{fig:denoising_caltech_comparison}\textbf{c},\textbf{d},
			the negative edges in the resulting reconstruction have significantly smaller weights than the positive edges. With $\xi = 0$, for example, this results in a classification AUC of 0.91. 
			Although the thinning parameter $\xi$ can take any value in $[0,1]$, in all of our experiments except the one in Figure \ref{fig:denoising_caltech_comparison}, we use only the extreme values $\xi = 0$ and
			$\xi = 1$. It seems to be unnecessary to use
			values of $\xi$ in $(0,1)$.

			
			
			\subsection*{Data sets}\label{subsection:Datasets}

			We use the following eight real-world networks:
			\begin{enumerate}[itemsep=0.1cm]
				\item \dataset{Caltech}: This connected network, which is part of the {\sc Facebook100} data set~\cite{traud2012social} (and which was studied previously as part of the {\sc Facebook5} data set~\cite{Traud2011}), has 762 nodes and 16,651 edges. The nodes represent user accounts in the Facebook network of Caltech on one day in fall 2005, and the edges encode Facebook `friendships' between these accounts. 
				
				\item \dataset{MIT}: This connected network, which is part of the {\sc Facebook100} data set~\cite{traud2012social}, has 6,402 nodes and 251,230 edges. The nodes represent user accounts in the Facebook network of MIT on one day in fall 2005, and the edges encode Facebook `friendships' between these accounts.
				
				\item \dataset{UCLA}: This connected network, which is part of the {\sc Facebook100} data set~\cite{traud2012social}, has 20,453 nodes and 747,604 edges. The nodes represent user accounts in the Facebook network of UCLA on one day in fall 2005, and the edges encode Facebook `friendships' between these accounts. 
				
				\item \dataset{Harvard}: This connected network, which is part of the {\sc Facebook100} data set~\cite{traud2012social}, has 15,086 nodes and 824,595 edges. The nodes represent user accounts in the Facebook network of Harvard on one day in fall 2005, and the edges represent Facebook `friendships' between these accounts.
				
				\item \dataset{SNAP Facebook} (with the shorthand \dataset{SNAP FB})~\cite{leskovec2012learning}: This connected network has 4,039 nodes and 88,234 edges. This network is a Facebook network that has been used as an example in a study of edge inference~\cite{grover2016node2vec}. The nodes represent user accounts in the Facebook network on one day in 2012, and the edges represent Facebook `friendships' between these accounts.
				
				\item \dataset{arXiv ASTRO-PH} (with the shorthand \dataset{arXiv})~\cite{Leskovec2014SNAP,grover2016node2vec}: This network has 18,722 nodes and 198,110 edges. Its largest connected component has 17,903 nodes and 197,031 edges. We use the complete network in our experiments. This network is a collaboration network between authors of astrophysics papers that were posted to the arXiv preprint server. The nodes represent scientists and the edges indicate coauthorship relationships. This network has 60 self-edges; these edges encode single-author papers.
				
				\item \dataset{Coronavirus PPI} (with the shorthand \dataset{Coronavirus}): This connected network is curated by \url{theBiogrid.org}~\cite{oughtred2019biogrid, oughtred2019biogrid_covid, gordon2020sars} from 142 publications and preprints. It has 1,536 proteins that are related to coronaviruses and 2,463 protein--protein interactions (in the form of physical contacts) between them. This network is the largest connected component of the Coronavirus PPI network that we downloaded on 24 July 2020; in total, there are 1,555 proteins and 2,481 interactions. Of the 2,481 interactions, 1,536 of them are for SARS-CoV-2 and were reported by 44 publications and preprints; the rest are related to coronaviruses that cause Severe Acute Respiratory Syndrome (SARS) or Middle Eastern Respiratory Syndrome (MERS).

				\item \dataset{Homo sapiens PPI} (with the shorthand \dataset{H. sapiens})~\cite{oughtred2019biogrid, oughtred2019biogrid_homo, grover2016node2vec}: This network has 24,407 nodes and 390,420 edges. Its largest connected component has 24,379 nodes and 390,397 edges. We use the complete network in our experiments. The nodes represent proteins in the organism \emph{Homo sapiens}, and the edges encode physical interactions between these proteins.  
			\end{enumerate}

			We now describe our eight synthetic networks: 
			\begin{enumerate}[resume, itemsep=0.1cm]
				\item \dataset{ER$_1$} and \dataset{ER$_2$}: An Erd\H{o}s--R\'{e}nyi (ER) network~\cite{erdds1959random,newman2018}, which we denote by $\textup{ER}(n,p)$, is a random-graph model. The parameter $n$ is the number of nodes and the parameter $p$ is the independent, homogeneous probability that each pair of distinct nodes has an edge between them. The network \dataset{ER$_1$} is an individual graph that we draw from $\textup{ER}(5000,0.01)$, and \dataset{ER$_2$} is an individual graph that we draw from $\textup{ER}(5000,0.02)$. 				
				
				\item \dataset{WS$_1$} and \dataset{WS$_2$}: A Watts--Strogatz (WS) network, which we denote by $\textup{WS}(n,k,p)$, is a random-graph model to study the small-world phenomenon~\cite{watts1998collective,newman2018}. In the version of WS networks that we use, we start with an $n$-node ring network in which each node is adjacent to its $k$ nearest neighbors. With independent probability $p$, we then remove and rewire each edge so that it connects a pair of distinct nodes that we choose uniformly at random. The network \dataset{WS$_1$} is an individual graph that we draw from {$\textup{WS}(5000,50,0.05)$}, and \dataset{WS$_2$} is an individual graph that we draw from $\textup{WS}(5000,50, 0.10)$.
				
				\item \dataset{BA$_1$} and \dataset{BA$_2$}: A Barab\'{a}si--Albert (BA) network, which we denote by $\textup{BA}(n,n_{0})$, is a random-graph model with a linear preferential-attachment mechanism~\cite{barabasi1999emergence,newman2018}. In the version of BA networks that we use, we start with $n_{0}$ isolated nodes and we introduce new nodes with $n_{0}$ new edges each that attach preferentially (with a probability that is proportional to node degree) to existing nodes until we obtain a network with 
				$n$ nodes. The network \dataset{BA$_1$} is an individual graph that we draw from $\textup{BA}(5000,25)$, and \dataset{BA$_2$} is an individual graph that we draw from $\textup{BA}(5000,50)$. 
				
				\item{\dataset{SBM}$_{1}$ and \dataset{SBM}$_{2}$: We use stochastic-block-model (SBM) networks in which each block is an ER network~\cite{holland1983stochastic}.
					Fix disjoint finite sets $C_{1} \cup \cdots \cup C_{k_{0}}$ and a $k_{0} \times k_{0}$ matrix $B$ whose entries are real numbers between $0$ and $1$. An SBM 
					network, which we denote by $\textup{SBM}(C_{1},\ldots, C_{k_{0}}, B)$, has the node set $V = C_{1} \cup \cdots \cup C_{k_{0}}$. For each unordered node pair $\{x,y\}$, there is an edge between $x$ and $y$ with independent probabilities $B[i_{0},j_{0}]$, with indices $i_{0},j_{0} \in \{1,\ldots,k_{0}\}$ such that $x \in C_{i_{0}}$ and $y \in C_{j_{0}}$. If $k_{0} = 1$ and $B$ has a constant $p$ in all entries, this SBM specializes to the Erd\H{o}s--R\'{e}nyi (ER) random-graph model $\textup{ER}(n,p)$ with $n = |C_{1}|$. The networks \dataset{SBM}$_{1}$ and \dataset{SBM}$_{2}$ are individual graphs that we draw from $\textup{SBM}(C_{1},\ldots,C_{k_{0}}, B)$ with $|C_{1}| = |C_{2}| = |C_{3}| =  \text{1,000}$, where $B$ is the $3 \times 3$ matrix whose diagonal entries are $0.5$ in both cases and whose off-diagonal entries are $0.001$ for \dataset{SBM}$_{1}$ and $0.1$ for \dataset{SBM}$_{2}$. Both networks have 3,000 nodes; \dataset{SBM}$_{1}$ has 752,450 edges and \dataset{SBM}$_{2}$ has 1,049,365 edges.}

			\end{enumerate}


			\subsection*{Types of noise}
			{\color{black}
				We now describe the three types of noise
				in our network-denoising experiments. (See Figures~\ref{fig:denoising_caltech_comparison} and~\ref{fig:Figure4}.) These noise types are as follows:
				\begin{enumerate}
					\item{(Noise type: {$-\textup{ER}$})  Given a network $G = (V,E)$, we choose a spanning tree $T$ of $G$ (such a tree includes all nodes of $G$) uniformly at random from all possible spanning trees. Let $E_{0}$ denote the set of all edges of $G$ that are not in the edge set of $T$. 
						We then obtain a corrupted network $G'$ by uniformly randomly removing half of the edges in $E_{0}$ from $G$. 
						Note that $G'$ is guaranteed to be connected.} 
					
					\item{(Noise type: {$+\textup{ER}$}) Given a network $G = (V,E)$, we uniformly randomly choose a set $E_{2}$ of pairs of nonadjacent nodes of $G$ of size $|E_{2}| = \lfloor |E|/2 \rfloor$. The corrupted network is $G' = (V, E\cup E_{2})$; we note that $50\%$ of the edges of $G'$ are new.}
					
					\item{(Noise type: $+\textup{WS}$) Given a network $G = (V,E)$, fix integers $n_{0} \in \{1,\dots,|V|\}$ and $k \in \{1,\dots, n_{0} \}$, and fix a real number $p \in [0,1]$. We uniformly randomly choose a subset $V_{0} \subseteq V$ (with $|V_{0}| = n_{0}$) of the nodes of $G$. We generate a network $H = (V_{0}, E_{3})$ from the Watts--Strogatz model $\textup{WS}(n_{0},k_{0},p)$ using the node set $V_{0}$. We then obtain the corrupted network $G' = (V, E\cup E_{3})$, which has $|E_{3}| = n_{0}\lfloor k/2 \rfloor$ new edges. When $G$ is \dataset{Caltech}, \dataset{SNAP FB}, \dataset{arXiv}, or \dataset{Coronavirus}, we use the parameters $n_{0} = 100$, $k = 20$, and $p=0.3$. In this case, $G'$ has 1,000 new edges. When $G$ is \dataset{H. sapiens}, we use the parameters $n_{0}=500$, $k = 120$, and $p=0.3$. In this case, $G'$ has 30,000 new edges.}
				\end{enumerate}
			}

			\pagestyle{myheadings}
			\markright{\uppercase{Learning low-rank latent mesoscale structures in networks}}
			\markleft{\uppercase{Learning low-rank latent mesoscale structures in networks}}

			
			\subsection*{Data availability}

			The data sets that we generated in the present study are available in the repository \url{https://github.com/HanbaekLyu/NDL_paper}. In the `Data sets' subsection of the `Methods' section, we give references for the real-world networks that we examine.

			
			
			\subsection*{Code availability}

			Our code for our 
			algorithms and simulations is publicly available in the repository \url{https://github.com/HanbaekLyu/NDL_paper}. We also provide a user-friendly version 
			at
			\url{https://github.com/jvendrow/Network-Dictionary-Learning}
			{as a {\sc Python} package \textsc{ndlearn}.}


			\printbibliography[keyword={methods}]
			
			
			\appendix


			\newpage
			${}$
			\vspace{0.5cm}
			
			\pagestyle{myheadings}
			\markright{\uppercase{Learning low-rank latent mesoscale structures in networks}}
			\markleft{\uppercase{Learning low-rank latent mesoscale structures  in networks}}

			\begin{center}
				\textbf{\large Supplementary Information:}
			\end{center}
			\vspace{0.1cm}
			\begin{center}
				\textbf{\large LEARNING LOW-RANK LATENT MESOSCALE STRUCTURES IN NETWORKS}
			\end{center}
			
			\vspace{1cm}

			

			In this supplement, we present our algorithms for network dictionary learning (NDL) and network denoising and reconstruction (NDR), and we prove theoretical results about their convergence and error bounds. In {Appendix}~\ref{section:NDL_problem}, we define a variety of technical terms and overview our theoretical results. In {Appendix}~\ref{section:MCMC}, we discuss Markov-chain Monte Carlo (MCMC) motif-sampling algorithms. We give the complete NDL algorithm (see Algorithm~\ref{alg:NDL}) in {Appendix}~\ref{section:NDL}. We introduce the notion of `latent-motif dominance' in {Appendix}~\ref{subsection:LMD} to measure the significance of each latent motif that we learn from a network.  In {Appendix}~\ref{subsection:various_mesoscales}, we show that various mesoscale structures of the networks we study in the present paper emerge in the latent motifs at various scales $k\in \{6,11,21,51\}$. In {Appendix}~\ref{subsection:NR}, we give the complete NDR algorithm (see Algorithm~\ref{alg:network_reconstruction}). {We give experimental details in {Appendix}~\ref{section:experimental_details}.} In Appendix~\ref{section:convergence}, we present a rigorous analysis of the NDL and NDR algorithms. In {Appendix}~\ref{subsection:aux_alg}, we state auxiliary algorithms that we use in the present paper. In {Appendix}~\ref{subsection:additional_figures}, we show additional figures.

			
			\section{Problem formulation and overview of theoretical results}
			\label{section:NDL_problem}
			
			
			\subsection{Definitions and notation}
			\label{subsection:Definition}
			
			To facilitate our discussions, we use terminology and notation from~\cite[Ch. 3]{lovasz2012large}. In the main manuscript, we described a network as a graph $G = (V,E)$ with a node set $V$ and an edge set $E$ without directed or multi-edges, but possibly with self-edges. An unordered pair $\{x,y\}$ of nodes in $G$ is an \textit{edge} in $G$ if $\{x,y\}\in E$; it
			is a \textit{self-edge} at $x$ if $\{x\}\in E$. 
			One can characterize the edge set $E$ of $\G$ using an \textit{adjacency matrix} $A_{G}:V^{2}\rightarrow \{0,1\}$, where $A(x,y) = \one(\{x,y\}\in E)$ for each $x,y\in V$. The function $\one(S)$ denotes the indicator of the event $S$; it takes the value $1$ if $S$ occurs and takes the value $0$ if $S$ does not occur. In this supplementary information, we formulate our NDL framework in the more general setting in which the edges of a network can have weights. Although one can extend the above definition of networks to include weighted edges by adjoining an additional object
			to $G = (V,E)$ for edge weights, it is convenient to instead extend the range of adjacency matrices from $\{0,1\}$ to the interval $[0,\infty)$. 
			
			We define a \textit{network} as a pair $\G=(V,A_{\G})$ with a node set $V$ and a \textit{weight matrix} (which is also often called a `weighted adjacency matrix') $A_{\G}:V^{2}\rightarrow [0,\infty)$ that encodes the weights of the edges between nodes. For simplicity, we often drop the subscript $\G$ in $A_{\G}$ and denote it by $A$.
			A graph $G = (V,E)$ determines a unique network $\G=(V, A_{G})$, where $A_{G}$ is the adjacency matrix of $G$. The set $V(\G)$ is the node set of the network $\G$, which has \textit{size} $|V(\G)|$, where $|S|$ is the number of elements in the set $S$. {An unordered pair $\{x,y\}$ of nodes of $\G$ is an \textit{edge} if $A(x,y) > 0$ or $A(y,x) > 0$; it is a \textit{nonedge if $A(x,y)=0$;} it is a \textit{self-edge} if $x = y$ and $A(x,x) > 0$. An ordered} pair $(x,y)$ of nodes of $\G$ is called a \textit{directed edge} if $A(x,y) > 0$.

			We say that a network $\G = (V,A)$ is \textit{symmetric} if its weight matrix is symmetric (i.e., $A(x,y)=A(y,x)$ for all $x,y\in V$), and we say that it is \textit{binary} (i.e., unweighted) if $A(x,y)\in \{0,1\}$ for all $x,y\in V$. The network $\G$ is \textit{bipartite} if it admits a `bipartition', which is a partition $V = V_{1} \cup V_{2}$ of the node set $V$ such that $V = V_{1}\cup V_{2}$ and $A(x,y) = 0$ if $x,y \in V_{1}$ or $x,y \in V_{2}$ for each $x,y \in V$. If two networks $\G = (V, A)$ and $\G' = (V',A')$ satisfy $V' \subseteq V$ and $A'(x,y) \le A(x,y)$ for all $x,y \in V'$, then we say that $\G'$ is a \textit{subgraph} of $\G$ and write $\G'\subseteq \G$. If $A'(x,y) = A(x,y)$ for all $x,y \in V'$, then we say that $\G'$ is an \textit{induced subgraph of $\G$} that is induced by the node set $V'$.

				For an integer $k \ge 2$ and nodes $x,y \in V$, we refer to a sequence $(x_{1},\ldots,x_{k})$ of (not necessarily distinct) nodes of $\G$ as a \textit{$k$-walk from $x$ to $y$} if $A(x_{i},x_{i+1}) > 0$ for all $i \in \{1, \ldots, k-1\}$ and $(x_{1},x_{k})=(x,y)$. A $k$-walk $(x_{1},\ldots,x_{k})$ is a \textit{$k$-path} if all nodes $x_{1},\ldots,x_{k}$ are distinct. We say that a network $\G$ is \textit{connected} if for any nodes $x,y \in V$, there exists a $k$-path from $x$ to $y$ for some $k \ge 1$. If $\G$ is connected, then for any two distinct nodes $x,y \in V$, we define $d_{\G}(x,y)$ to be the smallest 
				integer $k\ge 1$ such that there exists a $k$-walk from $x$ to $y$. The quantity $d_{\G}(x,y)$ is called the \textit{shortest-path distance} between $x$ and $y$. The maximum of $d_{\G}(x,y)$ over all node pairs $(x,y)$ is the \textit{shortest-path diameter} $\textup{diam}(\G)$ of $\G$. It equals the minimum number of edges in a walk in $\G$ that connects nodes $x$ and $y$. 
				
				Suppose that we are given $N$ elements $\mathbf{v}_{1},\ldots,\mathbf{v}_{m}$ in some vector space. When we say that we take their \textit{mean}, we refer to their sample mean $\bar{\mathbf{v}} = N^{-1}\sum_{i=1}^{N}\mathbf{v}_{i}$. When we say that we take a \textit{weighted average} of them, we refer to the expectation $\sum_{i=1}^{N} \mathbf{v}_{i}p_{i}$, where $(p_{1},\ldots, p_{N})$ is a probability distribution on the set of $N$ elements.


				\subsection{Homomorphisms between networks and motif sampling}
				\label{subsection:motif_sampling}

				The ability to sample from a complex data set according to a known probability distribution (e.g., a uniform distribution) is a crucial ingredient in dictionary-learning problems. For instance, in image-processing applications~\cite{elad2006image, mairal2008sparse, mairal2009non}, it is straightforward to uniformly randomly sample a $k \times k$ patch from an image. However, it is not straightforward to uniformly randomly sample a connected $k$-node subgraph of a network~\cite{bressan2020faster, kashtan2004efficient, wernicke2006efficient, leskovec2006sampling}. 
				To develop dictionary learning for networks, we use motif sampling, which was introduced recently in~\cite{lyu2023sampling}. In motif sampling, instead of directly sampling a connected subgraph, one samples a random function {that maps} the node set of a smaller network (i.e., a motif) to
				the node set of a target network {while preserving} adjacency relationships.
				One then uses the subgraph that is induced by the nodes in the image of the function. As we discuss below, such a function between networks is a homomorphism.

				Fix an integer $k \ge 1$ and a weight matrix $A_{F}:[k]^{2}\rightarrow [0,\infty)$, where we use the shorthand notation $[k] = \{1,\ldots,k\}$. 
				We use the term \textit{motif} for the corresponding network $F=([k],A_{F})$. A motif is a network, and we use motifs to sample from a given (and much larger) network. 
				The type of motif that particularly interests us is a \textit{$k$-chain}, for which $A_{F}=\one(\{(1,2),(2,3),\ldots,(k-1,k)\})$. A $k$-chain is a directed path with node set $[k]$. For simplicity, we refer to the $k$-chain motif with the corresponding network $F=([k],A_{F})$ as the \textit{$k$-chain motif $F = ([k],A_{F})$}. For a general $k$-node motif $F=([k],A_{F})$ and a network $\G=(V,A)$, we define the probability distribution $\pi_{F\rightarrow \G}$ on the set $V^{[k]}$ of all node maps (i.e., functions between node sets) $\mathbf{x}:[k]\rightarrow V$ by 
				\begin{equation}\label{eq:def_embedding_F_N}
					\pi_{F\rightarrow \G}( \mathbf{x} ) := \frac{1}{\mathtt{Z}}  \left( \prod_{i,j \in \{1, \ldots , k\}}  A(\mathbf{x}(i),\mathbf{x}(j))^{A_{F}(i,j)} \right)\,,
				\end{equation}  
				where $\mathtt{Z}=\mathtt{Z}(F,\G)$ is a normalization constant that we call the \textit{homomorphism density} of $F$ in $\G$~\cite{lovasz2012large}. A node map $\x:[k]\rightarrow V$ is a \textit{homomorphism} $F\rightarrow \G$ if $\pi_{F\rightarrow \G}(\x)>0$, which is the case if and only if  $A(\x(a),\x(b))>0$ for all $a,b\in [k]$ with $A_{F}(a,b)>0$ (with the convention that ${\zeta}^{0}=1$ for all ${\zeta}\in \mathbb{R}$). Informally, this means that if $(i,j)$ is a directed edge in the motif $F$, then $(\x(i),\x(j))$ is a directed edge in the network $\G$. 
				The term `motif sampling' refers to the problem of sampling a random homomorphism $\x: F \rightarrow \G$ according to the distribution \eqref{eq:def_embedding_F_N}. 
				
					To learn interpretable latent motifs, it is important to sample a homomorphism $\x: F \rightarrow \G$ such that $\x$ is injective. This ensures that
					the nodes $\x(1),\ldots,\x(k)$ in $V$ that correspond to nodes $1,\ldots,k$ in the motif $F$ through the homomorphism $\x$ are all distinct. 
					When $\x:F\rightarrow \G$ is an injective homomorphism, we write $\x:F\hookrightarrow \G$. 
					Using the subgraph of $\G$ that is induced by the node set $\{\x(1),\ldots,\x(k) \}$ when $\x$ is injective returns a $k$-node subgraph of $\G$. For convenience, we define the probability distribution 
					\begin{equation}\label{eq:def_embedding_F_N_inj}
						\pi_{F\hookrightarrow \G}( \mathbf{x} ) := C \,  \pi_{F\rightarrow \G}( \mathbf{x} )  \cdot \one(\textup{$\x(1),\ldots,\x(k)$ are distinct})\,,
					\end{equation}  
					where $C > 0$ is a normalization constant. 
					Injective motif sampling refers {to} the problem of sampling {a random} injective homomorphism $\x: F\hookrightarrow \G$ according to the distribution \eqref{eq:def_embedding_F_N_inj}. 
					The probability distribution \eqref{eq:def_embedding_F_N_inj} is well-defined as long as there exists an injective homomorphism $\x: F \rightarrow \G$. As a special case of interest, for a symmetric and binary motif $F$ and a network $\G$, the distributions $\pi_{F\rightarrow \G}$ and $\pi_{F\hookrightarrow \G}$ are the uniform distribution among all homomorphisms $F \rightarrow \G$ and among all injective {homomorphisms} $F\hookrightarrow \G$, respectively. That is, 
					\begin{align}\label{eq:def_embedding_F_N_uniform}
						\textup{$F$, $\G$ are symmetric and binary} \quad \Longrightarrow \quad \begin{matrix}
							\pi_{F\rightarrow \G} = \textup{Uniform}( \{ \x:F\rightarrow \G \})\,, \\
							\pi_{F\hookrightarrow \G} = \textup{Uniform}( \{ \x:F\hookrightarrow \G \}) \,,
						\end{matrix} 
					\end{align}
					which is the case for all of our examples in the main manuscript. In {Appendix}~\ref{section:MCMC}, we discuss three MCMC
					algorithms for motif sampling and propose corresponding algorithms for injective motif sampling by combining them with rejection sampling.

			
			\subsection{Mesoscale patches of networks}
			\label{subsection:mesoscale_patches}

			A homomorphism $F \rightarrow \G$ is a node map $V(F) \rightarrow V(\G)$ that maps the edges of a motif $F$ to edges of a network $\G$, so it maps $F$ onto a subgraph of $\G$. It thereby maps $F$ `into' $\G$. For each homomorphism $\x: F \rightarrow \G$ from a motif $F=([k],A_{F})$ into 
				a network $\G=(V,A)$, we define a $k \times k$ matrix 
				\begin{align} 
					A_{\mathbf{x}}(a,b) := A\big(\x(a), \x(b)\big) \quad \text{for all\,\, $a,b\in \{1,\ldots,k\}$}\,. \label{eq:def_A_F_patch}
				\end{align}
				We say that $A_{\x}$ in \eqref{eq:def_A_F_patch} is the \textit{mesoscale patch} of $\G$ that is induced by the homomorphism $\x:F\rightarrow \G$. The matrix $A_{\x}$
				is specified uniquely by the homomorphism $\x:F\rightarrow \G$ and the weight matrix $A$. Given a $k \times k$ matrix $B$ and a homomorphism $\x: F \rightarrow \G$, we say that the $(a,b)$ entries of $B$ are \textit{on-chain} if $A_{F}(a,b) > 0$ and are \textit{off-chain} otherwise. 
			The condition $A_{F}(a,b) > 0$ implies that $A(\x(a),\x(b)) > 0$ by the definition of the homomorphism $\x$, so the on-chain entries of $A_{\x}$ are always positive (and are always $1$ if $\G$ is unweighted). 
			However, the off-chain entries of $A_{\x}$ 
			are not necessarily positive, so they encode meaningful information about a network that one `detects' with the homomorphism $\x:F\rightarrow \G$. 
			As an illustration, suppose that $F$ is the $6$-chain motif and that $\G=(V,A)$ is an undirected and binary graph. For any homomorphism $\x:F\rightarrow \G$, we have 
			\begin{align}\label{eq:ex_patches}
				A_{\x}=
				\begin{bmatrix}
					0 & 1 & * & * &  * &  * \\
					1 & 0 & 1 & * &  * &  * \\
					* & 1 & 0 & 1 &  * &  * \\
					* & * & 1 & 0 &  1 &  * \\
					* & * & * & 1 &  0 &  1 \\
					* & * & * & * &  1 &  0 
				\end{bmatrix} \, ,
			\end{align}
			where each entry $*$ of $A_{\x}$ is either $0$ or $1$. In this example, the entries that we mark as $1$ are the on-chain entries of $A_{\x}$ and the other entries are off-chain entries.
			
			Let $\G_{\x}$ denote the induced subgraph of $\G$ whose node set is the image $\textup{Im}(\x_{t}) = \{\x(a)\,|\, a\in \{1,\ldots,k\} \}$ of the homomorphism $\x: F \rightarrow \G$. If $\x$ {has} $k$ distinct nodes in its image, then the weight matrix of $\G_{\x}$ is exactly the mesoscale patch $A_{\x}$. 
			However, this is not the case when $\x$ has fewer than $k$ distinct nodes. In that situation, we cannot interpret the mesoscale patch $A_{\x}$ as the weight matrix of the induced subgraph $\G_{\x}$ of $\G$. (For example, see Figure~\ref{fig:covid_dict_comparison} of the main manuscript.) This motivates us to sample an injective homomorphism $\x: F \rightarrow \G$ according to the distribution \eqref{eq:def_embedding_F_N_inj}, instead of according to the distribution \eqref{eq:def_embedding_F_N}.


			\subsection{Problem formulation for network dictionary learning (NDL)}
			\label{subsection:NDL_formulation}
			
			The goal of the \textit{NDL problem} is to learn, for a fixed integer $r \ge 1$, a set of $r$ nonnegative matrices $\mathcal{L}_{1},\ldots,\mathcal{L}_{r}$, with size $k \times k$ and Frobenius norms of at most $1$, such that
			\begin{align}\label{eq:NDL_approx}
				A_{\x} \approx a_{1}(\x)\mathcal{L}_{1}+ \cdots + a_{r}(\x)\mathcal{L}_{r}  
			\end{align}
			for each injective homomorphism $\x:F\hookrightarrow \G$ for some coefficients $a_{1}(\x),\ldots ,a_{r}(\x)\ge 0$. 
			For each injective homomorphism $\x:F\rightarrow \G$, this implies that one can approximate the mesoscale patch $A_{\x}$ of $\G$ that is induced by $\x$ as {a}  linear combination of the $r$ matrices $\mathcal{L}_{1},\ldots,\mathcal{L}_{r}$. 
			We say that the tuple {$(\mathcal{L}_{1},\ldots,\mathcal{L}_{r})$} is a \textit{network dictionary} for $\G$, and we say that each $\mathcal{L}_{i}$ is a \textit{latent motif} of $\G$. We identify a network dictionary {$(\mathcal{L}_{1},\ldots,\mathcal{L}_{r})$} with the nonnegative matrix $W\in \R_{\ge 0}^{k^{2}\times r}$ whose $j^{\textup{th}}$ column is the vectorization of the $j^{\textup{th}}$ latent motif $\mathcal{L}_{j}$ for $j\in \{1, \ldots, r\}$. The choice of vectorization $\R^{k\times k}\rightarrow \R^{k^{2}}$ is arbitrary, but we use a column-wise vectorization in Algorithm~\ref{alg:vectorize}. {One can interpret each $\mathcal{L}_{i}$ as the $k$-node weighted network with node set $\{1,\ldots,k\}$ and weight matrix $\mathcal{L}_{i}$. See Figure~\ref{fig:img_ntwk_dict} of the main manuscript for plots of latent motifs as weighted networks.}

			For the latent motifs $\mathcal{L}_{i}$ to be interpretable as subgraphs of $\G$, 
			we require both their entries and
			the coefficients $a_{i}(\x)$ to be nonnegative. 
			The nonnegativity constraint on each latent motif $\mathcal{L}_{i}$ allows one to interpret each $\mathcal{L}_{i}$ as the weight matrix of a $k$-node network. Additionally, because the coefficients $a_{j}(\x)$ are also nonnegative, the approximate decomposition \eqref{eq:NDL_approx} implies that $a_{i}(\x)\mathcal{L}_{i}\lessapprox A_{\x}$. Therefore, if $a_{i}(\x)> 0$, any network structure (e.g., large-degree nodes, communities, and so on) in the latent motif $\mathcal{L}_{i}$ must also exist in $A_{\x}$. Therefore, one can consider the latent motifs as approximate $k$-node subgraphs $\G$ that exhibit `typical' network structures of $\G$ at scale $k$. In the spirit of Lee and Seung~\cite{lee1999learning}, one can view the latent motifs as `parts'\footnote{Lee and Seung~\cite{lee1999learning} discussed a similar nonnegative decomposition in which the $A_{\x}$ are images of faces. In that scenario, the learned factors capture parts of human faces (such as eyes, noses, and mouths).} of a network $\G$.
			
			As a more precise formulation of \eqref{eq:NDL_approx}, consider the stochastic optimization problem
			\begin{align}\label{eq:NDL_exact}
				\argmin_{\substack{\mathcal{L}_{1},\ldots,\mathcal{L}_{r}\in \R_{\ge 0}^{k\times k} \\ {\lVert \mathcal{L}_{1} \rVert_{F},\ldots,\lVert \mathcal{L}_{r} \rVert_{F}\le 1}}}   \E_{\x\sim \pi_{F\hookrightarrow \G}}\left[ \inf_{a_{1}(\x),\ldots ,a_{r}(\x)\ge 0} \left\lVert A_{\x} - \sum_{i=1}^{r} a_{i}(\x)\mathcal{L}_{i}  \right\rVert_{F}  \right]\,,
			\end{align}
			where {$\pi_{F\hookrightarrow \G}$} is the probability distribution that we defined in \eqref{eq:def_embedding_F_N_inj} and $\lVert \cdot\rVert_{F}$ denotes the matrix Frobenius norm.  The choice of the probability distribution {$\pi_{F\hookrightarrow \G}$} for the injective homomorphisms {$\x: F \hookrightarrow \G$} is natural because it becomes the uniform distribution on the set of all {injective} homomorphisms $F \hookrightarrow \G$ when the adjacency matrices of
			$\G$ and $F$ are both unweighted. Exactly solving the NDL optimization problem \eqref{eq:NDL_exact} is computationally difficult because the objective function in it
			is non-convex and it is not obvious how to sample an injective homomorphism $F \hookrightarrow \G$ according to the distribution $\pi_{F\hookrightarrow \G}$ that we defined in \eqref{eq:def_embedding_F_N_inj}. In {Appendix}~\ref{section:NDL}, we give an algorithm for NDL that approximately solves \eqref{eq:NDL_exact}.


			\subsection{Overview of our algorithms and their theoretical guarantees} 
			
			We overview our algorithms and their theoretical guarantees.
			Our main theoretical results (which are all novel) are Theorems~\ref{thm:NDL} and~\ref{thm:NDL2} for NDL and Theorems~\ref{thm:NR} and~\ref{thm:NR2} for NDR. 
			We summarize our algorithms and main results in Table~\ref{table:NDL_comparison_SI}, and we compare and contrast them with the results in~\cite{lyu2020online}.

			\begin{table}[htbp]
				\centering
				\begin{tabular}{c|cccc}
					\thickhline 
					\textbf{NDL} & Sampling  & Convergence & Efficient MCMC \\ 
					\hline
					\rule{0pt}{1.1\normalbaselineskip} Lyu et al.~\cite{lyu2020online} & $k$-walks  & Non-bipartite networks & \xmark \\[0.2cm]
					\rule{0pt}{1.1\normalbaselineskip}  $\begin{matrix} \textup{The present} \\ \textup{work} \end{matrix}$ & $\begin{matrix} \textup{$k$-paths} \\ \textup{(Alg.~\ref{alg:motif_inj}) } \end{matrix}$ & $\begin{matrix} \textup{Non-bipartite networks}  \\ \textup{(Thm.~\ref{thm:NDL}) } \end{matrix}$ \quad $\begin{matrix} \textup{Bipartite networks}  \\ \textup{(Thm.~\ref{thm:NDL2}) } \end{matrix}$   & $\begin{matrix} \checkmark\\ \textup{(Prop.~\ref{prop:approximate_pivot}) } \end{matrix}$ \\[0.1cm]
					\hline
					\thickhline 
				\end{tabular}%
				\\[0.2cm]
				\begin{tabular}{c|ccccc}
					\thickhline 
					\textbf{NDR} & Sampling  &  \texttt{denoising} & Convergence & Error bound \\ 
					\hline 
					\rule{0pt}{1.2\normalbaselineskip} Lyu et al.~\cite{lyu2020online} & $k$-walks  & \texttt{F} & \xmark & \xmark \\[0.2cm]
					$\begin{matrix} \textup{The present work}  \end{matrix}$  & $\begin{matrix} \textup{$k$-walks} \\ \textup{$k$-paths} \end{matrix}$    &  {\texttt{F}, \texttt{T}} & $\begin{matrix} \checkmark\\  \textup{ (Thm.~\ref{thm:NR}\textbf{(i)}{--}\textbf{(ii)}) } \\ \textup{ (Thm.~\ref{thm:NR2}\textbf{(i)}{--}\textbf{(ii)}) }  \end{matrix}$ & $\begin{matrix} \checkmark\\ \textup{ (Thm.~\ref{thm:NR}\textbf{(iii)}) } \\ \textup{ (Thm.~\ref{thm:NR2}\textbf{(iii)}) }  \end{matrix}$  \\ 
					\hline
					\thickhline 
				\end{tabular}%
				\caption{A comparison of the algorithms and theoretical contributions of the present work {with those in}~\cite{lyu2020online}. {In the table for NDR, $\texttt{denoising}$ refers to the Boolean variable in {our} NDR algorithm (see Algorithm~\ref{alg:network_reconstruction}). {The special case of the} NDR algorithm with $k$-walk sampling (i.e., $\texttt{InjHom}=\texttt{F}$) and $\texttt{denoising}=\texttt{F}$ is
				 the network-reconstruction algorithm in~\cite{lyu2020online}.
				 }
				}
				\label{table:NDL_comparison_SI}
			\end{table}
			
			\medskip
			
			\begin{description}[leftmargin=0.7cm, topsep=0.2pt, itemsep=0.1cm]
				\item[Algorithm~\ref{alg:NDL}:] Given a network $\G$, the NDL algorithm (see Algorithm~\ref{alg:NDL}) computes a sequence $(W_{t})_{t\ge 0}$ of network dictionaries {(which take the form of $k^{2}\times r$ matrices)} of latent motifs.
				
				\vspace{0.1cm}
				\item[Algorithm~\ref{alg:network_reconstruction}:] Given a
				network $\G$, a network dictionary $W$, 
				the NDR algorithm (see Algorithm~\ref{alg:network_reconstruction}) computes a sequence of weighted networks $\G_{\textup{recons}}$.
				
				\vspace{0.1cm}
				\item[Theorem~\ref{thm:NDL}:] Given a non-bipartite network $\G$ and a choice of the parameters in Algorithm~\ref{alg:NDL}, we prove that the sequence $(W_{t})_{t\ge 0}$ of network dictionaries converges almost surely to the set of stationary points of the 
				{objective} function in \eqref{eq:NDL_exact}.

				\vspace{0.1cm}
				\item[Theorem~\ref{thm:NDL2}:] Given a bipartite network $\G$ and a choice of the parameters in Algorithm~\ref{alg:NDL}, we prove a convergence result that is analogous to the one in Theorem~\ref{thm:NDL}.

				\vspace{0.1cm}
				\item[Theorem~\ref{thm:NR}:] Given a non-bipartite target network $\G$ and a network dictionary $W$, we show that \textbf{(i)} the sequence of weighted reconstructed networks $\G_{\textup{recons}}$ that we 
				obtain using the NDR algorithm (see Algorithm~\ref{alg:network_reconstruction}) converges almost surely to some limiting network and \textbf{(ii)} we obtain a closed-form expression for the weight matrix of this limiting network. 
				 We also show that \textbf{(iii)} a suitable Jaccard reconstruction error between the original network $\G$ and the limiting reconstructed network satisfies 
					\begin{align}
						\textup{Jaccard reconstruction error} \le \frac{\textup{mesoscale approximation error} }{2(k-1)} \, ,
					\end{align}
					where $k$ denotes the mesoscale parameter (i.e., the number of nodes in a $k$-chain motif) and the mesoscale approximation error
					is the mean $L_{1}$ distance between the $k \times k$ mesoscale patches of $\G$ and their nonnegative linear approximations from the latent motifs in $W$.
				
				\vspace{0.1cm}
				\item[Theorem~\ref{thm:NR2}:] We show a convergence result that is {analogous} to the one in Theorem~\ref{thm:NR2} for a bipartite target network $\G$.

			\end{description}	
			

			
			\section{Markov-Chain Monte Carlo (MCMC) Motif-Sampling Algorithms}
			\label{section:MCMC}

			In {Appendix}~\ref{subsection:NDL_formulation}, we mentioned that one of the main difficulties in solving the optimization problem \eqref{eq:NDL_exact} is to directly sample an injective homomorphism $\x:F\hookrightarrow \G$ from the distribution $\pi_{F\hookrightarrow \G}$ (see \eqref{eq:def_embedding_F_N_inj}). To overcome this difficulty, we use (and extend to one new variant) the Markov-chain Monte Carlo (MCMC) algorithms that were introduced in~\cite{lyu2023sampling}. Although the algorithms in~\cite{lyu2023sampling} apply to networks with edge weights and/or node weights, we only use the simplified forms of them that we give in Algorithms~\ref{alg:pivot} and~\ref{alg:glauber}. Algorithm~\ref{alg:pivot} with the option $\texttt{AcceptProb} = \texttt{Approximate}$ is a novel algorithm of the present paper. Using these MCMC sampling algorithms, we generate a sequence $(\x_{t})_{t\ge 0}$ of homomorphisms $F\rightarrow \G$ such that the distribution of $\x_{t}$ converges to $\pi_{F\rightarrow \G}$ under some mild conditions on $\G$ and $F$~\cite[Thm. 5.7]{lyu2020online}.

			{\color{black} 
				Once we have 
				an iterative motif-sampling algorithm, we combine it with a standard rejection-sampling algorithm to shrink the support of the probability distribution $\pi_{F\rightarrow \G}$ to injective homomorphisms 
				$\x:F\hookrightarrow \G$. 
				(See, e.g.,~\cite{glasserman2004monte} for background information about rejection sampling.) The key idea is to ignore (i.e., `reject') the unwanted instances in the trajectory $(\x_{t})_{t\ge 0}$. In our case, the instances that we reject are the homomorphisms $\x_{t}$ that are not injective. That is, we reject situations in which $\x_{t}(1),\ldots,\x_{t}(k)$ are not all distinct. 
				In Algorithm~\ref{alg:motif_inj}, we state our algorithm for injective motif sampling.

				\begin{algorithm}
					\renewcommand{\thealgorithm}{IM}
					\begin{algorithmic}[1]
						\caption{\!\!. Injective MCMC motif sampling} \label{alg:motif_inj}
						\State \textbf{Input:} Network $\G=(V,A)$, motif $F=([k],A_{F})$, and homomorphism $\mathbf{x}:F\rightarrow \G$
						
						\vspace{0.1cm}
						\State \textbf{While:} $\mathbf{x}:F\rightarrow \G$ is injective (i.e., $\x'(1),\ldots,\x'(k)$ are distinct)
						\State \qquad Update $\x$ to a new homomorphism $F\rightarrow \G$ using either Algorithm~\ref{alg:pivot} or Algorithm~\ref{alg:glauber}
						
						\State \textbf{Output:} Injective homomorphism $\mathbf{x}:F\hookrightarrow \G$
						
					\end{algorithmic}
				\end{algorithm}

				Algorithm~\ref{alg:motif_inj} restricts the state space of the MCMC motif-sampling algorithms (see Algorithms~\ref{alg:glauber} and~\ref{alg:pivot}) {to} the subset of injective homomorphisms $F\hookrightarrow \G$. 
				By the strong Markov property, this restriction is a Markov chain. Therefore, Algorithm~\ref{alg:motif_inj} is an MCMC algorithm for the injective motif-sampling problem. (See Proposition~\ref{prop:inj_motif_conv} for details.) When there are only a few injective homomorphisms $F \hookrightarrow \G$ relative to the number of homomorphisms $F\rightarrow \G$, the rejection step (i.e., the while loop) in Algorithm~\ref{alg:motif_inj} may take a while to terminate. ({The number of rejections until termination} is inversely proportional to the probability that a random homomorphism under the probability distribution $\pi_{F\rightarrow \G}$ is injective.) 
				For example, this is the case when $\G$ is the network \dataset{Coronavirus PPI} and $F$ is a $k$-chain motif with $k \ge 21$. }

			We now give more details about the MCMC algorithms that we employ for non-injective motif sampling. In the \textit{pivot chain} (see Algorithm~\ref{alg:pivot} with $\texttt{AcceptProb} = \texttt{Exact}$), for each update $\x_{t}\mapsto \x_{t+1}$, the {\textit{pivot}} $\x_{t}(1)$ first performs a random-walk move on $\G$ (see \eqref{eq:RWkernel_G}) to move to a new node $\x_{t+1}(1)\in V$. It accepts this move with a suitable acceptance probability (see \eqref{eq:pivot_chain_acceptance_prob}) according to the Metropolis--Hastings algorithm (see, e.g.,~\cite[Sec. 3.2]{levin2017markov}), so the stationary distribution is	exactly the target distribution. 
			After the move $\x_{t}(1)\mapsto \x_{t+1}(1)$, we sample each $\x_{t+1}(i)\in V$ for $i \in \{2,3,\ldots,k\}$ successively from the conditional distribution \eqref{eq:pivot_conditional}. This ensures that the desired distribution $\pi_{F\rightarrow \G}$ in \eqref{eq:def_embedding_F_N} is a stationary distribution of the resulting Markov chain. In the \textit{Glauber chain} {(see Algorithm \ref{alg:glauber})}, we select one node $i\in [k]$ of $F$ uniformly at random, and we resample its location $\x_{t}(i)\in V(\G)$ at time $t$ to 
			$\x_{t+1}(i)\in V$ 
			from the conditional distribution \eqref{eq:marginal_distribution_glauber} (see Figure~\ref{fig:NDL_alg}\textbf{a}).
			See~\cite[Sec. 3.3]{levin2017markov} for discussions of the Metropolis--Hastings algorithm and Glauber-chain MCMC sampling.

			\begin{algorithm}
				\renewcommand{\thealgorithm}{MP}
				\begin{algorithmic}[1]
					\caption{\!\!. Pivot-Chain Update} \label{alg:pivot}

					\State \textbf{Input:} Symmetric network $\G=(V,A)$, a $k$-chain motif $F=([k],A_{F})$, and homomorphism $\mathbf{x}:F\rightarrow \G$
					\vspace{0.1cm}
					
					\State \textbf{Parameters:} $\mathtt{AcceptProb}\in \{\mathtt{Exact}, \mathtt{Approximate}\}$

					\vspace{0.1cm}
					\State \textbf{Do:}  $\x'\leftarrow \x$ 
					\State \qquad \textbf{If} $\sum_{c\in V}  A(\x(1),c)=0$: \textbf{Terminate}   
					
					\State \qquad \textbf{Else}:
					\State \quad \qquad Sample $\iota \in V$ at random from the distribution 
					
					\begin{align}\label{eq:RWkernel_G}
						p_{1}(w) = \frac{ A(\x(1),w) }{ \sum_{c\in V}  A(\x(1),c)  }\,\,, \quad w\in V
					\end{align}
					
					\State \quad\qquad Compute the acceptance probability $\alpha\in [0,1]$ by 
					\begin{align}\label{eq:pivot_chain_acceptance_prob}
						\alpha \leftarrow  
						\begin{cases}
							\min \left\{ \frac{ \sum_{c\in [n]} A^{k-1}(\iota,c) }{ \sum_{c\in [n]} A^{k-1}(\x(1),c) }\frac{\sum_{c\in V}  A(c,\x(1))}{\sum_{c\in V}  A(\x(1),c)},\,  1\right\}\,, & \text{if\quad  $\mathtt{AcceptProb}=\mathtt{Exact}$} \\[10pt]
							\min \left\{ \frac{\sum_{c\in V}  A(c,\x(1))}{\sum_{c\in V}  A(\x(1),c)},\,  1\right\}\,, &  \text{if\quad  $\mathtt{AcceptProb} = \mathtt{Approximate}$} 
						\end{cases}
					\end{align}
					
					\State \quad\qquad Sample $U\in [0,1]$ uniformly at random, independently of everything else
					
					\State \quad\qquad $\iota \leftarrow \x(1)$ if $U>\lambda$ and $\x'(1)\leftarrow \iota$ 
					
					\State \quad\qquad \textbf{For $i=2,3,\ldots,k$}:
					\State \quad\qquad \quad Sample $\x'(i)\in V$ from the distribution
					
					\begin{align}\label{eq:pivot_conditional}
						p_{i}(w) = \frac{A(\x(i-1), w)}{\sum_{c\in V} A(\x(i-1), c)}\,\,, \quad w\in V
					\end{align}

					\State \textbf{Output:} Homomorphism $\mathbf{x}':F\rightarrow \G$
				\end{algorithmic}
			\end{algorithm}

			\begin{algorithm}
				\renewcommand{\thealgorithm}{MG}
				\begin{algorithmic}[1]
					\caption{\!\!. Glauber-Chain Update} \label{alg:glauber}
					\State \textbf{Input:} Network $\G=(V,A)$, a $k$-chain motif $F=([k],A_{F})$, and homomorphism $\mathbf{x}:F\rightarrow \G$
					
					\vspace{0.1cm}
					\State \textbf{Do:} Sample $v\in [k]$ uniformly at random 
					\State \qquad Sample $z \in V$ at random from the distribution 
					\begin{align}\label{eq:marginal_distribution_glauber}
						p(w) = \frac{1}{Z}\left( \prod_{u\in [k]} A(\mathbf{x}(u), w)^{A_{F}(u,v)} \right) \left( \prod_{u\in [k]} A(w, \mathbf{x}(u))^{A_{F}(v,u)} \right)\,, \quad w \in V
					\end{align}
					\qquad where $Z=\sum_{c\in V} \left(\prod_{u\in [k]} A(\mathbf{x}(u), c)^{A_{F}(u,v)} \right) \left( \prod_{u\in [k]} A(c, \mathbf{x}(u))^{A_{F}(v,u)}\right)$ is the normalization constant
					\State \qquad Define a new homomorphism $\mathbf{x}':F\rightarrow \G$ by $\mathbf{x}'(w)=z$ if $w=v$ and $\mathbf{x}'(w)=\mathbf{x}(w)$ otherwise 
					
					\State \textbf{Output:} Homomorphism $\mathbf{x}':F\rightarrow \G$

				\end{algorithmic}
			\end{algorithm}

			Let $\Delta$ denote the maximum degree (i.e., number of neighbors) of the nodes of the network $\G=(V,A)$. We also say that the network $\G$ itself has a maximum degree of $\Delta$.
			The Glauber chain has an efficient local update (with a computational complexity of $O(\Delta)$). It converges quickly to the stationary distribution $\pi_{F\rightarrow \G}$ 
			{for} networks that are dense enough so that two homomorphisms that differ at one node have a probability of at least $1/(2\Delta)$ to coincide after a single Glauber-chain update. 
			See~\cite[Thm. 6.1]{lyu2023sampling} for a precise statement of this fact.

			The pivot chain (see Algorithm~\ref{alg:pivot} with $\texttt{AcceptProb} = \texttt{Exact}$) has more computationally expensive local updates than the Glauber chain. The pivot chain has a computational complexity of $O(\Delta^{k-1})$ (as discussed in~\cite[Remark 5.6]{lyu2023sampling}), but it converges as fast as a `lazy' random walk on a network. (In a lazy random walk, each move has a chance to be rejected; see~\cite[Thm. 6.2]{lyu2023sampling}.) In our computational experiments, we find that the Glauber chain is slow, especially for sparse networks (e.g., for \dataset{COVID PPI}, which has an edge density of $0.0010$, and \dataset{UCLA}, which has an edge density of $0.0037$) and that the pivot chain is too expensive to compute for chain motifs with $k \geq 21$. As a compromise, to simultaneously have low computational complexity and fast convergence (it
			is as fast as the standard random walk), we employ an approximate pivot chain, which is Algorithm~\ref{alg:pivot} with the option $\texttt{AcceptProb} = \texttt{Approximate}$. Specifically, we compute the acceptance probability $\alpha$ in \eqref{eq:pivot_chain_acceptance_prob} only approximately and thereby reduce the computational cost to $O(\Delta)$. Our compromise, which we discuss in the next paragraph, is that the stationary distribution of the approximate pivot chain may be slightly different from our target distribution $\pi_{F\rightarrow \G}$.

			Define a probability distribution $\hat{\pi}_{F\rightarrow \G}$ on the set of all node maps $\x:[k]\rightarrow V$ by
				\begin{align}\label{eq:approx_pivot_statioanary_dist}
					\hat{\pi}_{F\rightarrow \G}(\x) := \frac{\prod_{i=1}^{k} A(\x(i-1),\x(i))}{|V|\sum_{y_{2},\ldots,y_{k}\in V} A(\x(1),y_{2})\prod_{i=3}^{k}A(y_{i-1},y_{i})} \,.
				\end{align} 
				According to Proposition~\ref{prop:approximate_pivot}, the stationary distribution of the approximate pivot chain is \eqref{eq:approx_pivot_statioanary_dist}. The distribution \eqref{eq:approx_pivot_statioanary_dist} is different from the desired target distribution $\pi_{F\rightarrow \G}$. 
			Specifically, $\pi_{F\rightarrow \G}(\x)$ is proportional only to the numerator in \eqref{eq:approx_pivot_statioanary_dist}; the sum in the denominator of  \eqref{eq:approx_pivot_statioanary_dist} is a weighted count of the homomorphisms $\y: F\rightarrow \G$ for which $\y(1)=\x(1)$. Therefore, under $\hat{\pi}_{F\rightarrow \G}$, {we penalize the probability of each homomorphism $\x:F\rightarrow \G$ according} to the number of $k$-walks in $\G$ that start from $\x(1) \in V$. (The exact acceptance probability in \eqref{eq:pivot_chain_acceptance_prob} neutralizes this penalty.) It follows that $\hat{\pi}_{F\rightarrow \G}$ is close to $\pi_{F\rightarrow \G}$ when the $k$-step-walk counts that start from each node in $\G$ do not differ too much for different nodes. {For example, on degree-regular networks like lattices, such counts do not depend on the starting node, and it thus follows that $\hat{\pi}_{F\rightarrow \G}=\pi_{F\rightarrow \G}$.} Nevertheless, despite the potential discrepancy between $\pi_{F\rightarrow \G}$ and $\hat{\pi}_{F\rightarrow \G}$, the approximate pivot chain gives good results for the reconstruction and denoising experiments that we showed in Figures~\ref{fig:network_recons} and~\ref{fig:Figure4} of the main manuscript.


			\section{Algorithm for Network Dictionary Learning (NDL)} 
			\label{section:NDL}
			
			
			\subsection{Algorithm overview and statement}
			
			The essential idea behind our algorithm for NDL (see Algorithm~\ref{alg:NDL}) is as follows. Suppose that we compute all possible injective homomorphisms $\x_{1},\ldots,\x_{M}:F\hookrightarrow \G$ and their corresponding mesoscale patches $A_{\x_{t}}$ for $t \in \{1, \ldots, M\}$. These $M$ mesoscale patches of $\G$ form the data set 
			in which we apply a dictionary-learning algorithm. To do this, we column-wise vectorize each of these $k \times k$ matrices (using Algorithm~\ref{alg:vectorize}) and obtain a $k^{2}\times M$ data matrix $X$, and we then apply nonnegative matrix factorization (NMF)~\cite{lee1999learning} to obtain a $k^{2} \times r$ nonnegative matrix $W$ for some fixed integer $r \ge 1$ to yield an approximate factorization $X \approx WH$ for some nonnegative matrix $H$. From this procedure, we approximate each column of $X$ by the nonnegative linear combination of the $r$ columns of $W$ with coefficients that are given by entries of the $r^{\textup{th}}$ column of $H$. Therefore, if we let $\mathcal{L}_{i}$ be the $k \times k$ matrix that we obtain by reshaping the $i^{\textup{th}}$ column of $W$ (using Algorithm~\ref{alg:reshape}), then {$(\mathcal{L}_{1},\ldots,\mathcal{L}_{r})$} is an approximate solution of \eqref{eq:NDL_exact}. We give the precise meaning of `approximate solution' in Theorems~\ref{thm:NDL} and~\ref{thm:NDL2}.

			\begin{algorithm}
				\renewcommand{\thealgorithm}{{NDL}}
				\begin{algorithmic}[1]
					\caption{\!\!. Network Dictionary Learning (NDL)}\label{alg:NDL}
					
					\State \textbf{Input:} Network $\G=(V,A)$ 
					
					\vspace{0.1cm}
					\State \textbf{Parameters:}  {$F=([k],A_{F})$ (a {$k$-chain} motif)\,,\, $T\in \mathbb{N}$ (the number of iterations)\,,\, $N\in \mathbb{N}$ (the number of homomorphisms per iteration)\,,\, $r\in \mathbb{N}$ (the number of latent motifs)\,, \, $\lambda\ge 0$ ({the coefficient of an $L_{1}$-regularizer)}}
					
					\vspace{0.1cm}
					\State {\textbf{Options:} $\mathtt{MCMC}\in \{\mathtt{Pivot},\, \mathtt{PivotApprox},\, \mathtt{Glauber} \}$ }

					\vspace{0.1cm}
					\State \textbf{Requirement:} There exists at least one injective homomorphism $F\hookrightarrow \G$

					\vspace{0.1cm}
					\State \textbf{Initialization:} 
					\State \quad Sample a homomorphism $\x:F\rightarrow \G$ using the rejection sampling (see Algorithm~\ref{alg:rejection_motif}) 
					\vspace{0.1cm}
					\State \quad $W=$ matrix of size $k^{2} \times r$ with independent entries that we sample uniformly from $[0,1]$
					\State \quad $P_{0}=$ matrix of size $r \times r$ whose entries are $0$ 
					\State \quad $Q_{0}=$ matrix of of size $r \times k^{2}$ whose entries are $0$ 
					\vspace{0.1cm}
					\State \textbf{For $t=1,2,\ldots,T$:}
					\vspace{0.1cm}
					\State \quad \textit{MCMC update and sampling mesoscale patches}:
					\State \qquad Successively generate $N$ injective homomorphisms $\x_{N(t-1)+1}, \x_{N(t-1)+2},\ldots, \x_{Nt}$ {by applying Algorithm~\ref{alg:motif_inj} with}
					\begin{align}
						\text{Algorithm~\ref{alg:pivot} with $\mathtt{AcceptProb}=\mathtt{Exact}$}  &\qquad \text{if \quad $\mathtt{MCMC}=\mathtt{Pivot}$} \\ 
						\text{Algorithm~\ref{alg:pivot} with $\mathtt{AcceptProb}=\mathtt{Approximate}$}  &\qquad \text{if \quad $\mathtt{MCMC}=\mathtt{PivotApprox}$} \\ 
						\text{Algorithm~\ref{alg:glauber} with $\mathtt{AcceptProb}=\mathtt{Glauber}$}  &\qquad \text{if \quad $\mathtt{MCMC}=\mathtt{Glauber}$} 
					\end{align}
					
					\State \qquad \textbf{For $s = N(t-1)+1, \ldots, Nt$}:
					
					\State \qquad \quad  $A_{\mathbf{x}_{s}}\leftarrow$ $k\times k$ mesoscale patch of $\G$ that is induced by $\x_{s}$ (see \eqref{eq:def_A_F_patch})

					\State \qquad  \quad $X_{t}\leftarrow$ $k^{2}\times N$ matrix whose $j^{\textup{th}}$ column is $\mathtt{vec}(A_{\mathbf{x}_{\ell}})$ with  $\ell=N(t-1)+j$
					\Statex \qquad \qquad \qquad (where $\mathtt{vec}(\cdot)$ denotes the vectorization operator that we defined in Algorithm~\ref{alg:vectorize}) 
					
					\vspace{0.1cm}
					\State \quad \textit{Single iteration of online nonnegative matrix factorization}: 
					\begin{align}\label{eq:def_ONMF}
						\hspace{1cm}
						\begin{cases}
							H_{t} \leftarrow \argmin_{H\in \R_{\ge 0}^{r\times N}} \lVert X_{t} - W_{t-1}H \rVert_{F}^{2} + \lambda \lVert H \rVert_{1}\qquad (\text{using Algorithm~\ref{algorithm:spaser_coding}}) \\
							P_{t} \leftarrow (1-t^{-1})P_{t-1}+ t^{-1} H_{t}H_{t}^{T} \\
							Q_{t} \leftarrow (1-t^{-1})Q_{t-1}+ t^{-1} H_{t}X_{t}^{T} \\
							W_{t} \leftarrow \argmin_{W\in \mathcal{C}^{\textup{dict}} \subseteq \R_{\ge 0}^{k^{2}\times r}} \left(  \tr(W P_{t} W^{T})  - 2\,\tr(W Q_{t})\right) \qquad (\text{using Algorithm~\ref{algorithm:dictionary_update}})\,,
						\end{cases}
					\end{align}
					\qquad \qquad where $\mathcal{C}^{\textup{dict}} = \{ W\in \R_{\ge 0}^{k^{2}\times r} \,|\, \text{columns of $W$ have a Frobenius norm of at most $1$}\}$
					
					\State \textbf{Output:} Network dictionary $W_{T}\in \R_{\ge 0}^{k^{2}\times r}$
				\end{algorithmic}
			\end{algorithm}

			The scheme in the paragraph above requires one to store all $M$ mesoscale patches, entailing a memory requirement that is at least of order $k^{2}M$, where $M$ is the number of all possible injective homomorphisms $F\hookrightarrow \G$. Because $M$ grows with the number of nodes of $\G$, we need unbounded memory to handle arbitrarily large networks. To address this issue, Algorithm~\ref{alg:NDL} implements the above scheme in the setting of `online learning', where subsets (so-called `minibatches') of data arrive in a sequential manner and one does not store previous subsets of the data before processing new subsets. Specifically, at each iteration $t \in \{1,2,\ldots, T\}$, we process a sample matrix $X_{t}$ that is smaller than the full matrix $X$ and includes only $N\ll M$ mesoscale patches, where one can take $N$ to be independent of the network size. Instead of using a standard NMF algorithm for a fixed matrix~\cite{lee2001algorithms}, we use an `online' NMF algorithm~\cite{mairal2010online, lyu2020online} that one can use on sequences of matrices, where the intermediate dictionary matrices $W_{t}$ that we obtain by factoring the sample matrix $X_{t}$ typically improves as we iterate
			(see~\cite{mairal2010online, lyu2020online}).
			In Algorithm~\ref{alg:NDL}, we give a complete implementation of the NDL algorithm.

			\begin{figure*}[h]
				\centering
				\includegraphics[width=1 \linewidth]{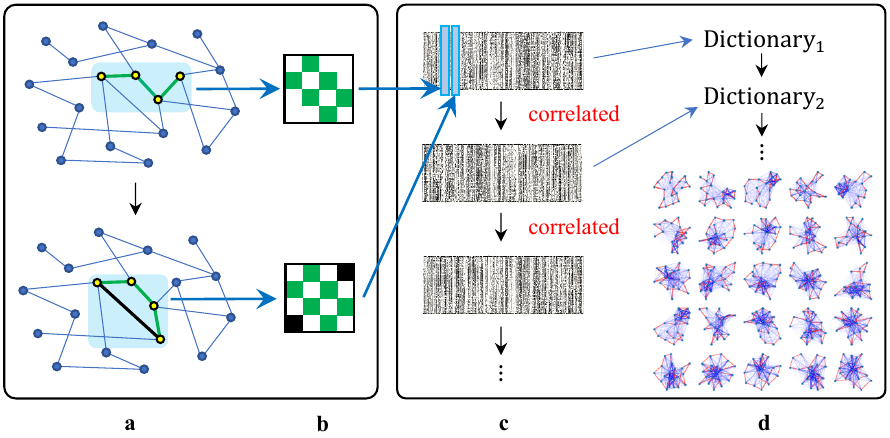}
				\caption{Illustration of our network dictionary learning (NDL) algorithm (see Algorithm~\ref{alg:NDL}). (\textbf{a}) Homomorphisms $\x_{t}: F\rightarrow \G$ from a $k$-chain motif into a target network $\G$ evolve as a Markov chain to yield a sequence of $k$-chain subgraphs (the green edges) in $\G$. (\textbf{b}) Each copy of the $k$-chain motif in $\G$ induces a $k$-node subgraph (i.e., the mesoscale patch $A_{\x_{t}}$ that we defined in \eqref{eq:def_A_F_patch}). (\textbf{c}) We form a sequence of matrices $X_{t}$ of size $k^{2}\times N$, where the $N$ columns of each $X_{t}$ are vectorizations of the $N$ consecutive $k \times k$ mesoscale patches in panel {\textbf{b}}. 
				The matrices $X_{1},X_{2},\ldots$ are correlated with each other because we sample their columns from the Markov chain $\x_{t}$.
					 (\textbf{d}) Using an online nonnegative matrix factorization (NMF) 
						algorithm, we progressively learn the desired number of latent motifs as the data matrix of mesoscale patches $X_{t}$ arrives. 
				}
				\label{fig:NDL_alg}
			\end{figure*}

			We now explain how our NDL algorithm works. It combines one of the three MCMC algorithms --- a pivot chain (in which we use Algorithm~\ref{alg:pivot} with $\mathtt{AcceptProb} = \mathtt{Exact}$), an approximate pivot chain (in which we use Algorithm~\ref{alg:pivot} with $\mathtt{AcceptProb} = \mathtt{Approximate}$), and a Glauber chain (in which we use Algorithm~\ref{alg:glauber}) --- for injective motif sampling that we presented in {Appendix}~\ref{section:MCMC} with the online NMF algorithm from~\cite{lyu2020online}. Suppose that we have {an undirected and unweighted graph} $\G=(V,A)$ and a $k$-chain motif $F=([k],A_{F})$. Furthermore, assume that we satisfy the requirement in Algorithm~\ref{alg:NDL} that there exists at least one injective homomorphism $F\hookrightarrow \G$. At each iteration $t \in \{1,2,\ldots,T\}$, the injective MCMC motif-sampling algorithm generates a sequence $\x_{s}: F \rightarrow \G$ of $N$ injective homomorphisms and corresponding mesoscale patches $A_{\x_{s}}$ (see Figure~\ref{fig:NDL_alg}\textbf{a}). We summarize this sequence in the $k^{2} \times N$ data matrix $X_{t}$. The online NMF algorithm in \eqref{eq:def_ONMF} learns a nonnegative factor matrix $W_{t}$ of size $k^{2}\times r$ by improving the previous factor matrix $W_{t-1}$ by using the new data matrix $X_{t}$. It is an `online' NMF algorithm because it factorizes a sequence $(X_{t})_{t \in \{1, \ldots, T\}}$ of data matrices, rather than a single matrix as in conventional NMF algorithms~\cite{lee2001algorithms}. As it proceeds, the algorithm only needs to store auxiliary matrices $P_{t}$ and $Q_{t}$ of fixed sizes $r\times r$ and $r\times k^{2}$, respectively; it does not need the previous data matrices $X_{1},\ldots, X_{t-1}$. Therefore, NDL is efficient in memory and scales well with network size. It is also applicable to time-dependent networks because of its online nature, although we do not study such networks in the present paper.

			In \eqref{eq:def_ONMF}, we solve convex optimization problems to find matrices $H_{t}\in \R^{r\times N}$ and $W_{t}\in \R^{k^{2}\times r}$.
			The subproblem in \eqref{eq:def_ONMF} of computing $H_{t}$ is a `coding problem'. Given two matrices $X_{t}$ and $W_{t-1}$, we seek to find a factor matrix (i.e., a `code matrix') $H_{t}$ such that $X_{t}\approx W_{t-1}H_{t}$. The parameter $\lambda \ge 0$ is an $L_{1}$-regularizer, which encourages $H_{t}$ to have a small $L_{1}$ norm. One can solve the coding problem efficiently by using Algorithm~\ref{algorithm:spaser_coding} or one of a variety of existing algorithms (e.g., layer-wise adaptive-rate scaling (LARS)~\cite{efron2004least}, LASSO~\cite{tibshirani1996regression}, or feature-sign search~\cite{lee2007efficient}). The second and third lines in \eqref{eq:def_ONMF} update the `aggregate matrices' $P_{t-1}\in \R^{r\times r}$ and $Q_{t-1}\in \R^{r\times k^{2}}$ by taking {a weighted} average {of them} with the new information $X_{t}H_{t}^{T}\in \R^{r\times r}$ and $H_{t}X_{t}^{T}$, respectively. We weight the old aggregate matrices by $1-t^{-1}$ and the new information by $t^{-1}$. By induction, 
			 $P_{t} = t^{-1}\sum_{s=1}^{t} H_{s}H_{s}^{T}$ and $Q_{t} = t^{-1}\sum_{s=1}^{t} H_{t}X_{t}^{T}$. We use the updated aggregate matrices, $P_{t}$ and $Q_{t}$, in the subproblem in \eqref{eq:def_ONMF} of computing $W_{t}$. 
		The subproblem in \eqref{eq:def_ONMF} of computing $W_{t}$ is a constrained quadratic problem; we can solve it using projected gradient descent (see Algorithm~\ref{algorithm:dictionary_update}). In all of our experiments, we take the compact and convex constraint set $\mathbb{R}_{\ge 0}^{k^{2}\times r}$ to be the set of $W\in \R_{\ge 0}^{k^{2}\times r}$ whose columns have a Frobenius norm of at most $1$ (as required in \eqref{eq:NDL_exact}).


			\subsection{Dominance scores of latent motifs}
			\label{subsection:LMD}

			In this subsection, we introduce a quantitative measurement of the `prevalence' of latent motifs in the network dictionary $W_{T}$ that we compute using NDL (see Algorithm~\ref{alg:NDL}) for a network $\G$.

				Given a network $\G$ and a $k$-chain motif, recall that the output of the NDL algorithm is a network dictionary $W_T$
				of $r$ latent motifs $\mathcal{L}_{1},\ldots,\mathcal{L}_{r}$ of size $k\times k$. Recall as well that the NDL algorithm yields data matrices $X_{1}, \ldots, X_{T}$ of size $k^{2} \times N$. Suppose that we have code matrices $H_{1}^{\star},\ldots,H_{T}^{\star}$ such that $X_{t} \approx W_{T}H_{t}^{\star}$ for all $t \in \{1,\ldots,T \}$. More precisely, we let 
				\begin{align}\label{eq:H_star}
					H_{t}^{\star} := \argmin_{H\ge 0} (\lVert X_{t} - W_{T}H \rVert^{2} + \lambda \lVert H \rVert_{1})\,, 
				\end{align}
				where we take the $\argmin$ over all $H\in \R_{\ge 0}^{k^{2}\times N}$ The columns of $H_{t}^{\star}$ encode how to nonnegatively combine the latent motifs in $W_{T}$ to approximate the mesoscale patches in $X_t \in \R_{\ge 0}^{k^{2}\times N}$,  {so the rows of $H_{t}^{\star}$ encode the linear coefficients of each latent motif in $W_{T}$ that we use to approximate the columns of $X_{t}$.} Consequently, the means of the Euclidean norms of the rows of $H_{t}$ for each $t\in \{1,\ldots,T\}$
					encode the mean prevalences in $\G$ of the latent motifs in $W_{T}$. This motivates us to consider the mean Gramian matrix~\cite{horn2012matrix}
					\begin{align}
						P_{T}^{\star} := \frac{1}{T} \sum_{t=1}^{T} H_{t}^{\star}(H_{t}^{\star})^{T}  \in \R^{r\times r}\,.
					\end{align}
					The square root of the diagonal entries of $P_{T}^{\star}$ yield the mean prevalences in $\G$ of the latent motifs in $W_{T}$. Accordingly, for each $i \in \{1,\ldots, r\}$,
									we define the \textit{dominance score} of the latent motif $\mathcal{L}_{i}$ to be $\sqrt{P_{T}^{\star}(i,i)}$.

			Computing $P_{T}^{\star}$ requires us to store the previous data matrices $X_{1},\ldots,X_{T}$ and to determine $H_{1}^{\star},\ldots,H_{T}^{\star}$ by solving \eqref{eq:H_star} for $t\in \{1,\ldots,T\}$. {This way of computing $P_{T}^{\star}$ is very expensive because of its extensive
			memory and computational requirements.} 
			To address this issue, we instead use the aggregate matrix $P_{T}$ that we compute as part of Algorithm~\ref{alg:NDL}. 
			We then do not require an
			extra computation. 
			Note that  
			\begin{align}
				P_{T} = \frac{1}{T} \sum_{t=1}^{T} H_{t}H_{t}^{T}\,,
			\end{align}
			where $H_{t}=\argmin_{H\ge 0} (\lVert X_{t} - W_{t-1}H \rVert^{2} + \lambda \lVert H \rVert_{1})\in \R_{\ge 0}^{r\times N}$ is the code matrix. {The matrix $P_{T}$ is an approximation of $P_{T}^{\star}$ because the defining equation of $H_{t}$ is the same as that of $H_{t}^{\star}$ in \eqref{eq:H_star} with $W_{T}$ replaced by $W_{t-1}$.} 
			The approximation error of using $P_{T}$ instead of $P_{T}^{\star}$
			vanishes as $T\rightarrow \infty$ under mild conditions. 
			Specifically, under the hypotheses of Theorems~\ref{thm:NDL} and~\ref{thm:NDL2}, the network dictionary $W_{t}$ converges almost surely to some limiting dictionary.
			It follows that $\lVert P_{T}^{\star} - P_{T} \rVert_{F}\rightarrow 0$ almost surely as $T \rightarrow \infty$.

			\subsection{Latent motifs of networks at various mesoscales}
			\label{subsection:various_mesoscales}

			As we discussed in Appendix~\ref{subsection:LMD}, we associate a scalar `dominance score' to each latent motif to measure its total contribution in our reconstruction of the sampled $k$-node subgraphs. 
			In Figure~\ref{fig:latent_motifs_multiscale_1}, we show the two most-dominant latent motifs (i.e., the two with the largest dominance scores) that we learn from each of 
			the example networks at various scales (specifically, for $k = 6$, $k = 11$, $k = 21$, and $k = 51$) when we use a dictionary with $r = 25$ latent motifs.

			\begin{figure*}[h]
				\centering
				\includegraphics[width= \linewidth]{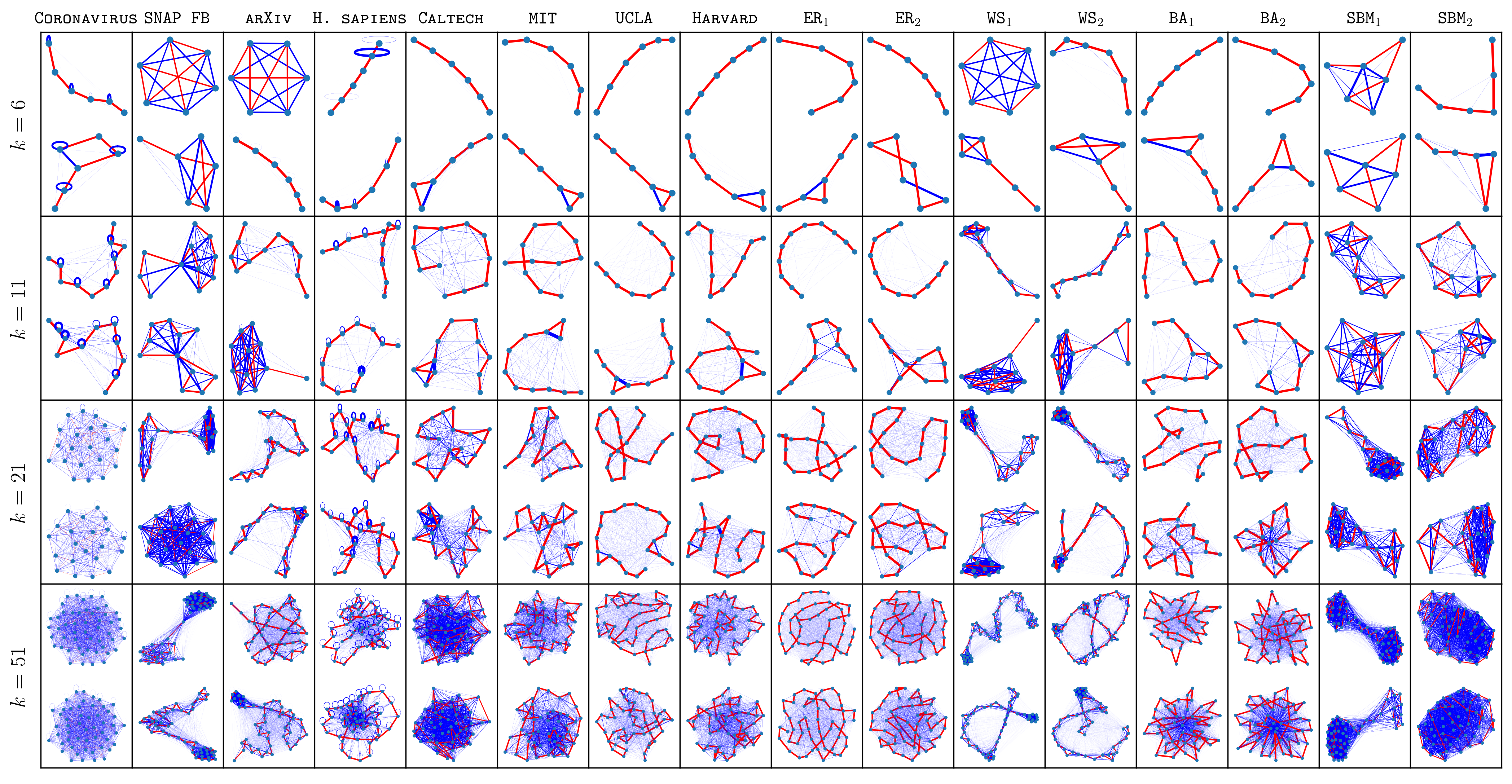}
				\vspace{-0.5cm}
				\caption{The latent motifs that we learn from our 16 networks (eight real-world networks and eight synthetic networks, which include two distinct instantiations of each of four random-graph models) at four different scales (specifically, for $k=6$, $k = 11$, $k = 21$, and $k = 51$), have distinct mesoscale structures in the networks. Using NDL, we learn network dictionaries of $r = 25$ latent motifs with $k$ nodes for each of the 16 networks.  For each network at each scale, we show the (top) first and (bottom) second most-dominant latent motif from each dictionary. See {Appendix}~\ref{subsection:LMD} 
					for details about how we measure latent-motif dominance. 
				}
				\label{fig:latent_motifs_multiscale_1}
			\end{figure*}

			For each network, as we increase the scale parameter $k$, various mesoscale structures emerge in the latent motifs in Figure~\ref{fig:latent_motifs_multiscale_1}. For instance, \dataset{SNAP FB}, \dataset{arXiv}, and \dataset{WS}$_{1}$ all have fully connected top (i.e., most-dominant) latent motifs at scale $k = 6$ but their second most-dominant latent motifs are distinct. In \dataset{SNAP FB}, \dataset{Caltech}, and \dataset{MIT} at scales $k \in \{6,11, 21\}$ and in the BA networks at scales $k\in \{11,21,51\}$, the two most-dominant latent motifs in Figure~\ref{fig:latent_motifs_multiscale_1} have nodes that are adjacent to many other nodes in the latent motif. Hubs (i.e., nodes that are adjacent to many other nodes) are characteristic of both BA networks (which have heavy-tailed degree distributions)~\cite{barabasi1999emergence} and most social networks (which typically have heavy-tailed degree distributions)~\cite{newman2018}. We also observe hubs in the network dictionaries of the latent motifs of the Facebook networks \dataset{UCLA} and  \dataset{Harvard} (see Figure~\ref{fig:all_dictionaries_1}).

			In Figure~\ref{fig:Figure_boxcompare} of the main manuscript, we saw that the community sizes (i.e., the numbers of nodes) in latent motifs reflect the community sizes of actual subgraphs in a network. 
		The type of community structure that we examine is different from typical network community structure. For example, consider the WS networks.
	The top latent motif of the network \dataset{WS}$_{1}$ at scale $k = 6$ is fully connected, but the top latent motif of \dataset{WS}$_{2}$ is not fully connected because of its larger rewiring probability. At larger scales (i.e., for larger $k$), both WS networks have latent motifs with multiple communities. The WS networks have locally densely connected nodes on a ring of nodes and random `shortcut' edges that can connect distant nodes of the ring. Therefore, when one samples a $k$-path uniformly at random, it is very likely to use at least one shortcut edge. When a $k$-path uses a shortcut edge, we expect the resulting 
induced subgraph to have two distinct densely connected communities. 
This local `community structure' in the WS networks is rather different than
standard types of community structure~\cite{porter2009communities,fortunato2016}. Although we do observe such community structure in subgraphs that are induced by $k$-paths 
(see, e.g., Figure~\ref{fig:subgraphs} of the main manuscript), this observation does not imply that the entire node set of the WS networks is partitioned into a few communities.
We also see the difference between our mesoscale structures and community structure by examining the latent motifs of the SBM networks at different scales in Figure~\ref{fig:latent_motifs_multiscale_1}. The SBM networks have three (equal-sized) communities by construction, but their latent motifs do not have three communities at any of the scales, because the
uniformly sampled $k$-paths do not always intersect with all three communities.
For example, the six $20$-paths from \dataset{SBM}$_{1}$ in Figure~\ref{fig:subgraphs} of the main manuscript intersect with only one or two of the network's planted communities.

				\subsection{Community sizes in subgraph samples and latent motifs}
				\label{sec:community_boxplot}
				
				By comparing the $21$-node latent motifs of \dataset{UCLA} and \dataset{Caltech} (we extract $r = 25$ of each) in Figure~\ref{fig:img_ntwk_dict} of the main manuscript, we observe that
				most latent motifs of \dataset{Caltech} have larger communities than those of \dataset{UCLA}. To what extent does the community structure of the latent motifs carry over to the subgraph samples in these networks? Latent motifs are $k$-node networks with nonnegative edge weights, so we can examine this question quantitatively by performing community detection using a standard approach on these $k$-node networks.

				\begin{figure*}[h]
					\centering
					\includegraphics[width= \linewidth]{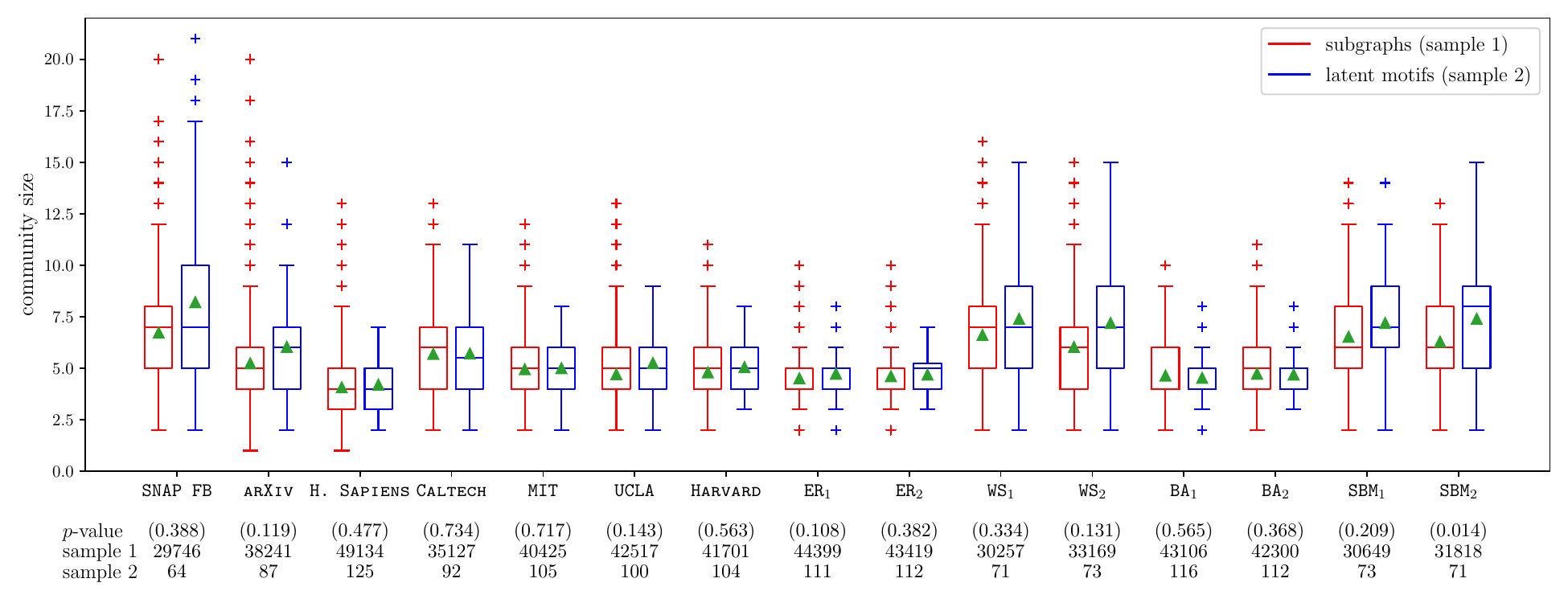}
					\vspace{-0.8cm}
					\caption{ 
						A comparison of box plots of community sizes for 10,000 
						sampled subgraphs that are induced by uniformly random paths with $k = 21$ nodes (in red) to corresponding box plots from $r = 25$ latent motifs of $k = 21$ nodes for various real-world and synthetic networks. We obtain communities of the subgraphs and latent motifs by using the Louvain modularity-maximization algorithm~\cite{blondel2008fast}. The triangles inside the boxes indicate the sample means. Under each network label, we show the $p$-value of Mood's median test \cite{brown1951median} and the number of samples in each population. 
					}
					\label{fig:Figure_boxcompare}
				\end{figure*}

				To detect communities, we use the locally greedy Louvain method for modularity maximization~\cite{blondel2008fast}. For most of our example networks, the community-size statistics of the learned latent motifs are close approximations of the corresponding statistics for the subgraph samples from the networks. 
				In Figure~\ref{fig:Figure_boxcompare}, we compare box plots of community sizes of 10,000-node subgraphs that are induced by uniformly randomly sampled $k$-paths to the corresponding box plots from community detection of $r = 25$ latent motifs of various networks. In our calculations, the median community sizes in the subgraphs and latent motifs differ by 2 for the network \dataset{SBM}$_{2}$; differ by 1 for \dataset{arXiv}, \dataset{WS}$_{2}$, and \dataset{SBM}$_{1}$; differ by $0.5$ for \dataset{Caltech}; and coincide for the other networks. Our experiments also demonstrate that there is no statistically significant difference between the medians of the two samples, except for \dataset{SBM}$_{2}$ at significance level $0.1$. See Appendix \ref{section:experimental_details} for more details.


			
			\section{Algorithm for Network Denoising and Reconstruction (NDR)}
			\label{subsection:NR}
			
			
			\subsection{Algorithm overview and statement}
			\label{subsection:NDR_1}

			The standard pipeline for image denoising and reconstruction~\cite{elad2006image, mairal2008sparse, mairal2009non} is to uniformly randomly sample a large number of $k \times k$ overlapping patches of an image and then average their associated approximations at each pixel to obtain a reconstructed version of the original image. (See the `Methods' section of the main manuscript for more details about image reconstruction.) A network analog of this pipeline proceeds as follows. Given a network $\G=(V,A)$, a $k$-chain motif 
			{$F = ([k], A_{F})$, and a network dictionary with latent motifs $(\mathcal{L}_{1},\ldots,\mathcal{L}_{r})$, we compute a weighted network $\G_{\textup{recons}}=(V,A_{\textup{recons}})$. To do this, we first uniformly randomly sample a large number $T$ of (not necessarily injective) homomorphisms $\x_{t}: F\rightarrow \G$ and determine the corresponding mesoscale patches $A_{{\x_{0}}},\ldots,A_{\x_{T}}$ using \eqref{eq:def_A_F_patch}. We then approximate each mesoscale patch $A_{\x_{t}}$ by a nonnegative linear combination $\hat{A}_{\x_{t}}$ of the latent motifs $\mathcal{L}_{i}$. Finally, for each ${x,y}\in V$, we define $A_{\textup{recons}}({x,y})$ as the mean of $\hat{A}_{\x_{t}}(a,b)$ over all $t=\{{0},\ldots,T\}$ and all $a,b\in \{1,\ldots,k\}$ such that $\x_{t}(a)={x}$ and $\x_{t}(b)={y}$.

			Our network denoising and reconstruction (NDR) algorithm (see Algorithm~\ref{alg:network_reconstruction}) uses the idea in the preceding paragraph. 
			Suppose that we have a network $\G=(V,A)$, a $k$-chain motif $F=([k],A_{F})$, and a network dictionary $W$ that consists of $r$ nonnegative $k \times k$ matrices $\mathcal{L}_{1},\ldots,\mathcal{L}_{r}$.  We provide two options to reconstruct $\G$. In one option ($\mathtt{InjHom} = \texttt{F}$), we use uniformly random homomorphisms from the distribution $\pi_{F\rightarrow \G}$ in \eqref{eq:def_embedding_F_N}. In the other option ($\mathtt{InjHom} = \texttt{T}$), we use only the injective homomorphisms, so
			 we instead use the distribution $\pi_{F\hookrightarrow \G}$ in \eqref{eq:def_embedding_F_N_inj}. The latter option has a larger computational cost, but it has better theoretical properties 
		 	for the NDL algorithm (see Algorithm~\ref{alg:NDL}). 
			 To sketch how the NDR algorithm with $\mathtt{InjHom} = \texttt{F}$ works, suppose that we sample homomorphisms $\x_{0}, \ldots,\x_{T}$ from the distribution $\pi_{F\rightarrow \G}$. For each $t\ge 0$, we approximate the mesoscale patch $A_{\x_{t}}$ (see \eqref{eq:def_A_F_patch}) by a nonnegative linear combination of latent motifs $\mathcal{L}_{i}$ and we then take the {mean} of the values of each entry $A(a,b)$ for all $t \in \{1, \ldots, T\}$. However, because sampling a homomorphism $\x_{t}:F\rightarrow \G$ from $\pi_{F\rightarrow \G}$ is not as straightforward as uniformly randomly sampling $k \times k$ patches of an image, we generate a sequence $(\x_{t})_{t\in \{0, \ldots,T\}}$ of homomorphisms using an MCMC motif-sampling algorithm (see Algorithms~\ref{alg:glauber} and~\ref{alg:pivot}). The NDR algorithm with $\mathtt{InjHom} = \texttt{T}$ works similarly, but it uses injective homomorphisms that are generated from the injective MCMC motif-sampling algorithm (see 
				Algorithm~\ref{alg:motif_inj}).

			\begin{algorithm}
				\renewcommand{\thealgorithm}{{NDR}}
				\begin{algorithmic}[1]
					\caption{\!\!. Network Denoising and Reconstruction (NDR)}\label{alg:network_reconstruction}
					
					\State \textbf{Input:} Network $\G=(V,A)$, network dictionary $W \in \mathbb{R}_{\ge 0}^{k^{2}\times r}$ 
					
					\vspace{0.1cm}
					\State \textbf{Parameters:}  {$F=([k],A_{F})$ (a {$k$-chain} motif)\,,\, $T\in \mathbb{N}$ (number of iterations)\,,\, $\lambda\ge 0$ (the coefficient of an $L_{1}$-regularizer)\,,\, $\theta\in [0,1]$ (an edge threshold) } 
					
					\vspace{0.1cm}
					\State \textbf{Options:} $\mathtt{denoising}\in \{\texttt{T}, \texttt{F}\}$\,,\,  $\mathtt{MCMC}\in \{\mathtt{Pivot},\, \mathtt{PivotApprox},\, \mathtt{Glauber} \}$\,,\, $\mathtt{InjHom}\in \{\texttt{T}, \texttt{F}\}$

					\vspace{0.1cm}
					\State \textbf{Requirement:} There exists at least one homomorphism $F\rightarrow \G$
					\vspace{0.1cm}
					\State \textbf{Initialization:} 
					\State \quad $A_{\textup{recons}}\,,\,A_{\textup{count}}:V^{2}\rightarrow \{0\}$ (matrices with $0$ entries) 
					\State \quad Sample a (not necessarily injective) homomorphism $\x_{0}:F \rightarrow \G$ using the rejection-sampling algorithm in Algorithm~\ref{alg:rejection_motif} 
					
					\vspace{0.1cm}
					\State \textbf{For $t=1,2,\ldots,T$:}
					\vspace{0.1cm}
					\State \quad \textit{MCMC update and mesoscale patch extraction}:
					\State \qquad $\x_{t}\leftarrow$ Updated homomorphism that we obtain by applying 
					\begin{align}
						\text{Algorithm~\ref{alg:pivot} with $\mathtt{AcceptProb}=\mathtt{Exact}$}  &\qquad \text{if \quad $\mathtt{MCMC}=\mathtt{Pivot}$} \\ 
						\text{Algorithm~\ref{alg:pivot} with $\mathtt{AcceptProb}=\mathtt{Approximate}$}  &\qquad \text{if \quad $\mathtt{MCMC}=\mathtt{PivotApprox}$} \\ 
						\text{Algorithm~\ref{alg:glauber} with $\mathtt{AcceptProb}=\mathtt{Glauber}$}  &\qquad \text{if \quad$\mathtt{MCMC}=\mathtt{Glauber}$} 
					\end{align}
					(If $\texttt{InjHom}=\mathtt{T}$, set $\x_{t}\leftarrow$ Updated injective homomorphism by applying Algorithm~\ref{alg:motif_inj} with the 
						specified MCMC algorithm.)
					\vspace{0.1cm}
					
					\State \qquad $A_{\mathbf{x}_{t}}\leftarrow $ $k\times k$ mesoscale patch of $\G$ that is induced by $\x_{t}$ (see \eqref{eq:def_A_F_patch})
					\State \qquad  $X_{t}\leftarrow $  $k^{2}\times 1$ matrix that we obtain by vectorizing $A_{\mathbf{x}_{t}}$ (using Algorithm~\ref{alg:vectorize})
					
					\State \quad \textit{Mesoscale reconstruction}: 
					\State \quad 	\vspace{-0.3cm}
					\begin{align}
						\begin{cases}
							\text{$\widetilde{X}_{t}\leftarrow X_{t}$ and $\widetilde{W}\leftarrow W$} & \text{if\quad $\mathtt{denoising}=\texttt{F}$} \\
							\text{$\widetilde{X}_{t}\leftarrow (X_{t})_{\textup{off}}$ and $\widetilde{W}\leftarrow (W)_{\textup{off}}$ using Algorithm~\ref{alg:off_chain}} & \text{if\quad $\mathtt{denoising}=\texttt{T}$}
						\end{cases}
					\end{align}
					\vspace{0.1cm}
					\State \qquad $H_{t}\leftarrow \argmin\limits_{H\in \mathbb{R}_{\ge 0}^{r\times 1}} (\|\widetilde{X}_{t} - \widetilde{W} H\|_F^2 + \lambda\lVert H\rVert_1)$ and $\hat{X_{t}}\leftarrow \widetilde{W}H_{t}$ 
					
					\State \label{line:mesoscale_recons} \!\!\qquad $\hat{A}_{\mathbf{x}_{t};W}\leftarrow$  $k\times k$ matrix that we obtain by reshaping the $k^{2}\times 1$ matrix $\hat{X}_{t}$ using Algorithm~\ref{alg:reshape}  
					\vspace{0.1cm}
					\State \quad \textit{Update 
						{reconstruction}:}
					\State \qquad \textbf{For} $a,b \in \{1, \ldots, k\}$: 
					\State \qquad \quad \textbf{If}  $ \left( \text{$\texttt{denoising}=\texttt{F}$ or $A_{F}(a,b)=0$} \right)$:
					\begin{align}\label{eq:NR_update_step}
						A_{\textup{count}}(\x_{t}(a),\x_{t}(b))&\leftarrow A_{\textup{count}}(\x_{t}(a),\x_{t}(b))+1 \\
						j&\leftarrow A_{\textup{count}}(\x_{t}(a),\x_{t}(b)) \\
						A_{\textup{recons}}(\x_{t}(a),\x_{t}(b))&\leftarrow (1-j^{-1})A_{\textup{recons}}(\x_{t}(a),\x_{t}(b)) + j^{-1} \hat{A}_{\x_{t};W}(\x_{t}(a),\x_{t}(b))
					\end{align}
					
					\State \textbf{Output:} Reconstructed network $\G_{\textup{recons}} = (V, A_{\textup{recons}})$ 
				\end{algorithmic}
			\end{algorithm}

			\begin{algorithm}
				\renewcommand{\thealgorithm}{2a}
				\begin{algorithmic}
					\caption{\!\!. Off-Chain Projection}\label{alg:off_chain}
					\State \textbf{Input:} Matrix $Y \in \mathbb{R}^{k^{2}\times m}$\,,\, {$k$-chain} motif $F=([k], A_{F})$
					
					\vspace{0.1cm}
					\State \textbf{Do:} Let $Y'$ be a $k\times k\times m$ tensor that we obtain by reshaping each column of $Y$ using Algorithm~\ref{alg:reshape}  
					
					\State \qquad Let $Y''$ be a $k\times k\times m$ tensor that we obtain from $Y'$ by calculating
					\begin{align}
						Y''(a,b,c) = Y'(a,b,c)\one(A_{F}(a,b)=0) \quad \text{for all} \quad a,b \in \{1, \ldots, k\}\,\, \text{and} \,\, c \in \{1, \ldots, m\}
					\end{align}
					\State \qquad Let $Y_{\textup{off}}$ be a $k^{2}\times m$ matrix that we obtain from $Y''$ by vectorizing each of its slices using Algorithm~\ref{alg:vectorize}: $Y''[:,:,c]$ for all $c \in \{1, \ldots, m\}$
					
					\State \textbf{Output:} Matrix $(Y)_{\textup{off}}\in \mathbb{R}^{k^{2}\times m}$
				\end{algorithmic}
			\end{algorithm}

				For network reconstruction, it is important to sample homomorphisms $\x_{1},\ldots,\x_{T}:F\rightarrow \G$ that cover an entire network $\G$ (or at least a large portion of it). {{A node $x$ of $G$ is} `covered' by the homomorphisms $\x_{1},\ldots,\x_{T}$ if it is contained in the image of $\x_{t}$ for some $t \in \{1,\ldots,T\}$.} In Proposition~\ref{prop:inj_hom_coverage}, we show that one can cover all nodes of $\G$ by the images of injective homomorphisms $F\hookrightarrow \G$ if $2(k-1)\le \textup{diam}(\G)$ if $\G$ is symmetric and connected. However, even when this inequality is satisfied, we have to sample more homomorphisms using one of our MCMC motif-sampling algorithms (see Algorithms~\ref{alg:glauber} and~\ref{alg:pivot}) to cover the same portion of the network $\G$ than when we use all sampled homomorphisms. This gives a computational advantage to using $\mathtt{InjHom} = \texttt{F}$ instead of $\mathtt{InjHom} = \texttt{T}$ in our NDR algorithm.

				Despite the 	computational disadvantage of using $\mathtt{InjHom} = \texttt{T}$, this choice has a nice theoretical advantage. Recall that the NDL algorithm (see Algorithm~\ref{alg:NDL}) computes latent motifs $\mathcal{L}_{1},\ldots,\mathcal{L}_{r}$ from mesoscale patches $A_{\y_{1}},\ldots, A_{\y_{M}}$ of $\G$ for injective homomorphisms $\y_{t}:F\rightarrow \G$  for $t \in \{1,\dots, M\}$
				such that $\mathcal{L}_{1},\ldots,\mathcal{L}_{r}$
				give an approximate solution of \eqref{eq:NDL_exact}. 
				Consequently, when linearly approximating a mesoscale patch $A_{\x}$ of $\G$, these latent motifs $\mathcal{L}_{1},\dots,\mathcal{L}_{r}$ are less effective if the homomorphism $\x$ is non-injective than if $\x$ is injective.  We need to linearly approximate multiple mesoscale patches $A_{\x_{1}},\ldots,A_{\x_{T}}$ for homomorphisms $\x_{t}:F\rightarrow \G$ for $t \in \{1,\dots, T\}$, so we expect the reconstructed network that we obtain using Algorithm~\ref{alg:network_reconstruction} with only injective homomorphisms to be more accurate than when using all sampled homomorphisms. In Theorem~\ref{thm:NR}\textbf{(iii)}, we obtain an upper bound for the Jaccard reconstruction error (which we define in \eqref{eq:JD}). This upper bound is optimized by using the latent motifs that we obtain with the NDL algorithm using only injective homomorphisms. 
				Therefore, the NDR algorithm with the option $\mathtt{InjHom} = \texttt{T}$ has a theoretical advantage over the NDR algorithm with the option $\mathtt{InjHom} = \texttt{F}$.
			
			As in the first line of \eqref{eq:def_ONMF}, the problem of determining $H_{t}$ in line 15 of Algorithm~\ref{alg:network_reconstruction} is a standard convex problem, which one can solve by using Algorithm~\ref{algorithm:spaser_coding}. There are two variants of the NDR algorithm. The variant is specified by the Boolean variable $\mathtt{denoising}$. The NDR algorithm with $\mathtt{denoising} = \texttt{F}$ is identical to the network-reconstruction algorithm in~\cite{lyu2020online}, except for the thresholding step. The NDR algorithm with $\mathtt{denoising} = \texttt{T}$ is a new variant of NDR that we present in this paper for network denoising.
			

			\subsection{Further discussion of the denoising variant of the NDR algorithm}
			\label{subsection:denoising}

			We now give a detailed discussion of Algorithm~\ref{alg:network_reconstruction} with $\texttt{denoising}=\texttt{T}$ for {network-denoising applications}. Recall that our network-denoising problem  is to reconstruct a true network $\G_{\textup{true}}=(V,A)$ from an observed network $\G_{\textup{obs}}=(V,A')$. The scheme that we used to produce Figure~\ref{fig:Figure4} is the following:
			\begin{enumerate}[label=\textbf{D.\arabic*}]
				\item Learn a network dictionary $W\in \R_{\ge 0}^{k^{2}\times r}$ from an observed network $\G_{\textup{obs}}=(V,A)$ using NDL (see Algorithm~\ref{alg:NDL}). \label{eq:NDL_classification_scheme0}
				
				\item Compute a reconstructed network $\G_{\textup{recons}}=(V,A_{\textup{recons}})$ using NDR (see Algorithm~\ref{alg:network_reconstruction}) with inputs $\G_{\textup{obs}}=(V,A)$ and $W$.  \label{eq:NDL_classification_scheme1}
				
				\item \label{classification_scheme}  Fix an edge threshold $\theta\in [0,1]$. If $\G_{\textup{obs}}$ is $\G_{\textup{true}}$ with additive (respectively, subtractive) noise, we classify each edge (respectively, nonedge) {$\{x,y\}$} as `positive' if and only if  $A_{\textup{recons}}(x,y)>\theta$.  \label{eq:NDL_classification_scheme2}
			\end{enumerate}

			As was discussed in~\cite[Remark 4]{lyu2020online}, a limitation of using NDR with $\texttt{denoising}=\texttt{F}$ for network denoising is that the meaning of successful classification for subtractive-noise cases is an `inversion' of its meaning for additive-noise cases. Specifically, Lyu et al.~\cite{lyu2020online} used network denoising with NDR for additive noise, but it is necessary to classify each nonedge $\{x,y\}$ as `positive' if $A_{\textup{recons}}(x,y) < \theta$. In the experiment in Figure~\ref{fig:denoising_caltech_comparison} of the main manuscript, we noted that NDR with $\texttt{denoising} = \texttt{F}$ may assign large weights to false edges and small weights to true edges. Our NDR algorithm with $\texttt{denoising} = \texttt{T}$ addresses this directionality issue and allows us to use the unified classification scheme above for both additive and subtractive noise.

			Using NDR with $\texttt{denoising} = \texttt{T}$ allows one to handle an issue that occurs when denoising additive noise for sparse real-world networks but does not arise in the image-denoising setting. Suppose that we obtain $\G_{\textup{obs}}$ by adding 
			false edges to a sparse unweighted network $\G_{\textup{true}}$. The on-chain entries of the mesoscale patches $A_{\x}$ are always equal to $1$. Therefore, the latent motifs that we learn from $\G_{\textup{obs}}$ have the same on-chain entries. (See, e.g., Figure~\ref{fig:img_ntwk_dict} of the main manuscript.) 
			Consequently, {linearly approximating} the mesoscale patches $A_{\x}$ of $\G_{\textup{obs}}$ {using} the latent motifs that we learn from $\G_{\textup{obs}}$ cannot distinguish between true and false on-chain entries. 
			Furthermore, because $\G_{\textup{obs}}$ is sparse, there are many fewer positive off-chain entries of $A_{\x}$ than on-chain entries of $A_{\x}$. Therefore, in a network reconstruction, linear approximations of $A_{\x}$ that use the latent motifs are likely to assign larger weights to on-chain entries of $A_{\x}$ than to off-chain entries. 
			The resulting reconstruction of $\G_{\textup{obs}}$ is thus similar to $\G_{\textup{obs}}$, and it is very hard to detect false edges in $\G_{\textup{obs}}$. Using the option $\texttt{denoising} = \texttt{T}$ prevents this issue by ignoring all on-chain entries both for each sampled mesoscale patch $A_{\x}$ and for each latent motif in the network dictionary $W$ that we use for denoising.

			
			\section{Experimental Details}
			\label{section:experimental_details}
			

			\subsection{Figure~\ref{fig:anomaly_detection}}\label{subsection:anomaly_detection}

			The network in Figure \ref{fig:anomaly_detection}{\textbf{a}} has 100 nodes and 216 edges. 
			Of these edges, 164 are from the original network (see Figure \ref{fig:anomaly_detection}\textbf{b}) and the remaining 51 are anomalous edges (see Figure \ref{fig:anomaly_detection}{\textbf{c}}) that we generate using the $G(N,p)$ Erd\H{o}s--R\'{e}nyi (ER) network model with the same node set and an edge probability of $p = 0.01$. 
			Specifically, we independently connect each pair of nonadjacent nodes by an anomalous edge with probability $0.01$. We learn the $r = 25$ latent motifs in Figure \ref{fig:anomaly_detection}{\textbf{d}}
			using the NDL algorithm (see Algorithm~\ref{alg:NDL}) for a $k$-chain motif $F=([k], A_{F})$ at scale $k = 9$ for
			 $T=50$ iterations,
			$N = 100$ injective homomorphisms per iteration  (so each iteration consists of sampling $N$ injective homomorphisms and applying the online NMF update \eqref{eq:def_ONMF}), an $L_{1}$-regularizer with coefficient $\lambda=1$, and the MCMC motif-sampling algorithm $\mathtt{MCMC}=\mathtt{PivotApprox}$.
			

			To compute the weighted-network reconstruction of the observed network in Figure \ref{fig:anomaly_detection}{\textbf{e}}, we use the NDR algorithm (see Algorithm~\ref{alg:network_reconstruction}) with a $k$-chain motif $F=([k], A_{F})$ at scale $k = 9$ for $T=10^{5}$ iterations, 
			the 25 latent motifs in Figure \ref{fig:anomaly_detection}{\textbf{d}}, an $L_{1}$-regularizer with coefficient $\lambda=0$ (i.e., no regularization), the MCMC motif-sampling algorithm $\mathtt{MCMC}=\mathtt{PivotApprox}$, and $\mathtt{denoising}=\texttt{F}$.

		To evaluate our results, we split the data into training and test sets, with 50$\%$ of the 164 true edges and 50$\%$ of the 51 anomalous edges in each set. 
			 To maximize classification accuracy, we then determine {an} optimal threshold value $\theta$ to weight the edges in the reconstruction network. Specifically, we classify all edges in the test set as positive if their weight in the reconstruction network exceeded $\theta$, and otherwise negative. In Figure \ref{fig:anomaly_detection}{\textbf{f}}, we show all edges in the weighted reconstruction in Figure \ref{fig:anomaly_detection}{\textbf{e}}} whose weights are at most $\theta$.


			\subsection{Figure~\ref{fig:subgraphs}}\label{subsection:fig_subgraph}

			In Figure~\ref{fig:subgraphs}, we show six subgraphs that are induced by approximately uniformly random samples of $k$-paths with $k = 20$ (red edges) from the networks \dataset{Caltech}, \dataset{UCLA}, \dataset{ER}$_{1}$, \dataset{BA}$_{2}$, \dataset{WS}$_{2}$, and \dataset{SBM}$_{1}$.  
			To sample such $k$-paths, we use Algorithm~\ref{alg:pivot} with $\mathtt{AcceptProb} = \mathtt{Approximate}$. The red edges in each subgraph designate edges in the sampled $k$-paths (i.e., on-chain edges), and the blue edges designate edges (the off-chain edges) that connect nonadjacent nodes on sampled paths. 
			

			\subsection{Figure~\ref{fig:recons_illustration}}\label{subsection:recons_illustration}

			In Figure~\ref{fig:recons_illustration}, we illustrate our low-rank network-reconstruction process using two sets of latent motifs. 
			For both sets of latent motifs 
			to compute the weighted reconstructions in \textbf{a4} and \textbf{b2}, we use 2500 $k$-paths (with $k = 7$) that we sample using the MCMC motif-sampling algorithm $\mathtt{MCMC} = \mathtt{PivotApprox}$.

			
			
			\subsection{Figure~\ref{fig:img_ntwk_dict}}\label{subsection:F1}
			
			In Figure~\ref{fig:img_ntwk_dict}, we illustrate latent motifs that we learn from the networks \dataset{UCLA} and \dataset{Caltech} and we compare these latent motifs to the elements of an image dictionary. The image in Figure~\ref{fig:img_ntwk_dict}\textbf{a} is from the collection \dataset{Die Graphik Ernst Ludwig Kirchners bis 1924, von Gustav Schiefler Band I bis 1916} (Accession Number 2007.141.9, Ernst Ludwig Kirchner, 1926). We use this image with permission from the National Gallery of Art in Washington, DC, USA.\footnote{{See \url{https://www.nga.gov/notices/open-access-policy.html} for the open-access policy of the National Gallery of Art.}} We use a $k$-chain {motif} 
				$F=([k], A_{F})$ and a scale $k=21$ for $T=100$ iterations, $N=100$ injective homomorphisms per iteration, $r=25$ latent motifs, an $L_{1}$-regularizer with coefficient $\lambda = 1$, and the MCMC motif-sampling algorithm $\mathtt{MCMC} = \mathtt{PivotApprox}$. The image dictionary for the artwork \dataset{Cycle} in Figure~\ref{fig:img_ntwk_dict} uses an algorithm that is similar to Algorithm~\ref{alg:NDL}, except that we uniformly randomly sample $21 \times 21$ square patches of the image instead of $k \times k$ mesoscale patches of a network.


			
			\subsection{Figure~\ref{fig:network_recons}}
			\label{subsection:Fig3_SI}

			To generate Figure~\ref{fig:network_recons}, we first apply the NDL algorithm (see Algorithm~\ref{alg:NDL}) to each network in the figure to learn $r = 25$ latent motifs for a $k$-chain motif
			$F=([k], A_{F})$ at scale $k = 21$ for $T=100$ iterations, $N = 100$ injective homomorphisms per iteration, an $L_{1}$-regularizer with coefficient $\lambda=1$, and the MCMC motif-sampling algorithm $\mathtt{MCMC}=\mathtt{PivotApprox}$. For each self-reconstruction $X \leftarrow X$ (see the caption of Figure~\ref{fig:network_recons}), we apply the NDR algorithm (see Algorithm~\ref{alg:network_reconstruction}) to a $k$-chain {motif} 
			$F=([k], A_{F})$ at scale $k = 21$ for $T = \lfloor n\ln n \rfloor$ iterations (where $n$ is the number of nodes of the network), $r = 25$ latent motifs, an $L_{1}$-regularizer with coefficient $\lambda=0$ (i.e., no regularization), the MCMC motif-sampling algorithm $\mathtt{MCMC}=\mathtt{PivotApprox}$, and $\mathtt{denoising}=\texttt{F}$. For each cross-reconstruction $Y \leftarrow X$ (see the caption of Figure~\ref{fig:network_recons}), we apply the NDR algorithm (see Algorithm~\ref{alg:network_reconstruction}) to a $k$-chain motif with the corresponding network $F=([k], A_{F})$ at scale $k = 21$ for $T =(1 + 3 \cdot \one(n<1000)) \lfloor n\ln n \rfloor$ iterations (where $n$ is the number of nodes of the network), the edge threshold $\theta = 0.4$, an $L_{1}$-regularizer with coefficient $\lambda=1$, the MCMC motif-sampling algorithm $\mathtt{MCMC}=\mathtt{PivotApprox}$, and $\mathtt{denoising}=\mathtt{InjHom}=\mathtt{F}$. We use several choices of the number $r$ of latent motifs; we indicate them in the caption of Figure~\ref{fig:network_recons}.
		
	In the main manuscript, we made several claims based on the reconstruction accuracies in Figure~\ref{fig:network_recons} in conjunction with the latent motifs in Figures~\ref{fig:img_ntwk_dict},~\ref{fig:latent_motifs_multiscale_1},~\ref{fig:all_dictionaries_1}, and~\ref{fig:all_dictionaries_3}. We now justify these claims.

\begin{description}[leftmargin=0.7cm, topsep=0.2pt, itemsep=0.1cm]
\item[(1)] The mesoscale structures of the network \dataset{Caltech} are rather different than those of
\dataset{Harvard}, \dataset{UCLA}, and \dataset{MIT} at scale $k = 21$.
\begin{itemize}
	\item In Figure~\ref{fig:network_recons}\textbf{c}, we observe that the accuracy of the cross-reconstruction 
	$X\leftarrow Y$
	is consistently higher for $X\in \{\dataset{UCLA},\, \dataset{Harvard},\,  \dataset{MIT}\}$ than for $X=\dataset{Caltech}$ for all values of $r$. For instance, for $r=9$, we can reconstruct \dataset{UCLA} with more than 90\% accuracy and we can reconstruct \dataset{Harvard} and \dataset{MIT} with more than 80\% accuracy. However, the latent motifs that we learn from \dataset{Caltech} for $r = 9$ gives only about $80\%$ accuracy for reconstructing $\dataset{UCLA}$ and only about $70\%$ accuracy for reconstructing $\dataset{MIT}$ and $\dataset{UCLA}$. This indicates that the mesoscale structures  of $\dataset{Caltech}$ differ significantly from those of
	the other three universities' Facebook networks at scale $k = 21$. Indeed, from Figures~\ref{fig:img_ntwk_dict},~\ref{fig:latent_motifs_multiscale_1}, and~\ref{fig:all_dictionaries_1}, we see that the $r = 25$ latent motifs of \dataset{Caltech} at scale $k = 21$ have larger off-chain entries than those of   \dataset{UCLA}, \dataset{MIT}, and \dataset{Harvard}.
\end{itemize}
\vspace{0.1cm}
\item[(2)] The mesoscale structures of \dataset{Caltech} at scale $k = 21$ are higher-dimensional than those
of the other three 
universities' Facebook networks.
\begin{itemize}
	\item {Consider the cross-reconstructions 
		$\dataset{Caltech} \leftarrow Y$ for $Y \in \{\dataset{UCLA},\, \dataset{Harvard},\, \dataset{MIT}\}$
		in Figure~\ref{fig:network_recons}\textbf{b}.} 
	With $r = 9$, the latent motifs that we
	learn from \dataset{Caltech} have accuracies
	as low as $64$\%. By contrast, the accuracies are
	$80$\% or higher for the self-reconstructions $X \leftarrow X$ for the Facebook networks of the other universities. In other words, $r = 9$ latent motifs at scale $k = 21$ do not approximate the mesoscale structures of \dataset{Caltech} as well as those of the other three universities' Facebook networks. This indicates that the dimension of the mesoscale structures of \dataset{Caltech} at scale $k = 21$ is larger than those of the other three universities' Facebook networks.
\end{itemize}

\vspace{0.1cm}
\item[(3)] The network $\dataset{BA}_{2}$ is better than the networks $\dataset{ER}_{2}$, $\dataset{WS}_{2}$, and $\dataset{SBM}_{2}$ at capturing the mesoscale structures of \dataset{MIT}, \dataset{Harvard}, and \dataset{UCLA} at scale $k = 21$.  However, {for $r \in \{ 9,15,25,49\}$, the network} \dataset{SBM}$_{2}$ captures the mesoscale structures of \dataset{Caltech} better than all but one of the seven other networks in Figure~\ref{fig:network_recons}\textbf{b}. 
(The only exception is \dataset{Caltech} itself.)  
\begin{itemize}
	\item From the reconstruction accuracies for  
	$X\leftarrow Y$
	in Figures~\ref{fig:network_recons}\textbf{b},\textbf{c}, where 
	$Y$ is one of the {four synthetic networks (\dataset{ER$_2$}, \dataset{WS$_2$}, \dataset{BA$_2$}, and \dataset{SBM$_{2}$})}, we observe that the two BA networks have higher accuracies then the networks from the ER and WS models for $Y\in \{\dataset{UCLA}, \dataset{Harvard}, \dataset{MIT} \}$. This suggests that the mesoscale structures of $\dataset{UCLA}, \dataset{Harvard}$, and $\dataset{MIT}$ are more similar in some respects to those of {$\dataset{BA}_{2}$ than to those of $\dataset{ER}_{2}$,  $\dataset{WS}_{2}$, and $\dataset{SBM}_{2}$.} The latent motifs of \dataset{BA$_2$} in Figures~\ref{fig:latent_motifs_multiscale_1} and \ref{fig:all_dictionaries_3} at the scales $k \in \{11, 21\}$ have characteristics that we also observe in \dataset{UCLA}, \dataset{Harvard}, and \dataset{MIT}. (Specifically, they have nodes that are adjacent to many other nodes and off-chain entries that are much smaller --- and hence in lighter shades --- than the on-chain entries.) 
	By contrast, in Figure~\ref{fig:all_dictionaries_4}, we see that the latent motifs of \dataset{ER$_2$} have sparse but seemingly randomly distributed off-chain connections and that the 
	latent motifs for \dataset{WS$_2$} have strongly interconnected communities of about $10$ nodes. 
	These patterns differ from the ones that we observe in the latent motifs for \dataset{UCLA}, \dataset{MIT}, and \dataset{Harvard} (see Figure~\ref{fig:all_dictionaries_1}).  
	For the claim about \dataset{SBM}$_{2}$, observe that the cross-reconstruction accuracy of $\dataset{Caltech} \leftarrow \dataset{SBM}_{2}$ in Figure~\ref{fig:network_recons} 
	is larger than those
	of all other reconstructions $\dataset{Caltech} \leftarrow {Y}$ except ${Y} = \dataset{Caltech}$. 
	Additionally, the theoretical lower bound of the Jaccard reconstruction accuracy for $\dataset{Caltech} \leftarrow \dataset{SBM}_{2}$ in Figure~\ref{fig:recons_bd_plot}\textbf{b}
	is larger than
	{the corresponding lower bounds for} all other
	$\dataset{Caltech} \leftarrow {Y}$ except ${Y} = \dataset{Caltech}$. 
\end{itemize}

\vspace{0.1cm}
\item[(4)] If we uniformly randomly sample a path with $k = 21$ nodes, we are more likely to obtain communities with $10$ or more nodes in an associated induced subgraph for \dataset{Caltech} than for \dataset{UCLA}, \dataset{Harvard}, and \dataset{MIT}.
\begin{itemize}
	\item This observation manifests
	directly in the box plots in Figure~\ref{fig:Figure_boxcompare} for the community sizes in the latent motifs and subgraphs that are induced by 
	$k$-paths. We can also indirectly justify this observation. From the reconstruction accuracies for 
	{$X\leftarrow Y$}
	in Figures~\ref{fig:network_recons}\textbf{b}--\textbf{e}, where  
	{$Y$}
	is one of the four synthetic networks (\dataset{ER$_2$}, \dataset{WS$_2$}, \dataset{BA$_2$}, and \dataset{SBM$_{2}$}), we observe that \dataset{WS$_2$} is better than the BA and ER networks at reconstructing \dataset{Caltech} but that it is one of the worst-performing networks for reconstructing the Facebook networks of the other three universities.
	In other words, the nonnegative linear combinations of the latent motifs of \dataset{WS$_{2}$}
	better approximate the mesoscale patches of \dataset{Caltech} than the mesoscale patches of \dataset{UCLA}, \dataset{Harvard}, and \dataset{MIT}. Recall that most latent motifs of \dataset{WS$_{2}$} at scale $k = 21$ have 
	communities with $10$ or more nodes. It seems that this community structure is more likely to occur in subgraphs that are induced by uniformly random samples of $k$-paths
	in \dataset{Caltech} with $k = 21$ nodes than from such samples in \dataset{UCLA}, \dataset{Harvard}, or \dataset{MIT}. 
\end{itemize}

\end{description}


			\subsection{Figure~\ref{fig:deg_dist_recons_plot}}

				In Figure~\ref{fig:deg_dist_recons_plot}, we compare the degree distributions 
				and the mean local clustering coefficients 
				of the original and the reconstructed networks that use $r$ latent motifs at scale $k = 21$. We conduct this experiment for the five networks in Figure~\ref{fig:network_recons}\textbf{a}. 
				In Figure~\ref{fig:deg_dist_recons_plot}\textbf{a}, we use the 
				unweighted reconstructed networks for \dataset{Caltech} with $r \in \{9,16,25,64\}$ latent motifs that we used to compute the self-reconstruction accuracies in~\ref{fig:network_recons}\textbf{b}. 
				In Figures~\ref{fig:deg_dist_recons_plot}\textbf{b}--\textbf{e}, we use the unweighted reconstructed networks for \dataset{Coronavirus}, \dataset{H. sapiens}, \dataset{SNAP FB}, and \dataset{arXiv} with $r = 25$ latent motifs that we used to compute the self-reconstruction accuracies in Figures~\ref{fig:network_recons}\textbf{b}--\textbf{e}.


			\subsection{Figure~\ref{fig:recons_bd_plot}}
			
				For each experiment 
				{$X\leftarrow Y$}
				in Figures~\ref{fig:recons_bd_plot}\textbf{b}--\textbf{e}, we plot
				\begin{align}\label{eq:fig_bd_plot}
					1 - \frac{ \E_{\x\sim \pi}[ \lVert A_{\x} - \hat{A}_{\x; W} \rVert_{1}  ]  }{2(k-1)}\,,
				\end{align}
				where $k = 21$ and $W$ is the network dictionary of $r = 25$ latent motifs in network 
				{$Y$}
				that we determine using our NDL algorithm (see Algorithm~\ref{alg:NDL}) 
				for a $k$-chain {motif} 
				$F=([k], A_{F})$ for $T=100$ iterations, $N = 100$ injective homomorphisms per iteration, an $L_{1}$-regularizer with coefficient $\lambda = 1$, and the MCMC motif-sampling algorithm $\mathtt{MCMC} = \mathtt{PivotApprox}$. The distribution $\pi$ is the stationary distribution $\hat{\pi}_{F\hookrightarrow \G}$ (see \eqref{eq:approx_pivot_statioanary_dist_inj}) of the injective {MCMC} motif-sampling algorithm (see Algorithm~\ref{alg:motif_inj} and Proposition~\ref{prop:inj_hom_coverage}). 
				For each mesoscale patch $A_{\x}\in \R^{k\times k}$, we compute the linear approximation $\hat{A}_{\x;W}$ (see line 16 of 
				{Algorithm~\ref{alg:network_reconstruction}}). 
				We approximate the expectation in the numerator of \eqref{eq:fig_bd_plot} using a Monte Carlo method.
				Specifically, from the convergence result in Proposition~\ref{prop:inj_motif_conv}, we have 
				\begin{align}\label{eq:monte_carlo_limit}
					\E_{\x\sim \pi}\left[ \lVert A_{\x} - \hat{A}_{\x; W} \rVert_{1}  \right] = \lim_{N\rightarrow \infty} \frac{1}{N} \sum_{t=1}^{N} \lVert A_{\x_{t}} - \hat{A}_{\x_{t}; W} \rVert_{1}\,, 
				\end{align}
				where $(\x_{t})_{t\ge 0}$ is a sequence of injective homomorphisms $F\hookrightarrow \G$ {that we sample} using the injective MCMC motif-sampling algorithm (see Algorithm~\ref{alg:motif_inj}). {We use the finite sample mean  $\frac{1}{N} \sum_{t=1}^{N} \lVert A_{\x_{t}} - \hat{A}_{\x_{t}; W} \rVert_{1}$ {with} $N=10^{4}$ as a proxy of the expectation in the left-hand side {of} \eqref{eq:monte_carlo_limit}. }

			\subsection{Figure~\ref{fig:Figure4}}
			\label{subsection:fig_denoising_AUC}
			
			{\color{black}
				
				To generate Figure~\ref{fig:Figure4}, we first apply the NDL algorithm (see Algorithm~\ref{alg:NDL}) to each corrupted network in the figure to learn $r = 25$ latent motifs for a {$k$-chain motif} 
				$F=([k], A_{F})$ at scale $k = 21$ for $T = 400$ iterations, $N=\text{1,000}$ homomorphisms per iteration, an $L_{1}$-regularizer with coefficient $\lambda=1$, and the MCMC motif-sampling algorithm $\mathtt{MCMC}=\mathtt{PivotApprox}$. The NDR algorithm (see Algorithm~\ref{alg:network_reconstruction}) that we use to generate the results in Figure~\ref{fig:Figure4}
				uses {{$r \in \{2, 25\}$}} latent motifs for a $k$-chain motif with the corresponding network
				$F=([k], A_{F})$ at scale $k = 21$ for $T = 4 \times 10^5$ iterations for \dataset{H. sapiens} and $T = 2 \times 10^5$ iterations for all other networks, an $L_{1}$-regularizer with coefficient $\lambda = 1$, the MCMC motif-sampling algorithm $\mathtt{MCMC}=\mathtt{PivotApprox}$, {$\mathtt{InjHom}=\texttt{F}$, and $\mathtt{denoising}\in \{ \texttt{T}, \texttt{F} \}$.}

				For Figure~\ref{fig:Figure4}, we do not conduct the denoising experiment for \dataset{Coronavirus PPI} with $-50\%$ noise because the resulting network (with 1,536 nodes and 1,232 edges) cannot be connected. (To be connected, its spanning trees need to have 1,535 edges.)

				We implement several existing network-denoising methods --- the {\sc Jaccard index}, {\sc preferential attachment}, the {\sc Adamic--Adar index}, {\sc spectral embedding}, {\sc DeepWalk}, and {\sc node2vec} --- and compare the performance of our method to those of these existing approaches. 
				Let $\G = (V,A)$ be an original network and let $\G' = (V,A)$ be the associated corrupted network. For our experiments in Figure~\ref{fig:Figure4}, both of these networks are undirected and unweighted. In all cases, we obtain a network $\hat{\G} = (V,\hat{A})$ from $\G' = (V,A)$ without using
				$\G$. For each $x,y \in V$, we compute the `confidence score' $\hat{A}(x,y)$ that the node pair $(x,y)$ is an edge $\{x,y\}$ in the original 
				network $\G$. Let $N(x)$ denote the set of neighbors of node $x$ in $\G'$. (This set includes $x$ itself when there is a self-edge at $x$.)
				For the {\sc {Jaccard index}}, {\sc preferential attachment}, and the {\sc Adamic--Adar index}, we compute the confidence score $\hat{A}(x,y)$ 
				by calculating $|N(x)\cap N(y)|/ | N(x)\cup N(y)|$, $ |N(x)| \cdot |N(y)| $, and $\sum_{z \in N(x)\cap N(y)} 1/\ln |N(z)|$, respectively.
				The {\sc Adamic--Adar index} is not defined for nodes with self-edges, and the networks \dataset{arXiv}, \dataset{Coronavirus}, and \dataset{H. sapiens} have self-edges. Therefore, we do not include self-edges in the network-denoising experiments in Figure~\ref{fig:Figure4} (but we do not remove self-edges for any other experiment).
				
				For {\sc spectral embedding}, {\sc DeepWalk}, and {\sc node2vec}, we first obtain a $128$-dimensional vector representation of the nodes of a network; this is a so-called `node embedding' of the network. 
				We then use this node embedding to obtain vector representations of the edges using binary operations. (We use the Hadamard product; see~\cite{grover2016node2vec} for details.) 
				We then use logistic regression (but one can alternatively employ some other algorithm for binary classification)
				to attempt to detect the false edges. \textsc{Spectral Clustering} uses the top $128$ eigenvectors of the combinatorial Laplacian matrix of $\G'$ to learn vector embeddings of the nodes. (We obtain the combinatorial Laplacian matrix of $\G'$ by subtracting its adjacency matrix from the diagonal matrix of node degrees.)
				See~\cite{tang2011leveraging} for details. {\sc DeepWalk} and {\sc node2vec} first sample sequences of random walks on $\G'$ and then apply the popular
				word-embedding algorithm {\sc word2vec}~\cite{mikolov2013distributed}. In {\sc DeepWalk}, each random walk is a standard random walk on the network $\G'$.
				For {\sc node2vec}, we use the 16 choices of the `return parameter' $p$ and the `in--out parameter' $q$ with $(p,q) \in \{0.25, 0.5, 1, 2\}$.
				{In addition to the random walk-sampling in these two methods, we use the following common choices ({which are the same ones that} were made in~\cite{grover2016node2vec}) for both methods{. Using each node of $\G'$
						as a starting point, we independently sample 10 random walks of 80 steps, a 
						context window size of 10, one stochastic-gradient epoch (i.e., 1 pass through the training data), and 8 workers (i.e., 8 parallel threads for training).
						See~\cite{grover2016node2vec} for details.

						For all approaches, we first split the data into a
						25/25/50 split of training/validation/test sets. We then construct the network $\hat{\G} = (V,\hat{A})$.
						By varying the threshold parameter $\theta$, we construct a receiver-operating characteristic (ROC) curve that consists of points whose horizontal and vertical coordinates are the false-positive rates and true-positive rates, respectively. For denoising noise of type {${-}\textup{ER}$ (respectively, ${+}\textup{ER}$ and ${+}\textup{WS}$)}, we also infer an optimal value of $\theta$ for a 25\% validation set of nonedges (respectively, edges) in $\G'$ with known labels and then use that value of $\theta$ to compute classification measures (such as accuracy and precision) for the test set.

					}

			\subsection{Figure~\ref{fig:img_recons_ex}}\label{subsection:img_recons_ex}
			
				In Figure~\ref{fig:img_recons_ex}, we illustrate cross-reconstruction experiments for images using mesoscale patches of size $21 \times 21$. The image that we seek to reconstruct is \dataset{Woman with a Parasol - Madame Monet and Her Son} (Claude Monet, 1875), which we show in Figure~\ref{fig:img_recons_ex}\textbf{a}. The image in Figure~\ref{fig:img_recons_ex}\textbf{b} is a reconstruction of this image
				using the dictionary with 25 basis images of size $21 \times 21$ pixels in Figure~\ref{fig:img_recons_ex}\textbf{c}, where we choose the color of each pixel uniformly at random from all possible colors (which we represent as vectors in $[0,256]^{3}$ for red--green--blue (RGB) weights). The image in Figure~\ref{fig:img_recons_ex}\textbf{d} is a reconstruction of the image in Figure~\ref{fig:img_recons_ex}\textbf{a} using the dictionary with 25 basis images of size $21 \times 21$ pixels in Figure~\ref{fig:img_recons_ex}\textbf{e}. We learn this basis from the image in Figure~\ref{fig:img_recons_ex}\textbf{f} using NMF~\cite{lee1999learning}. The image in Figure~\ref{fig:img_recons_ex}\textbf{f} is from the collection \dataset{Die Graphik Ernst Ludwig Kirchners bis 1924, von Gustav Schiefler Band I bis 1916} (Accession Number 2007.141.9, Ernst Ludwig Kirchner, 1926). We use the images in Figures~\ref{fig:img_recons_ex}\textbf{a},\textbf{f} with permission from the National Gallery of Art in Washington, DC, USA. Their open-access policy is available at \url{https://www.nga.gov/notices/open-access-policy.html}.

			
			\subsection{Figure~\ref{fig:covid_dict_comparison}}
			
				In Figure~\ref{fig:covid_dict_comparison}, we compare $10^{4}$ subgraphs (we show 33 of them) that are induced by
				{approximately} uniformly random (\textbf{a}) $k$-paths and (\textbf{b}) $k$-walks for the network \dataset{Coronavirus PPI} with $k = 10$. To sample $k$-walks, we use an MCMC
				motif-sampling algorithm (see Algorithm~\ref{alg:pivot} with $\mathtt{AcceptProb} = \mathtt{Approximate}$). 
				To sample $k$-paths, we use the injective MCMC motif-sampling algorithm (see Algorithm~\ref{alg:motif_inj}).

				To obtain the latent motifs in Figure~\ref{fig:covid_dict_comparison}\textbf{c}, we use the NDL Algorithm (see Algorithm~\ref{alg:NDL}) with $\mathtt{MCMC} = \mathtt{PivotApprox}$. To obtain the latent motifs in Figure~\ref{fig:covid_dict_comparison}\textbf{d}, we instead use the NDL algorithm of Lyu et al.~\cite{lyu2020online}. 
				This algorithm is equivalent to our NDL algorithm (see Algorithm~\ref{alg:NDL}) if we use all homomorphisms from the $k$-chain motif $F$ to the network $\G$, rather than only the	injective ones as in Algorithm~\ref{alg:NDL}. For the experiments in Figures~\ref{fig:covid_dict_comparison}\textbf{c},\textbf{d}, we use a $k$-chain motif 
				$F=([k], A_{F})$ with $k = 21$, $T=100$ iterations, $N = 100$ injective homomorphisms per iteration (so we sample a total of $10^{4}$ injective homomorphisms), and an $L_{1}$-regularizer with coefficient $\lambda = 1$.

			\subsection{Figure~\ref{fig:denoising_caltech_comparison}}\label{subsection:fig_caltech_ddenoising}

			{In Figure~\ref{fig:denoising_caltech_comparison}, we show histograms of various statistics of true and false edges for denoising the network \dataset{Caltech} after corrupting it with $50$\% additive noise of type ${+}\textup{ER}$.
				
				We say that a $k$-walk $\x_{t}$ `visits' an edge between distinct nodes $x$ and $y$ if there exist indices $a$ and $b$ such that
				$1 \le a < b \le k$ with $\x_{t}(a) = x$ and $\x_{t}(b) = y$. We define the `distance' between nodes $x$ and $y$ along the $k$-walk $\x_{t}$ that visits the edge between them by the minimum value of $|i - j|$, where $i$ and $j$ are integers in $\{1,\ldots,k\}$ such that $\x_{t}(i)=x$ and $\x_{t}(j)=y$. (See the illustration in Figure~\ref{fig:denoising_caltech_comparison}\textbf{a}.) We then compute the mean of such distances between $x$ and $y$ for all $t\in \{1,\ldots,T\}$ for which
				$\x_{t}$ visits the edge between $x$ and $y$. In Figure~\ref{fig:denoising_caltech_comparison}\textbf{g}, we show the histogram of the mean distances along the $k$-walks $\x_{1},\ldots,\x_{T}$ between the two ends of true edges and the two ends of false edges. 
				
				We compute the $r = 25$ latent motifs in Figure~\ref{fig:denoising_caltech_comparison}\textbf{h} using NDL (see Algorithm~\ref{alg:NDL}) with  
				a $k$-chain motif with the corresponding network
				$F=([k], A_{F})$ for $T = 100$ iterations, $N = 100$ 
				injective homomorphisms per iteration (so we sample a total of $10^{4}$ injective homomorphisms), an $L_{1}$-regularizer with coefficient $\lambda=1$, and the injective MCMC motif-sampling algorithm $\mathtt{MCMC} = \mathtt{PivotApprox}$.

				For our reconstructions of the corrupted network \dataset{Caltech} in Figures~\ref{fig:denoising_caltech_comparison}\textbf{b}--\textbf{f}, which use the corresponding latent motifs in Figures~\ref{fig:denoising_caltech_comparison}\textbf{h}--\textbf{l}, respectively, we apply the NDR algorithm (see Algorithm~\ref{alg:network_reconstruction}) to a $k$-chain motif with the corresponding network
				$F=([k], A_{F})$ at scale $k = 21$ for $T = 2 \times 10^{5}$ iterations,
				$N = 100$ homomorphisms per iteration, $r = 25$ latent motifs,
				an $L_{1}$-regularizer with coefficient $\lambda = 0$ (i.e., no regularization), the MCMC motif-sampling algorithm $\mathtt{MCMC} = \mathtt{PivotApprox}$, $\mathtt{InjHom} = \texttt{F}$, and $\mathtt{denoising}=\texttt{F}$. 
			}


			
			\subsection{Figure~\ref{fig:NDL_alg}}
			
			In Figure~\ref{fig:NDL_alg}, {we give a schematic illustration of the NDL algorithm (see Algorithm~\ref{alg:NDL})}. 
			
			\subsection{{Figure}~\ref{fig:latent_motifs_multiscale_1} }
			\label{subsection:dict_figures_SI}
			
			In this figure, we show latent motifs of the examined networks (which we described in the `Methods' section of the main manuscript) using Algorithm~\ref{alg:NDL} with various parameter choices. In each column of this figure, we use a $k$-chain {motif} 
			$F = ([k], A_{F})$ 
			for $T = 100$ iterations, $N = 100$ injective homomorphisms per iteration (so we sample a total of $10^{4}$ injective homomorphisms), an $L_{1}$-regularizer with coefficient $\lambda = 1$, and the injective MCMC motif-sampling algorithm $\mathtt{MCMC} = \mathtt{PivotApprox}$. 
			We specify the number $r$ of latent motifs and the scale $k$ in the caption of Figure~\ref{fig:latent_motifs_multiscale_1}.

			
			
			\subsection{Figure~\ref{fig:Figure_boxcompare}}\label{subsection:fig_boxplot}

			In Figure~\ref{fig:Figure_boxcompare}, we compare box plots of the community sizes of 10,000 
			sampled subgraphs that are induced by approximately uniformly random paths of $k = 21$ 
			nodes (in red) to the corresponding box plots of $r = 25$ latent motifs of $k = 21$ nodes for various real-world and synthetic networks. To sample these paths, we use Algorithm~\ref{alg:pivot} with $\mathtt{AcceptProb}=\mathtt{Approximate}$.
			{(See the `Methods' section of the main manuscript.)
				We determine latent motifs using NDL (see Algorithm~\ref{alg:NDL}) with a $k$-chain motif
				$F=([k], A_{F})$ for $T = 100$ iterations, $N = 100$ injective homomorphisms per iteration (so we sample a total of $10^{4}$ injective homomorphisms), an $L_{1}$-regularizer with coefficient $\lambda = 1$, and the injective MCMC motif-sampling algorithm $\mathtt{MCMC} = \mathtt{PivotApprox}$. We determine the communities of the subgraphs and the latent motifs using the locally-greedy Louvain algorithm for modularity maximization~\cite{blondel2008fast}.

				We perform statistical testing to compare the community sizes of the subgraphs (`sample 1') and the latent motifs (`sample 2'). We select a uniformly random subset of sample 1 to match the size of
				sample 2 and perform Mood's median test~\cite{brown1951median} to obtain a $p$-value. We repeat this experiment 100 times for each network and report the mean value of the resulting 100 $p$-values in parentheses in Figure \ref{fig:Figure_boxcompare}. A sufficiently small $p$-value indicates that there is statistically significant evidence that the two samples come from populations with distinct medians.


					
					\subsection{Figures~\ref{fig:Figure4__PRC_REC}, \ref{fig:Figure_PR_curve}, and \ref{fig:Figure_PR_curve_flipped}}
			
					One summarizes the result of a binary classification using combinations of four quantities: 
					$\text{TP}$ (true positives), which is the number of positives that are classified as positive; $\text{TN}$ (true negatives), which is the number of {negatives} that are classified as negative; $\text{FP}$ (false positives), which is the number of positives that are classified as negative; and $\text{FN}$ (false negatives), which is the number of negatives that are classified as positive. The total number of examples is the sum of these four quantities. Accuracy is $\frac{\text{TP} + \text{TN}}{\text{TP} + \text{TN} + \text{FP} + \text{FN}}$, precision is $\frac{\text{TP}}{\text{TP} + \text{FP}}$, recall is $\frac{\text{TP}}{\text{TP} + \text{FN}}$, negative predictive value (NPV) is $\frac{\textup{TN}}{\textup{TN}+\textup{FN}}$, and specificity is $\frac{\textup{TN}}{\textup{TN}+\textup{FP}}$. 
					By relabeling positives as negatives and negatives as positives, precision becomes NPV and recall becomes specificity. See \cite{parikh2008understanding} for a discussion of NPV and specificity.
					
					
					In Figure~\ref{fig:Figure4__PRC_REC}, we show the accuracy, precision, and recall scores of the network-denoising experiments in Figure~\ref{fig:Figure4} at a fixed threshold $\theta$. 
					In Figure~\ref{fig:Figure_PR_curve}, we show the dependence of the precision and recall scores of the network-denoising experiments in Figure~\ref{fig:Figure4} on the threshold $\theta$.  
					For denoising $+$WS and $+$ER noise,
					our approach yields larger AUCs for
					 the precision--recall curves than all of the other examined network-denoising methods.
					 Our approach performs competitively  
					 for denoising $-$ER noise, except for the network \dataset{H. Sapiens}.
	  In Figure~\ref{fig:Figure_PR_curve_flipped}, we show the dependence of the NPV and specificity scores of the network-denoising experiments in Figure~\ref{fig:Figure4} on the threshold $\theta$. 
	  Our approach yields larger AUCs for the
	  NPV--specificity curves than all of the other examined
	  methods for denoising $+$WS and $+$ER noise, except for $+$ER for the network \dataset{Caltech}. 
	  For denoising 
	  $-$ER noise, our approach does not seem to be particularly effective
	  at detecting unobserved edges. It is outperformed by other methods for \dataset{SNAP FB} {(by all of them except {\sc preferential attachment})}, \dataset{arXiv} {(by all of them except {\sc preferential attachment} and {\sc spectral embedding})}, and \dataset{H. Sapiens} (by all of them except {\sc preferential attachment} and {\sc spectral embedding}).


					\subsection{Figures~\ref{fig:all_dictionaries_1},~\ref{fig:all_dictionaries_2},~\ref{fig:all_dictionaries_3},~\ref{fig:all_dictionaries_4},~\ref{fig:all_dictionaries_rank_1},~\ref{fig:all_dictionaries_rank_2},~\ref{fig:all_dictionaries_rank_3}, and~\ref{fig:all_dictionaries_rank_4}}

					In these figures, we show latent motifs of the examined networks (which we described in the `Methods' section of the main manuscript) using Algorithm~\ref{alg:NDL} with various parameter choices. For each network, we use a $k$-chain motif with the corresponding network $F=([k], A_{F})$ for $T=100$ iterations, $N = 100$ injective homomorphisms per iteration (so we sample a total of $10^{4}$ injective homomorphisms), an $L_{1}$-regularizer with coefficient $\lambda = 1$, and the injective MCMC motif-sampling algorithm $\mathtt{MCMC} = \mathtt{PivotApprox}$.  We specify the number $r$ of latent motifs and the scale $k$ in the caption of each figure.

					
					\section{Convergence Analysis}
					\label{section:convergence}
					
						In this section, we give rigorous convergence guarantees for our main algorithms for NDL (see Theorems~\ref{thm:NDL} and~\ref{thm:NDL2}) and NDR (see Theorems~\ref{thm:NR} and~\ref{thm:NR2}). All of these results are novel.
						Lyu et al.~\cite{lyu2020online} proposed a network-reconstruction algorithm that
						corresponds to our NDR algorithm~\ref{alg:network_reconstruction} with the choice $\texttt{denoising} = \texttt{InjHom}=\texttt{F}$. 
						They did not do any theoretical analysis of this network-reconstruction algorithm.		
						Our most significant theoretical contribution is our guarantees about the NDR algorithm (see Algorithm~\ref{alg:network_reconstruction}). Specifically, in Theorems~\ref{thm:NR} and~\ref{thm:NR2}, we establish convergence, exact formulas, and error bounds of the reconstructed networks for all four choices $(\texttt{denoising}, \texttt{InjHom}) \in \{\texttt{T}, \texttt{F}\}^{2}$ for both non-bipartite and bipartite networks. The most interesting aspect of these results is our bound for the 
						Jaccard reconstruction error in terms of the mesoscale approximation error divided by the number of nodes in the subgraphs at that mesoscale (see Theorem~\ref{thm:NR}\textbf{(iii)}). 
						Roughly speaking, this result guarantees that one can accurately reconstruct a network
						if one has a dictionary of latent motifs that can accurately approximate the subgraphs in the network at a fixed mesoscale.
						We illustrate this result with supporting experiments in Figure~\ref{fig:recons_bd_plot} of the main manuscript.
						A crucial feature of our proof of Theorem~\ref{thm:NR}\textbf{(iii)} is our use of an
						explicit formula for the weight matrix of the limiting reconstructed network as the number of iterations that we use for network reconstruction tends to infinity.

						In~\cite[Corollary 6.1]{lyu2020online}, Lyu et al.~presented a convergence guarantee for the original NDL algorithm in~\cite{lyu2020online} for non-bipartite networks $\G$. 
						The key difference between our NDL algorithm (see Algorithm~\ref{alg:NDL}) and the NDL algorithm in~\cite{lyu2020online} is that we employ
						$k$-path motif sampling but they employ
						$k$-walk motif sampling. 
						This results in different objective functions to minimize. (See \eqref{eq:NDL_exact} for our objective function.)} Therefore, 
						\cite[Corollary 6.1]{lyu2020online} does not apply to our NDL algorithm (see Algorithm~\ref{alg:NDL}). In Theorem,~\ref{thm:NDL}, we establish a convergence result for our NDL algorithm. A key step in the proof of this result is guaranteeing convergence of the injective MCMC motif-sampling algorithm (see Algorithm~\ref{alg:motif_inj}) for non-bipartite networks (see Proposition~\ref{prop:inj_motif_conv}).
						In Theorem~\ref{thm:NDL2}, we extend the convergence results for our NDL algorithm to bipartite networks. Convergence for bipartite networks was not established} in~\cite{lyu2020online} even for the original NDL algorithm. 
						The key technical difficulty for bipartite networks
						is that Markov chains that are generated by the MCMC motif-sampling algorithms in Algorithms \ref{alg:pivot} and \ref{alg:glauber} are not irreducible, so the main convergence results for online NMF
						in~\cite[Thm. 1]{lyu2020online} are not applicable. 
						Our proof of Theorem~\ref{thm:NDL2} uses a careful coupling argument between two reducible
						classes of Markov chains.

					Let $F=([k],A_{F})$ be the network corresponding to a $k$-chain motif, and let $G = (V,A)$ be a network. Let $\Omega \subseteq V^{[k]}$ denote the set of all homomorphisms (which do not have to be injective) $\x: F \rightarrow \G$. Algorithm~\ref{alg:NDL} generates three stochastic sequences. The first one is the sequence $(\x_{t})_{t\ge 0}$ of injective homomorphisms $F\hookrightarrow \G$ that we obtain from the injective 
					MCMC motif-sampling algorithm (see Algorithm~\ref{alg:rejection_motif}). 
					The second one is the sequence $(X_{t})_{t\ge 0}$ of $k^{2} \times N$ data matrices whose columns encode $N$ mesoscale patches of $\G$. More precisely, for each $\y_{1},\ldots,\y_{n} \in \Omega$, we write $\Psi(\y_{1},\ldots,\y_{n}) \in \mathbb{R}_{\ge 0}^{k^{2}\times N}$ for the $k^{2} \times N$ matrix whose $i^{\textup{th}}$ column is the vectorization (using Algorithm~\ref{alg:vectorize}) of the corresponding $k \times k$ mesoscale patch $A_{\y_{i}}$ of $\G$ (see \eqref{eq:def_A_F_patch}). For each $\y_{0} \in \Omega$, define 
					\begin{align}\label{eq:def_X_N}
						X^{(N)}(\y_{0}):=\Psi(\y_{1},\ldots,\y_{N})\in \mathbb{R}_{\ge 0}^{k^{2}\times N}\,,
					\end{align}
					where we generate $\y_{1},\ldots, \y_{N}$ {iteratively using the injective MCMC motif-sampling algorithm (see Algorithm~\ref{alg:motif_inj}), which we initialize with the homomorphism $\y_{0}:F\rightarrow \G$.
						It then follows that $X_{t} = X^{(N)}(\x_{Nt})$ for each $t \ge 1$, where $\x_{Nt}$ is the injective homomorphism $F\hookrightarrow \G$ that we obtain after $Nt$ applications of Algorithm~\ref{alg:motif_inj}. For the third (and final) sequence that we generate using Algorithm~\ref{alg:NDL}, let $(W_{t})_{t \ge 0}$ denote the sequence of dictionary matrices, where we define each $W_{t}=W_{t}(\x_{0})$ via \eqref{eq:def_ONMF} with an initial homomorphism $\x_{0}:F\rightarrow \G$ that we sample using Algorithm~\ref{alg:rejection_motif}.

						
						\subsection{Convergence of the MCMC algorithms}
						\label{subsection:convergence_approx_pivot}

						In this subsection, we establish the convergence properties of the MCMC algorithms (see Algorithms~\ref{alg:pivot} and~\ref{alg:glauber}) for sampling a $k$-walk according to the target distribution $\pi_{F\rightarrow \G}$ {(see \eqref{eq:def_embedding_F_N}).}
						Recall that this {distribution} {is} the uniform distribution on the set of all homomorphisms $F\rightarrow \G$ when the motif $F$ and the network $\G$ are both symmetric and unweighted (see \eqref{eq:def_embedding_F_N_uniform}).

						\begin{prop}\label{prop:exact_MCMC}
							Fix a network $\G = (V,A)$ and a $k$-chain motif $F=([k], A_{F})$. 
							Let $(\x_{t})_{t\ge 0}$ denote a sequence of homomorphisms $\x_{t}: F \rightarrow \G$ that we generate using the exact pivot chain (in which we use Algorithm~\ref{alg:pivot} with $\mathtt{AcceptProb}=\mathtt{Exact}$) or the Glauber chain (in which we use Algorithm~\ref{alg:glauber}). Suppose that 
							\begin{description}[leftmargin=0.7cm, topsep=0.2pt, itemsep=0.1cm]
								\item{\textup{(a)}} The weight matrix $A$ is `bidirectional' (i.e., $A(x,y)>0$ implies that $A(y,x)>0$ for all $x,y \in V$) and that the binary network $(V, \one(A>0))$ is connected and non-bipartite.
							\end{description} 
							It then follows that $(\x_{t})_{t\ge 0}$ is an irreducible and aperiodic Markov chain with the unique stationary distribution $\pi_{F\rightarrow \G}$ that we defined in \eqref{eq:def_embedding_F_N}.		
						\end{prop}

						\begin{proof}
								This proposition was proved rigorously in~\cite[Thms. 5.7 and 5.8]{lyu2023sampling}. In the present paper, we sketch the proof for the exact pivot chain to illustrate the main idea behind the 
								acceptance probability in \eqref{eq:pivot_chain_acceptance_prob}. 
								The trajectory of the first node $\x_{t}(1)$ of the homomorphism $\x$ gives a standard random walk on the network $\G$ that is modified by the Metropolis--Hastings algorithm (see, e.g.,~\cite[Sec. 3.2]{levin2017markov}) so that it has the following marginal distribution as its unique stationary distribution: 
								\begin{equation}
									\pi^{(1)}(x_{1}) = \frac{\sum_{x_{2},\ldots,x_{k}\in [n]} \prod_{i=2}^{k} A(x_{i-1},x_{i})}{\mathtt{Z}} \,, 
								\end{equation} 
								where {the denominator $\mathtt{Z}=\mathtt{Z}(F,\G)$ is the normalization constant that we call the `homomorphism density' of $F$ in $\G$ (see~\cite{lovasz2012large})} {and} the numerator is proportional to the 
								probability of sampling $x_{2},\ldots,x_{{k}} \in [n]$ {for $\x_{t}(2),\ldots,\x_{t}(k)$.}
								Fix a homomorphism $\x:F\rightarrow \G$ with $\x(i) = x_{i}$ for $i=1,\ldots,k$. 
								We then obtain 
									\begin{align}
										\P(\x_{t}(1)=x_{1},\ldots,\x_{t}(k)=x_{k}) &= \P\left( \x_{t}(1)=x_{1} \right)	\P( \x_{t}(2)=x_{2},\ldots,\x_{t}(k)=x_{k} \,|\, \x_{t}(1)=x_{1})  \\
										&\approx \pi^{(1)}(x_{1}) \frac{\prod_{i=2}^{k} A(x_{i-1},x_{i}) }{\sum_{y_{2},\ldots,y_{k}\in [n]} \prod_{i=2}^{k} A(y_{i-1},y_{i})  } \\
										&=  \frac{ \prod_{i=2}^{k} A(x_{i-1},x_{i})  }{ \mathtt{Z} } = \pi_{F\rightarrow \G}(\x)\,, 
									\end{align}
								where the approximation in the second line above becomes exact as $t\rightarrow \infty$.  
						\end{proof}

						We now prove Proposition~\ref{prop:approximate_pivot}, which guarantees the convergence of the approximate pivot chain and gives an explicit formula for its unique stationary distribution.

						\begin{prop}\label{prop:approximate_pivot}
							Fix a network $\G=(V,A)$ and {a} $k$-chain motif 
							{$F = ([k], A_{F})$}. 
							Let $(\x_{t})_{t\ge 0}$ denote a sequence of homomorphisms $\x_{t}:F\rightarrow \G$ that we generate using the approximate pivot chain (in which we use Algorithm~\ref{alg:pivot} with $\mathtt{AcceptProb}=\mathtt{Approximate}$). Suppose that 
							\begin{description}[leftmargin=0.7cm, topsep=0.2pt, itemsep=0.1cm]
								\item{\textup{(a)}} The weight matrix $A$ is bidirectional (i.e., $A({x,y})>0$ implies that $A({y,x})>0$ for all ${x,y} \in V$) and that the undirected and unweighted graph $(V, \one(A>0))$ is connected and non-bipartite.
							\end{description}
							It then follows that $(\x_{t})_{t\ge 0}$ is an irreducible and aperiodic Markov chain with the unique stationary distribution $\hat{\pi}_{F\rightarrow \G}$ that we defined in \eqref{eq:approx_pivot_statioanary_dist}.		
						\end{prop}

						\begin{proof}
							We follow the proof of~\cite[Thm. 5.8]{lyu2023sampling}. Let $P: V^{2} \rightarrow [0,1]$ be a matrix with entries
							\begin{align}
								P(x,y) := \frac{A(x,y)}{\sum_{c\in V} A(a,c)}\,, \quad \text{$x,y\in V$}\,. 
							\end{align}
							The matrix $P$ is the transition matrix of the standard random walk on the network $\G$. By hypothesis (a), $P$ is irreducible and aperiodic. 
							Using a result in~\cite[Ch. 9]{levin2017markov}, it has the unique stationary distribution 
							\begin{align}
								\pi^{(1)}(v) := \sum_{c\in V} A(v,c) / \sum_{c,c'\in V} A(c,c')\,.
							\end{align} 
							The approximate pivot chain generates a move $\x_{t}(1)\mapsto \x_{t+1}(1)$ of the pivot according to the distribution $P(\x_{t}(1),\cdot)$. We accept this move 
							independently of everything else with the approximate acceptance probability $\alpha$ in \eqref{eq:pivot_chain_acceptance_prob}. 
							If we always accept each move of the pivot, then the pivot performs a random walk on $\G$ with the unique stationary distribution $\pi^{(1)}$. We compute the acceptance probability $\alpha$ using the Metropolis--Hastings algorithm (see~\cite[Sec. 3.3]{levin2017markov}), and we thereby modify the stationary distribution of the pivot from $\pi^{(1)}$ to the uniform distribution on $V$. (See the discussion in~\cite[Sec. 5]{lyu2023sampling}.) Therefore, $(\x_{t}(1))_{t\ge 0}$ is an irreducible and aperiodic Markov chain on $V$ that has the uniform distribution as its unique stationary distribution. Because we sample the locations $\x_{t+1}(i)\in V$ of the subsequent nodes $i \in \{2,3,\ldots,k$\} independently, conditional on the location $\x_{t+1}(1)$ of the pivot, it follows that the approximate pivot chain $(\x_{t})_{t\ge 0}$ is also an irreducible and aperiodic Markov chain with a unique stationary distribution, which we denote by $\hat{\pi}_{F\rightarrow \G}$. 
							
							To determine the stationary distribution $\hat{\pi}_{F\rightarrow \G}$, we decompose $\x_{t}$ into the return times of the pivot $\x_{t}(1)$ to a fixed node $x_{1}\in V$ in $\G$. Specifically, let $\tau(j)$ be the $j^{\textup{th}}$ return time of $\x_{t}(1)$ to $x_{1}$. By the independence of sampling $\x_{t}$ over $\{2,\ldots,k \}$ for each $t$, the strong law of large numbers yields  
							\begin{align}
								&\lim_{M\rightarrow \infty}\frac{1}{M} \sum_{j=1}^{M}\one(\x_{\tau(j)}(2)=x_{2},\, \ldots, \x_{\tau(j)}(k)=x_{k}) \\
								&\qquad =  \frac{\prod_{i=2}^{k} A(x_{i-1},x_{i}) }{\sum_{ y_{2},\ldots,y_{k}\in V} \prod_{i=2}^{k} A(x_{1},y_{2})A(y_{2},y_{3})\cdots A(y_{k-1},y_{k})}\,.
							\end{align} 
							For each fixed homomorphism $\x:F\rightarrow \G$, which maps $i\mapsto x_{i}$, we use the Markov-chain ergodic theorem (see, e.g.,~\cite[Theorem 6.2.1 and Example 6.2.4]{Durrett} or~\cite[Theorem 17.1.7]{meyn2012markov}) to obtain 
							\begin{align}
								\hat{\pi}_{F\rightarrow \G}(\x) &= \lim_{N\rightarrow \infty} \frac{1}{N} \sum_{t=0}^{N} \one(\x_{t}=\x) \\
								&=  \lim_{N\rightarrow \infty} \frac{\sum_{t=0}^{N}   \one(\x_{t}=\x) }{\sum_{t=0}^{N} \one(\x_{t}(1)=x_{1})}  \, \frac{\sum_{t=0}^{N} \one(\x_{t}(1)=x_{1})}{N}\\
								&=\mathbb{P}\left( \x_{t}(2)=x_{2},\ldots,\x_{t}(k)=x_{k} \,\bigg| \, \x_{t}(1)=x_{1}\right) \, \pi^{(1)}(x_{1})\\ 
								&=\frac{\prod_{i=1}^{k}A(x_{i-1},x_{i}) }{\sum_{ y_{2},\ldots,y_{k}\in V} \prod_{i=2}^{k} A(x_{1},y_{2})A(y_{2},y_{3}) \cdots A(y_{k-1},y_{k})} \frac{1}{|V|}\,.
							\end{align} 
							This proves the assertion. 
						\end{proof}
						
						{\color{black} 
							Finally, we establish the following asymptotic convergence result for the injective MCMC motif-sampling algorithm in Algorithm~\ref{alg:motif_inj}.

							\begin{prop}[Convergence of injective motif sampling]\label{prop:inj_motif_conv}
								Fix a network $\G = (V,A)$ and a $k$-chain motif 
								$F = ([k], A_{F})$. Suppose that $\G$ has at least one $k$-path. Let $(\x_{t})_{t\ge 0}$ denote a sequence of injective homomorphisms $\x_{t}:F\hookrightarrow \G$ that we generate using Algorithm~\ref{alg:motif_inj}. Suppose that 
								\begin{description}[leftmargin=0.7cm, topsep=0.2pt, itemsep=0.1cm]
									\item{\textup{(a)}} The weight matrix $A$ is bidirectional (i.e., $A({x,y})>0$ implies that $A({y,x})>0$ for all ${x,y} \in V$) and that the undirected and unweighted graph $(V, \one(A>0))$ is connected and non-bipartite.
								\end{description}
								The following statements hold:
								\begin{description}[leftmargin=0.7cm, topsep=0.2pt, itemsep=0.1cm]
									\item[(i)] If we use Algorithm~\ref{alg:glauber} or Algorithm~\ref{alg:pivot} with $\mathtt{AcceptProb} = \mathtt{Exact}$ in Algorithm~\ref{alg:motif_inj}, then $(\x_{t})_{t\ge 0}$ is an irreducible and aperiodic Markov chain with the unique stationary distribution $\pi_{F\hookrightarrow \G}$ that we defined in \eqref{eq:def_embedding_F_N_inj}.
									
									\item[(ii)] If we use Algorithm~\ref{alg:pivot} with $\mathtt{AcceptProb}=\mathtt{Approximate}$ in Algorithm~\ref{alg:motif_inj}, then $(\x_{t})_{t\ge 0}$ is an irreducible and aperiodic Markov chain with the unique stationary distribution 
									$\hat{\pi}_{F\hookrightarrow \G}$ that is defined by 
									\begin{align}\label{eq:approx_pivot_statioanary_dist_inj}
										\hat{\pi}_{F\hookrightarrow \G}(\x) = C' \hat{\pi}_{F\rightarrow \G}(\x) \, \one( \textup{$\x(1),\ldots,\x(k)$ are distinct} ) \,,
									\end{align}
									where $\hat{\pi}_{F\rightarrow \G}$ is the probability distribution on the set of homomorphisms $F\rightarrow \G$ in \eqref{eq:approx_pivot_statioanary_dist}.
								\end{description}
							\end{prop}
							
							\begin{proof}
								This assertion follows from standard Markov-chain theory (see, e.g.,~\cite{levin2017markov}) in conjunction with Propositions~\ref{prop:exact_MCMC} and~\ref{prop:approximate_pivot}. 
							\end{proof}
							
						}


						\subsection{Convergence of the NDL algorithm}

						Recall the problem statement for NDL in \eqref{eq:NDL_exact}. Informally, we seek to learn $r$ latent motifs $\mathcal{L}_{1},\ldots,\mathcal{L}_{r} \in \R_{\ge 0}^{k \times k}$ to minimize the {expectation of the error of approximating the mesoscale patch $A_{\x}$ by a nonnegative combination of the motifs $\mathcal{L}_{i}$}, where $\x:F\hookrightarrow \G$ is a random injective homomorphism that we sample from the distribution $\pi_{F\hookrightarrow \G}$ \eqref{eq:def_embedding_F_N_inj}. We reformulate this problem as a matrix-factorization problem that
						generalizes \eqref{eq:NDL_exact}. Let $\mathcal{C}^{\textup{dict}}$ denote the set of all matrices $W \in \R_{\ge 0}^{k^{2}\times r}$ whose columns have a Frobenius norm of at most $1$. The matrix-factorization problem is then
						\begin{align}\label{eq:NDL_exact_minibatch}
							\argmin_{W\in \mathcal{C}^{\textup{dict}} \subseteq \R_{\ge 0}^{k^{2}\times r}}   \bigg( f(W):= \E_{\x\sim \pi_{F\hookrightarrow \G}}\left[\ell(X^{(N)}(\x), W)  \right] \bigg)\,,
						\end{align}
						where we define the loss function 
						\begin{align}\label{eq:def_loss}
							\ell(X,W) := \inf_{H \in \mathbb{R}^{r\times N}_{\ge 0} }  \lVert X - WH  \rVert_{F}^{2} + \lambda \lVert H \rVert_{1}\,, \quad \text{$X\in \R^{k^{2}\times N}$\,, $W\in \R^{k^{2}\times r}$}\,.
						\end{align}
						The parameters $N \in \mathbb{N}$ and $\lambda \ge 0$ appear in Algorithm~\ref{alg:NDL}. The former is the number of homomorphisms that we sample at each iteration of Algorithm~\ref{alg:NDL}, and the latter is the coefficient of an $L_{1}$-regularizer that we use to find the code matrix $H_{t}$ in \eqref{eq:def_ONMF}. If
						$N = 1$ and $\lambda = 0$, then 
						problem \eqref{eq:NDL_exact_minibatch} is equivalent to 
						problem \eqref{eq:NDL_exact} because $X^{(1)}(\x)$ and the columns of $W$ are vectorizations (using Algorithm~\ref{alg:vectorize}) of the mesoscale patch $A_{\x}$ and the latent motifs $\mathcal{L}_{1},\ldots,\mathcal{L}_{r}$, respectively.

						The objective function $f$ in the optimization problem \eqref{eq:NDL_exact_minibatch} for NDL is non-convex, so it is 
							generally difficult to find a global optimum of $f$. 
							However, local optima are often good enough for practical applications (such as in image restoration~\cite{elad2006image, mairal2008sparse}). We find that this is also the case for our network-denoising problem (see Figure~\ref{fig:Figure4}). 
							Theorems~\ref{thm:NDL} and~\ref{thm:NDL2} guarantee that our NDL algorithm (see Algorithm~\ref{alg:NDL}) finds a sequence $(W_{t})_{t\ge 0}$ of dictionary matrices such that almost surely $W_{t}$ is asymptotically a local optimum of the objective function $f$.

							To make precise statements {about asymptotic convergence of the NDL algorithm to a local optimum of the objective function in \eqref{eq:NDL_exact},
								we consider measures of the local optimality of
								non-convex constrained optimization problems. Suppose that we have a differentiable objective function $g:\R^{p} \rightarrow \R$ for some integer $p \ge 1$. Fix a parameter set $\Param \subseteq \R^{p}$. We say that $\param^{*}\in \Param$ is a \textit{stationary point} of $g$ in $\Param$ if 
								\begin{align}\label{eq:stationary}
									\inf_{\param\in \Param} \, \langle \nabla g(\param^{*}) ,\,  \param - \param^{*} \rangle \ge 0 \,,
								\end{align}	
								where $\langle \cdot,\, \cdot \rangle$ denotes the dot project on $\R^{p}$.
								If $\param^{*}$ is a stationary point of $g$ in $\Param$ and it is in the interior of $\Param$, then $\lVert \nabla g(\param^{*}) \rVert=0$. }

						We are now ready to state the convergence result for our NDL algorithm (see Algorithm~\ref{alg:NDL}) for non-bipartite networks.

						\begin{theorem}[Convergence of the NDL Algorithm for Non-Bipartite Networks]\label{thm:NDL}
							Let $F=([k],A_{F})$ be a $k$-chain motif, and let $G = (V,A)$ be a network that satisfies the following properties:
							\begin{description}[leftmargin=0.7cm, topsep=0.2pt, itemsep=0.1cm]
								\item{(a)} The weight matrix $A$ is bidirectional (i.e., $A({x,y})>0$ implies that $A({y,x})>0$ for all ${x,y} \in V$) and the binary network $(V, \one(A > 0))$ is connected and non-bipartite.
								
								\item{(b)} For all $t \ge 0$, there exists a unique solution $H_{t}$ in \eqref{eq:def_ONMF}.
								
								\item{(c)} For all $t \ge 0$, the eigenvalues of the positive semidefinite matrix $A_{t}$ that is defined in \eqref{eq:def_ONMF} are at least as large as some constant $\kappa_{1}>0$.
							\end{description}
							
							Let $(W_{t})_{t\ge 0}$ denote the sequence of dictionary matrices that we generate using Algorithm~\ref{alg:NDL}. The following statements hold: 
							\begin{description}[leftmargin=0.7cm, topsep=0.2pt, itemsep=0.1cm]
								\item[(i)] For $\mathtt{MCMC}\in \{ \mathtt{Pivot}, \mathtt{Glauber}\}$, it is almost sure true
								as $t \rightarrow \infty$ that $W_{t}$ converges to the set of stationary points of the objective
								function $f$ that we defined in \eqref{eq:NDL_exact_minibatch}. Furthermore, if $f$ has finitely many stationary points in $\mathcal{C}^{\textup{dict}}$, it is then the case that $W_{t}$ converges to a single stationary point of $f$ almost surely as $t \rightarrow \infty$.

								\vspace{0.1cm}
								\item[(ii)] For $\mathtt{MCMC}=\mathtt{PivotApprox}$, it is almost surely true as $t\rightarrow \infty$ that $W_{t}$ converges to the set of stationary points of the objective function
								\begin{align}\label{eq:def_f_hat}
									\hat{f}(W):= \E_{\x\sim \hat{\pi}_{F\hookrightarrow \G}}\left[\ell(X^{(N)}(\x), W)  \right]\,, 
								\end{align}
								where the distribution $\hat{\pi}_{F\hookrightarrow \G}$ is defined in \eqref{eq:approx_pivot_statioanary_dist_inj} and the loss function $\ell$ is defined in \eqref{eq:def_loss}.
								 Furthermore, if $\hat{f}$ has finitely many stationary points in $\mathcal{C}^{\textup{dict}}$, it is then the case that $W_{t}$ converges to a single stationary point of $\hat{f}$ almost surely as $t\rightarrow \infty$.
							\end{description}
						\end{theorem}

						\begin{remark}
							\textup{
								Assumptions (a)--(c) in Theorem~\ref{thm:NDL} are all reasonable and are easy to satisfy. Assumption (a) is satisfied if $\G$ is undirected, unweighted, and connected, which is the case for all of the networks that we study in the present paper. Assumptions (b) and (c) are standard assumptions in the study of online dictionary learning~\cite{mairal2010online, mairal2013stochastic, lyu2020online}. For instance, (b) is a common assumption in methods such as layer-wise adaptive-rate scaling (LARS)~\cite{efron2004least} that aim to find good solutions to problems of the form \eqref{eq:def_loss}. 
								See~\cite[Sec. 4.1]{mairal2010online} and~\cite[Sec. 4.1]{lyu2020online} for more detailed discussions of these assumptions. 
							}
						\end{remark}
						
						\begin{remark}
							\textup{		
								It is also possible to slightly modify both the optimization problem \eqref{eq:NDL_exact_minibatch} and our NDL algorithm so that Theorem~\ref{thm:NDL} holds for the modified problem and the algorithm without needing to assume (b) and (c). The modified problem is
								\begin{align}\label{eq:NDL_exact_minibatch_quadratic}
									\argmin_{W\in \mathcal{C}^{\text{dict}} \subseteq \R_{\ge 0}^{k^{2}\times r}}   \left( \E_{\x\sim \pi}\left[\inf_{H{\in}  \mathbb{R}^{r\times N}_{\ge 0} }  \lVert X - WH  \rVert_{F}^{2} + \lambda \lVert H \rVert_{1} + {\kappa'} \lVert H\rVert_{F}^{2} + {\lambda'}\lVert W \rVert_{F}^{2} \right] \right)\,,
								\end{align}
								where $\pi=\hat{\pi}_{F\rightarrow \G}$ if $\mathtt{MCMC} = \mathtt{PivotApprox}$ and $\pi=\pi_{F\rightarrow \G}$ otherwise. 
								The problem \eqref{eq:NDL_exact_minibatch_quadratic} is the same as problem \eqref{eq:NDL_exact_minibatch} with additional quadratic penalization terms for both $H$ and $W$ in the loss function $\ell$ (see \eqref{eq:def_loss}). By contrast, consider the modification of the NDL algorithm (see Algorithm~\ref{alg:NDL}) in which the objective function for $H_{t}$ in Algorithm~\ref{eq:def_ONMF} has the additional term $\kappa' \lVert H\rVert_{F}^{2}$ and we replace $P_{t}$ in Algorithm \ref{eq:def_ONMF} by $P_{t} + {\kappa'} I$. When $\lambda' > 0$ and $\kappa' > 0$, the modified objective function for $H_{t}$ is strictly convex, so $H_{t}$ is unique.
								Therefore, $H_t$ satisfies the uniqueness condition (b) in Theorem~\ref{thm:NDL}. Additionally, the smallest eigenvalue of each matrix $P_{t}$ that we compute using the modified NDL algorithm has a lower bound of $\kappa_{1}$ for all $t$, so it satisfies condition (c) in Theorem~\ref{thm:NDL}. One can then show that all parts of Theorem~\ref{thm:NDL} for the modified problem \eqref{eq:NDL_exact_minibatch_quadratic} and the modified NDL algorithm hold without assumptions (b) and (c).
								The argument, which we omit, is almost identical to the {proof of} Theorem~\ref{thm:NDL}.
							}
						\end{remark}

						\begin{proof}[\textbf{Proof of Theorem~\ref{thm:NDL}}]
							The proof of the first part of \textbf{(i)} is almost identical to the proof of~\cite[Corollary 6.1]{lyu2020online}.  Because Algorithm~\ref{alg:motif_inj} uses
							rejection sampling in addition to the pivot and the Glauber chains, we need to ensure that the following statements hold. First, we need the sequence of injective homomorphisms $(\x_{t})_{t\ge 0}$ that we sample using Algorithm~\ref{alg:motif_inj} to be 
									an irreducible and aperiodic Markov chain on the set of injective homomorphisms $F\hookrightarrow \G$. Second, we need $(\x_{t})_{t\ge 0}$ to have a unique stationary distribution that coincides with $\pi_{F\hookrightarrow \G}$ (see \eqref{eq:def_embedding_F_N_inj}). 
									We proved both of these statements in Proposition~\ref{prop:inj_motif_conv}. 
							
							To prove the first part of \textbf{(ii)}, we use the same essential argument as in the proof of~\cite[Corollary 6.1]{lyu2020online}. However, because our assertion is for the approximate pivot chain from the present article (see Algorithm~\ref{alg:pivot} with $\texttt{AcceptProb} = \texttt{Approximate}$), we need to use Proposition~\ref{prop:approximate_pivot} (instead of~\cite[Prop. 5.8]{lyu2023sampling}) to establish irreducibility and convergence of the associated Markov chain. The proofs of the second parts of both \textbf{(i)} and \textbf{(ii)} are identical for exact and approximate pivot chains.

							We give a detailed proof of \textbf{(ii)}. Let $\pi = \hat{\pi}_{F{\hookrightarrow} \G}$ if $\mathtt{MCMC}=\mathtt{PivotApprox}$ and $\pi=\pi_{F{\hookrightarrow }\G}$ {{for} $\mathtt{MCMC}\in \{ \mathtt{Pivot}, \mathtt{Glauber} \}$} {(see \eqref{eq:def_embedding_F_N} and \eqref{eq:approx_pivot_statioanary_dist}).} We define $(\x_{t})_{t\ge 0}$, $(X_{t})_{t\ge 1}$, and $(W_{t})_{t\ge 0}$ as before. We use a general convergence result for online NMF for Markovian data~\cite[Theorem 4.1]{lyu2020online}. 
							
							The matrices $X_{t}\in \R_{\ge 0}^{k^{2}\times N}$ that we compute in line 15 of Algorithm~\ref{alg:NDL} do not necessarily form a Markov chain because the forward evolution of the Markov chain depends both on the induced mesoscale patches and on the actual homomorphisms $(\x_{s})_{N(t-1)<s\le Nt}$. 
							However, if one considers the sequence $\overline{X}_{t}:=(X_{t}, \x_{Nt})$, then $\overline{X}_{t}$ forms a Markov chain. Specifically, the distribution of $X_{t+1}$ given $X_{t}$ depends only on $\x_{Nt}$ and $A$. Indeed, $\x_{Nt}$ and $A$ determine the distribution of the homomorphisms $(\x_{s})_{Nt<s \le N(t+1)}$, which in turn determine the $k^{2}\times N$ matrix $X_{t+1}$.
							
							 With assumption (a),~\cite[Theorems 5.7 and 5.8]{lyu2023sampling} and Proposition~\ref{prop:inj_motif_conv} imply that the Markov chain $(\x_{t})_{t\ge 0}$ of {injective} homomorphisms $F\hookrightarrow \G$ is a finite-state Markov chain that is irreducible and aperiodic with a unique stationary distribution $\pi$ {(see \eqref{eq:def_embedding_F_N_inj})}. This implies that the $N$-tuple of homomorphisms $(\x_{s})_{N(t-1)<s\le Nt}$ also yields a finite-state, irreducible, and aperiodic chain with a unique stationary distribution. Consequently, the Markov chain $(\overline{X}_{t})_{t\ge 0}$ 
							 is also a finite-state, irreducible, and aperiodic chain with a unique stationary distribution. In this setting, one can regard Algorithm~\ref{alg:NDL} as the online NMF algorithm in~\cite{lyu2020online} for the input sequence $X_{t}=\varphi(\overline{X}_{t})${, with} $t\ge 1$, where $\varphi(X, \x)=X$ is the projection onto the first coordinate. Because $\overline{X}_{t}$ takes only finitely-many values, the range of $\varphi$ is bounded. This verifies all hypotheses of~\cite[Theorem 4.1]{lyu2020online}, so the first part of \textbf{(ii)} follows.  
							
							Now suppose that there are finitely many stationary points $W_{1}^{*},\ldots,W_{m}^{*}$ of $\hat{f}$ in $\mathcal{C}^{\textup{dict}}$. Because $\lVert  W_{t-1} - W_{t} \rVert_{F}=O(1/t)$ (see~\cite[Prop. 7.5]{lyu2020online}), the first part of \textbf{(ii)} (which we proved above) implies that $W_{t}$ converges to $W_{i}^{*}$ almost surely for some unique index $i\in \{1,\ldots, m\}$.
						\end{proof}

						Our second convergence result for the NDL algorithm is similar to Theorem~\ref{thm:NDL}, but it concerns the case in which the network $\G$ is bipartite. Suppose that $\G$ is bipartite and that $F$ is a $k$-chain motif. Let $V_{1}\cup V_{2}=V$ be a bipartition of $\G$. {Let $\Omega_{0}$ denote the set of injective homomorphisms $F\hookrightarrow \G$.} We can define a subset $\Omega_{1}\subseteq {\Omega_{0}}$ of injective homomorphisms $F \hookrightarrow \G$ by $\x \in \Omega_{1}$ if and only if $\x(1)\in V_{1}$. Let $\Omega_{2}={\Omega_{0}}\setminus \Omega_{1}$. Because $\G$ is bipartite, neither the pivot chain nor the Glauber chain is
						irreducible as Markov chains with {the} state space $\Omega$. {Consequently, the injective motif-sampling {chains with the pivot chain} or the Glauber chain (see Algorithm~\ref{alg:motif_inj}) are not irreducible on the state space $\Omega_{0}$.}
						However, they are irreducible when we restrict them to each $\Omega_{i}$ with a unique stationary distribution $\pi^{(i)}_{F{\hookrightarrow} \G}$.
						(See the proof of~\cite[Theorem 5.7]{lyu2023sampling}.) 
						More explicitly, we compute
						\begin{align}\label{eq:def_embedding_F_N_bipartite}
							\pi^{(i)}_{F{\hookrightarrow} \G}( \mathbf{x} ) &:= \mathbb{P}_{\y\sim \pi_{F{\hookrightarrow} \G}}\left( \y=\x\,|\, \y\in \Omega_{i} \right) \\ 
							&=\frac{1}{\mathtt{Z}_{i}}  \left( \prod_{j \in \{2, \ldots, k\}}
							A(\mathbf{x}(j-1),\mathbf{x}(j))  \right) {\one(\textup{$\x$ is injective})} = \frac{{\mathtt{Z}_{0}}}{\mathtt{Z}_{i}} \pi_{F{\hookrightarrow} \G}(\x)\,, \quad i \in \{1,2\}\,,
						\end{align}  
						where 
						\begin{align}
							\mathtt{Z}_{i}&=\sum_{\x\in \Omega_{i}} \prod_{j \in \{{2}, \ldots, k\}}  A(\mathbf{x}({j-1}),\mathbf{x}(j)) {\one(\textup{$\x$ is injective})}, \\
							\mathtt{Z}_{0}&=\sum_{\x\in \Omega_{0}} \prod_{j \in \{{2}, \ldots, k\}}  A(\mathbf{x}({j-1}),\mathbf{x}(j)) {\one(\textup{$\x$ is injective})}.
						\end{align}
						We define two associated {conditional} expected loss functions 
						\begin{align}\label{def:expected_loss_bipartite}
							f^{(i)}(W) &:= \E_{\x\sim \pi^{(i)}_{{F\hookrightarrow \G}}}\left[\ell(X^{(N)}(\x),W)\right] \notag \\
							&=\E_{\x\sim \pi_{{F\hookrightarrow \G}}}\left[\ell(X^{(N)}(\x),W)\,\bigg| \, \x\in \Omega_{i}\right]\,, \quad i \in \{1,2\}\,,
						\end{align}
						where the equality follows from the first equality in \eqref{eq:def_embedding_F_N_bipartite}. Because the Markov chain $(\x_{t})_{t\ge 0}$ stays in $\Omega_{i}$ if it is initialized in $\Omega_{i}$, the 
						conditional expected loss functions $f^{(i)}$ are the natural objective function to minimize (instead of the 
						{expected} loss function $f$ in 
						\eqref{eq:NDL_exact_minibatch}). Similarly, for the approximate pivot chain, we define the distributions $\hat{\pi}_{{F\hookrightarrow \G}}^{(i)}$ and the {conditional} expected loss functions $\hat{f}^{(i)}$ as follows. For the probability distribution
						$\hat{\pi}_{{F\hookrightarrow \G}}$ (see \eqref{eq:approx_pivot_statioanary_dist}), let
						\begin{align}
							\hat{\pi}^{(i)}_{{F\hookrightarrow \G}}( \mathbf{x} ) &:= \mathbb{P}_{\y\sim \hat{\pi}_{{F\hookrightarrow \G}}}\left( \y=\x\,|\, \y\in \Omega_{i} \right)\\
							&=\frac{1}{|V_{i}|}  \left( \frac{\prod_{j \in \{2, \ldots, k\}} A(\mathbf{x}(j-1),\mathbf{x}(j)) }{\sum_{y_{2},\ldots,y_{k}\in V} A(\x(1),y_{2})\prod_{i = 3}^{k}A(y_{i-1},y_{i})}\right) = \frac{|V|}{|V_{1}|} \hat{\pi}_{{F\hookrightarrow \G}}(\x)\,, \label{eq:def_embedding_F_N_bipartite_approximate}\\
							f^{(i)}(W) &:= \E_{\x\sim \hat{\pi}^{(i)}_{{F\hookrightarrow \G}}}\left[\ell(X^{(N)}(\x),W)\right]\\ 
							&=\E_{\x\sim \hat{\pi}_{{F\hookrightarrow \G}}}\left[\ell(X^{(N)}(\x),W)\,\bigg| \, \x\in \Omega_{i}\right]\,, \quad i \in \{1,2\}\,. \label{def:expected_loss_bipartite_approx}
						\end{align}

						We now state our second convergence result for the NDL algorithm. This result is for bipartite networks. A convergence result that is analogous to the one in Theorem~\ref{alg:NDL} holds for the associated conditional expected loss function. 
						This implies that one can initialize homomorphisms in each $\Omega_{i}$ and compute sequences $W_{t}^{(i)}$ (with $i \in \{1,2\}$) of dictionary matrices to learn stationary points of both associated conditional expected loss functions $f^{(i)}$. However, when a $k$-chain motif $F=([k], A_{F})$ has an even number of nodes, one only needs to compute one sequence of dictionary matrices, because one can obtain the other sequence by the algebraic operation of taking a `mirror image' of a given square matrix. More precisely, define a map $\mathtt{Flip}:\R^{k^{2}\times r}\rightarrow \R^{k^{2}\times r}$ that maps $W \mapsto \overline{W}$ with
						the $j^{\textup{th}}$ column $\overline{W}(:,j)$ of $\overline{W}$ defined by 
						\begin{align}\label{eq:def_flip}
							\bar{X}(:,j):=\mathtt{vec}\circ \mathtt{rev} \circ \mathtt{reshape}(W(:,j))\,, \quad  j \in \{1, \ldots, r\}\,, 
						\end{align}
						where $W(:,j)$ denotes the $j^{\textup{th}}$ column of $W$, the operator $\mathtt{reshape}: \mathbb{R}^{k^{2}}\rightarrow \mathbb{R}^{k\times k}$ is the reshaping operator that we define in Algorithm~\ref{alg:reshape}, $\mathtt{rev}$ maps a $k \times k$ matrix $K$ to the $k \times k$ matrix {$(\bar{K}_{ab})_{a,b \in \{1, \ldots, k\}}$
						with entries $\bar{K}_{ab} = K(k-a+1, k-b+1)$}, and $\mathtt{vec}$ denotes the vectorization operator in Algorithm~\ref{alg:vectorize}. Applying $\mathtt{Flip}$ twice gives the identity map.

						\begin{theorem}[Convergence of the NDL Algorithm for Bipartite Networks]\label{thm:NDL2}
							Let $F=([k],A_{F})$ be a
							$k$-chain motif, and let $G = (V,A)$ be a network that satisfies the 
							the following properties:
							\begin{description}[leftmargin=0.7cm, topsep=0.2pt, itemsep=0.1cm]
								\item{(a')} $A$ is symmetric and the undirected and unweighted graph $(V, \one(A>0))$ is connected and bipartite. 
								
								\item{(b)} For all $t\ge 0$, there exists a unique solution $H_{t}$ in \eqref{eq:def_ONMF}.
								\item{(c)} For all $t\ge 0$, the eigenvalues of the positive semidefinite matrix $A_{t}$ in \eqref{eq:def_ONMF} are at least as large as some constant $\kappa_{1}>0$.
							\end{description}
							Let $(W_{t})_{t\ge 0}$ denote the sequence of dictionary matrices that we generate using Algorithm~\ref{alg:NDL}. 
							The following statements hold:
							\begin{description}[leftmargin=0.7cm, topsep=0.2pt, itemsep=0.1cm]
								\item[(i)] Suppose that $\mathtt{MCMC}\in \{ \mathtt{Pivot}, \mathtt{Glauber} \}$. For each $i\in \{1,2\}$, conditional on $\x_{0}\in \Omega_{i}$, the sequence {of dictionary matrices $W_{t}$} converges almost surely as $t\rightarrow \infty$ to the set of stationary points of the associated conditional expected loss function $f^{(i)}$ 
								in \eqref{def:expected_loss_bipartite}. If \,$\mathtt{MCMC} = \mathtt{PivotApprox}$, then the same statement holds with $f^{(i)}$ replaced by the function $\hat{f}^{(i)}$ 
								in \eqref{def:expected_loss_bipartite_approx}. {If we also assume that $f^{(i)}$ 
									(respectively, $\hat{f}^{(i)}$), 
									with $i\in \{1,2\}$, has
									only finitely many stationary points in $\mathcal{C}^{\textup{dict}}$, it then follows that $W_{t}$ converges almost surely to a single stationary point of $f^{(i)}$ (respectively, $\hat{f}^{(i)}$)  as $t\rightarrow \infty$.}

								\vspace{0.1cm}
								\item[(ii)] Suppose that $\mathtt{MCMC} = \mathtt{Glauber}$ in Algorithm~\ref{alg:NDL} and that $k$ is even. Assume that $\x_{0}\in \Omega_{1}$. It then follows that, almost surely as $t\rightarrow\infty$, there is simultaneous convergence of
								$W_{t}$ to the set of stationary points of $f^{(1)}$ and convergence of $\overline{W}_{t}$ to the set of stationary points of $f^{(2)}$.
								Moreover, $f^{(1)}(W_{t}) = f^{(2)}(\overline{W_{t}})$ for all $t\ge 0$. If we also assume that $f^{(i)}$, with $i\in \{1,2\}$, has only finitely many stationary points in $\mathcal{C}^{\textup{dict}}$, it then follows that 
								$W_{t}$ converges to a stationary point of $f^{(1)}$ as $t \rightarrow \infty$ and $\overline{W}_{t}$ converges to a stationary point of $f^{(2)}$ as $t \rightarrow \infty$.
							\end{description}
						\end{theorem}

						\begin{proof}

							We first prove \textbf{(i)}. 
							Fix ${i}\in \{1,2\}$ and recall the conditional stationary distribution $\pi^{(i)}_{{F\hookrightarrow \G}}$ in
							\eqref{eq:def_embedding_F_N_bipartite}. Conditional on $\x_{0}\in \Omega_{i'}$, the Markov chain $(\x_{t})_{t\ge 0}$ of injective homomorphisms is irreducible and aperiodic with a unique stationary distribution $\pi^{(i')}_{{F \hookrightarrow \G}}$. Recall that the conclusion of Theorem~\ref{thm:NDL} holds as long as the underlying Markov chain is irreducible. Therefore, $W_{t}$ converges almost surely to the set of stationary points of the associated conditional expected loss function $f^{(i')}$ 
							in \eqref{def:expected_loss_bipartite}. The same argument verifies the case with $\mathtt{MCMC} = \mathtt{PivotApprox}$.

							We now verify \textbf{(ii)}. We first establish some notation and claims. Define
							$\mu_{i} := \pi^{(i)}_{{F\hookrightarrow \G}}$ and suppose that $k$ is even. For each homomorphism $\x:F\rightarrow \G$, we define a map $\overline{\x}:[k]\rightarrow V$ by 
							\begin{align}
								\overline{\x}(j):=\x(k-j+1) \quad \text{for all \, $j \in \{1, \ldots, k\}$}\,.
							\end{align} 
							Note that $\x$ is injective if and only if {$\overline{\x}$} is injective.
							For even $k$, we have that $\x\in \Omega_{1}$ if and only if $\overline{\x} \in \Omega_{2}$. Because $A$ is symmetric, it follows that 
							\begin{align}
								\prod_{j=1}^{k-1} A(\x(j),\x(j+1))=\prod_{j=1}^{k-1} A(\x(j+1),\x(j)) = \prod_{j=1}^{k-1} A(\overline{\x}(j),\overline{\x}(j+1)) \,.
							\end{align}
							Therefore, $\mathtt{Z}_{1}=\mathtt{Z}_{2}=\mathtt{Z}_{0}/2$. Consequently, for each $\x \in \Omega_{1}$, \eqref{eq:def_embedding_F_N_bipartite} implies that 
							\begin{align}\label{eq:of_thm1_bipartite_1}
								\mu_{1}(\x) = \mu_{2}(\overline{\x}) = 2\pi_{{F\hookrightarrow \G}}(\x)\,.
							\end{align}

							Consider two Glauber chains, $({\z}_{t})_{t\ge 0}$ and $({\z}_{t}')_{t\ge 0}$, where ${\z}_{0}=\y$ and ${\z}_{0}'=\overline{\y}$. We evolve these two Markov chains using a common source of randomness so that individually they have Glauber-chain trajectories; 
							additionally,
							$\z_{t}'=\overline{\z_{t}}$ for all $t \ge 0$. We prove this by an induction on $t$. The claim clearly holds for $t = 0$. Suppose that $\z_{t}' = \overline{\z_{t}}$ for some $t \ge 0$. We want to show that $\z_{t+1}' = \overline{\z_{t+1}}$.
							For the update ${\z}_{t}\mapsto {\z}_{t+1}$ and ${\z}_{t}'\mapsto {\z}_{t+1}'$, we choose a node $v\in [k]$ uniformly at random and sample $z\in V$ according to the conditional distribution \eqref{eq:marginal_distribution_glauber}. We define 
							\begin{align}
								&\text{${\z}_{t+1}(v) = z$ \, and \, ${\z}_{t+1}(u)={\z}_{t}(u)$\,\,  for \, $u\ne v$}\,, \\
								&\text{${\z}'_{t+1}(k-v+1) = z$ \, and \, ${\z}'_{t+1}(u)={\z}_{t}(u)$ \,\,for\, $u\ne k-v+1$}\,.  
							\end{align}
							The update ${\z}_{t}\mapsto {\z}_{t+1}$ follows the Glauber-chain update in Algorithm~\ref{alg:glauber}. Additionally,
							${\z}'_{t+1}=\overline{{\z}}_{t+1}$ because 
							\begin{align}
								{\z}'_{t+1}(k-v+1) &= z = {\z}_{t+1}(v) = \overline{{\z}}_{t+1}(k-v+1)\,, \\
								{\z}'_{t+1}(u) &= {\z}_{t}'(u)=\overline{{\z}_{t}}(u) = {\z}_{t}(k-u+1) = {\z}_{t+1}(k-u+1) \\
								&=\overline{{\z}}_{t+1}(u) \,\,\, \text{for\, $u\ne k-v+1$\,.}
							\end{align}
							Finally, we need to verify that ${\z}_{t}'\mapsto {\z}_{t+1}'$ also follows the Glauber-chain update in Algorithm~\ref{alg:glauber}. It suffices to check that $z\in V$ has the same distribution as ${\z}'_{t+1}(k-v+1)$. Because $v$ is uniformly distributed on $[k]$, so is $k-v+1$. The distribution of $z\in V$ is determined by 
							\begin{align}
								p(z) \propto 
								\begin{cases}
									A(z, {\z}_{t}(2))) = A(z, \overline{{\z}_{t}}(k-1)))\,, & \text{if \quad$v=1$} \\
									\begin{matrix} \hspace{-2.65cm}A({\z}_{t}(v-1), z)A(z, {\z}_{t}(v+1)) \\  \qquad \qquad = A(\overline{{\z}_{t}}(k-v), z)A(z, \overline{{\z}_{t}}(k-v+2)) \end{matrix} \, \,, & \text{if \quad$v \in \{2, \ldots, k-1\}$}\\
									A({\z}_{t}(k-1), z)) = A(\overline{{\z}_{t}}(2), z)\,, & \text{if \quad $v = k$}\,.
								\end{cases}
							\end{align}
							Because ${\z}_{t}'=\overline{{\z}_{t}}$, it follows that $z$ follows
							the conditional distribution \eqref{eq:marginal_distribution_glauber} of ${\z}'_{t+1}(k-v+1)$, as desired.

							For the two Glauber chains, ${\z}_{t}$ and ${\z}_{t}'$, we observe that
							\begin{align}\label{eq:bipartite_flip_even}
								\overline{X^{(N)}(\y) } = X^{(N)}(\overline{\y})
							\end{align}
							almost surely. This result follows from the facts that ${\z}_{t}'=\overline{{\z}_{t}}$ for all $t\ge 0$ and $\mathtt{rev}(A_{{\z}}) = A_{\overline{{\z}}}$ for all ${\z}\in \Omega$. (See \eqref{eq:def_A_F_patch} for the definition of $A_{{\z}}$.) From this, we note that 
							\begin{align}\label{eq:def_conditional_loss_flipped}
								f^{(1)}(W) &= \E_{{\z}\sim \pi}\left[\ell(X^{(N)}({\z}), W) \,\bigg|\, {\z}\in \Omega_{1}\right] \\
								&= \E_{{\z} \sim \pi }\left[\ell\left(\overline{X^{(N)}({\z})}, \overline{W}\right) \,\bigg|\, \overline{{\z}}\in \Omega_{2}  \right] \\
								&=\E_{{\z}\sim \pi}\left[\ell\left(X^{(N)}(\overline{{\z}}), \overline{W}\right) \,\bigg|\, \overline{{\z}}\in \Omega_{2}  \right] \\
								&=\E_{{\z}\sim \pi}\left[\ell(X^{(N)}({\z}), \overline{W})   \,\bigg|\, {\z}\in \Omega_{2} \right] = f^{(2)}(\overline{W})\,.
							\end{align}
							The first and the last equalities use the second equality in \eqref{eq:def_embedding_F_N_bipartite}. The second equality uses the fact that $\ell(X,W)=\ell(\overline{X},\overline{W})$. The third equality follows from \eqref{eq:bipartite_flip_even}. The fourth equality follows from the change of variables $\overline{{\z}}\mapsto {\z}$ and the fact that ${\z}\sim \pi$ if and only if $\overline{{\z}}\sim \pi$ (see \eqref{eq:of_thm1_bipartite_1}) 
							
								We are now ready to prove \textbf{(ii)}. The first part of \textbf{(ii)} follows immediately from \textbf{(i)} and the above construction of 
								Glauber chains ${\z}_{t}$ and ${\z}_{t}'$ that satisfy ${\z}_{t}'=\overline{{\z}_{t}}$ for all $t\ge 0$. {Recall that the Markov chain of injective homomorphisms $\x_{t}$ is a subsequence of $\z_{t}$. Additionally, recall that $\z$ is injective if and only if $\overline{\z}$ is injective. Therefore, $\overline{\x}_{t}$ is the same subsequence of $\overline{z}$. That is, there exist integers $t_{l}$ (with $l \ge 1$) such that $\x_{l} = \z_{t_{l}}$} and $\overline{\x}_{l} = \overline{\z_{t_{l}}}$.
								Let $W_{t} = W_{t}({\x}_{0})$ and $W_{t}' = W_{t}'({\x}_{0}')$ denote the sequences of dictionary matrices that we compute using Algorithm~\ref{alg:NDL} with initial (not necessarily injective) homomorphisms $\x_{0}$ and $\x_{0}'$, respectively. Suppose that $\x_{0} \in \Omega_{1}$, from which we see that $\x_{0}' = \overline{\x_{0}}\in \Omega_{2}$. By \textbf{(i)}, $W_{t}$ and $W_{t}'$ converge almost surely to the set of stationary points of the associated conditional expected loss functions $f^{(1)}$ and $f^{(2)}$, respectively. We complete the proof of the first part of \textbf{(ii)} by observing that, almost surely,
								\begin{align}\label{eq:dict_flipped_bipartite}
									W_{t}'=\overline{W_{t}} \quad \text{for all \,\,$t\ge 0$}\,.
								\end{align}
								The second part of \textbf{(ii)} follows immediately from \eqref{eq:def_conditional_loss_flipped}.  
							
							We still need to verify \eqref{eq:dict_flipped_bipartite}. 
							All $k \times k$ mesoscale patches $A_{\overline{\x_{t}}} = \mathtt{ref}(A_{\x_{t}})$ have a reversed row and column ordering relative to their original ordering, so $k\times k$ latent motifs that we train on such matrices also have this reversed ordering of rows and columns. 
							One can check this claim by induction on $t$ together with \eqref{eq:bipartite_flip_even} and the uniqueness assumption (b). 
							We omit the details.  
						\end{proof}

						
						\subsection{{Convergence and reconstruction {guarantees of our} NDR algorithm}}
						\label{sec:SI_NDR}
						
						{\color{black}
							
							We prove various theoretical guarantees for our NDR algorithm (see Algorithm~\ref{alg:network_reconstruction}) in Theorem~\ref{thm:NR}. Specifically, we show that the reconstructed network that we obtain using Algorithm~\ref{alg:network_reconstruction} at iteration $t$ converges almost surely to some limiting network as $t\rightarrow \infty$, and we give a closed-form expression of the limiting network. {We also derive an upper bound of
							the 
							{reconstruction error}. 
							Roughly, we state the bound as follows:
								\begin{align}
									\textup{Jaccard reconstruction error} \le \frac{\textup{{mesoscale} approximation error} }{2(k-1)} \,,
								\end{align}
								where $k$ denotes the number of nodes in a $k$-chain motif. The parameter $k$ is effectively a 
								scale parameter. In \eqref{eq:NR_Jaccard_bound}, we give a precise statement of the above bound.}

							Before stating our mathematical results, we first introduce some notation. Fix a network $\G=(V,A)$, a $k$-chain motif 
							$F = ([k], A_{F})$, and a homomorphism $\x:F\rightarrow \G$. 
							In this discussion, we do not assume that a homomorphism $F \rightarrow \G$ is injective, as we use all sampled homomorphisms for our network reconstruction using our NDR algorithm (see Algorithm~\ref{alg:network_reconstruction}), unlike in our NDL algorithm (see Algorithm~\ref{alg:NDL}) for learning latent motifs. 
							Let $\mathtt{denoising}$ denote the Boolean variable in Algorithm~\ref{alg:network_reconstruction}. For each matrix $B:V^{2}\rightarrow [0,\infty)$ and a node map $\x:[k]\rightarrow V$, define the $k\times k$ matrix $B_{\x}$ by 
							\begin{align}
								B_{\x}(a,b) := B(\x(a), \x(b))  \quad \text{for all \,\, $a,b \in \{1, \ldots, k\}$}\,.
							\end{align} 
							If $B=A$, then $B_{\x}=A_{\x}$ is the same as
							the mesoscale patch of $\G$ that is induced by $\x$ (see \eqref{eq:def_A_F_patch}). Additionally, given a network $\G=(V,A)$, a $k$-chain motif $F=([k], A_{F})$, a homomorphism $\x:F\rightarrow \G$, and a nonnegative matrix $W\in \R_{\ge 0}^{k^{2}\times r}$, let $\hat{A}_{\x;W}$ denote the $k \times k$ matrix that we defined in line 16 of Algorithm~\ref{alg:network_reconstruction}. 
							This matrix depends on the Boolean variable $\mathtt{denoising}$. Recall that $\hat{A}_{\x;W}$ is a nonnegative linear approximation of $A_{\x}$ that uses $W$.

						We consider
						the event $(x,y)\overset{\x} \hookleftarrow (a, b)$ using the following indicator function:
							\begin{align}\label{eq:def_indicator_visit}
								\one\bigg(\text{$(x,y)\overset{\x}{\hookleftarrow} (a, b)$}\bigg) &:= \one(\x(a)=x, \,\x(b)=y)  \\
								&\hspace{1cm} \times \one\left(\begin{matrix} \text{$\mathtt{InjHom}=\texttt{F}$} \\ \text{or $\x$ is injective} \end{matrix}\right)\,  \one\left( \begin{matrix} \text{$\mathtt{denoising}=\texttt{F}$} \\ \text{or $A_{F}(a,b)=0$} \end{matrix}\right) \,.
							\end{align}
							For each homomorphism $\x:F\rightarrow \G$ and $x,y \in V$, we say that the pair $(x,y)$ is \textit{visited by $(a,b)$ through $\x$} whenever the indicator on the left-hand side of \eqref{eq:def_indicator_visit} is $1$. Additionally,
							\begin{align}\label{eq:def_N_pq}
								N_{xy}(\x) := \sum_{a,b \in \{1, \ldots, k\}} \one\bigg(\text{$(x,y)\overset{\x}{\hookleftarrow} (a, b)$}\bigg)
							\end{align} 
							is the total number of visits to $(x,y)$ through $\x$. When $N_{xy}(\x)>0$, we say that the pair $(x,y)$ is \textit{visited by $\x$}. In Algorithm~\ref{alg:network_reconstruction}, observe that both $A_{\textup{count}}(x,y)$ and $A_{\textup{recons}}(x,y)$ change at iteration $t$ if and only if $N_{xy}(\x_{t})>0$. {Let}
							\begin{align}\label{eq:def_Omega_pq}
								\Omega_{xy} := \left\{ \x:F\rightarrow \G\,\big|\, N_{xy}(\x)>0\right\}  
							\end{align}
							{denote} the set of all homomorphisms $\x:F\rightarrow \G$ that visit the pair $(x,y)$.

							Suppose that $\hat{G}=(V,\hat{A})$ is a reconstructed network for $\G=(V,A)$. {Fix a probability distribution $\pi$ on the set of homomorphisms $F\rightarrow \G$. For a matrix $Q:V^{2}\rightarrow \R$, define the weighted $L_{1}$ norm 
								\begin{align}
									\lVert Q \rVert_{1,\pi} := \sum_{x,y\in V} |Q(x,y)|  \, \E_{\x\sim \pi}[ N_{xy}(\x) ] \,.
								\end{align}
								We define the following two quantities:
								\begin{align}\label{eq:JD}
									{\JD_{\pi}(\G,\hat{\G}) := \frac{ \lVert A - \hat{A} \rVert_{1,\pi } }{   \lVert A \lor \hat{A} \rVert_{1,\pi}  }\,, \quad   \JD(\G,\hat{\G}) } := \frac{ \lVert A - \hat{A} \rVert_{1 } }{   \lVert A \lor \hat{A} \rVert_{1}  }\,,
								\end{align}
								where $A \lor \hat{A}$ is defined as $(A\lor \hat{A})(x,y) = A(x,y)\lor \hat{A}(x,y)=\max\{ A(x,y), \hat{A}(x,y) \}$ for $x,y\in V$. 
								We refer to $\JD_{\pi}(\G, \hat{\G})$ as the \textit{Jaccard distance} between $\G=(V,A)$ and $\hat{G} = (V,\hat{A})$ with respect to $\pi$. 
									We refer to $\JD(\G,\hat{\G})$ as the \textit{unweighted Jaccard distance} between $\G=(V,A)$ and $\hat{G}=(V,\hat{A})$.

								To make sense of the definitions in \eqref{eq:JD}, consider the special case in which the weight matrices $A$ and $\hat{A}$ are both symmetric and binary with $0$ diagonal entries (i.e., no self-edges). We also assume that $k = 2$ and that $\pi$ is the uniform distribution on the set of homomorphisms $F\rightarrow \G$.
								 We then have that 
								\begin{equation} \label{above}
									\JD_{\pi}(\G, \hat{\G}) = \JD(\G, \hat{\G}) = 1 - \frac{ \textup{{number of} edges in both $\G$ and $\hat{\G}$}}{\textup{{number of} edges in $\G$ {or} 
											$\hat{\G}$}} \,.
								\end{equation}
								The fraction in the right-hand side of \eqref{above} 
								equals the Jaccard index for the edge sets of $\G$ and $\hat{\G}$. We used this Jaccard index as a reconstruction-accuracy measure in Figure~\ref{fig:network_recons} of the main manuscript. Consequently, it is reasonable to view the Jaccard distances $\JD_{\pi}(\G, \hat{\G})$ and $\JD(\G, \hat{\G})$, which coincide in this case, as the `Jaccard reconstruction error' of reconstructing $\G$ as $\hat{\G}$.}

							In the following proposition, we show that, under mild conditions, {the Jaccard distance $\JD_{\pi}(\G, \G')$ with respect to $\pi$ is close to the unweighted Jaccard distance $\JD(\G, \G')$}  if the weights $\E_{\x\sim \pi}[ N_{xy}(\x) ]$ do not vary much for node pairs {$(x,y)\in V^{2}$} with $|\Omega_{xy}| \ge 1$.

							\begin{prop}\label{prop:Jaccard_distance}
								Let $F=([k],A_{F})$ be {a} $k$-chain motif, and fix connected and symmetric networks $\G=(V,A)$ and $\hat{G}=(V,\hat{A})$ {with} the same node set $V$. Let $\pi$ be a probability distribution on the set of
								homomorphisms $F\rightarrow \G$. Suppose that the following conditions hold:
								\begin{description}[leftmargin=0.7cm, topsep=0.1cm, itemsep=0.1cm]
									\item{(a)} If $x,y\in V$ {satisfies}
									$|\Omega_{xy}| \ge 1$, {it follows that} $\hat{A}(x,y)=0$; 
									
									\item{(b)} Consider a homomorphism $\x:F\rightarrow \G$ {with} $\pi(\x)>0$. 
								\end{description}
								Let $\rho$ denote the maximum value of the ratio $\E_{\x\sim \pi}[N_{xy}(\x)]/\E_{\x\sim \pi}[N_{x'y'}(\x)]$ for $x,x',y,y'\in V$ {and suppose that $\E_{\x\sim \pi}[N_{x'y'}(\x)] \neq 0$.}
								{We {then} have that 
									\begin{align}
										\rho^{-1} {\JD(\G, \hat{\G})}  \le 	\JD_{\pi}(\G, \hat{\G}) \le  \rho {\JD(\G, \hat{\G})}\,.
									\end{align}
								}
							\end{prop}

							\begin{proof}
								Consider two nodes $x,y\in V$, and let $d_{\G}(x,y)$ denote the {shortest-path distance between them} in $\G$. {This distance is} the minimum number of edges in a walk in $\G$ that connects $x$ and $y$. 
								If $d_{\G}(x,y)>k$, we then have that $A(x,y)=\hat{A}(x,y)=0$. Moreover, because $\G$ is symmetric, $d_{\G}(x,y)\le k$ implies that there exists a homomorphism $\x:F\rightarrow \G$. By condition (b), it follows that $\E_{\x\sim \pi } \left[ N_{xy}(\x) \right] > 0$ if $d_{\G}(x,y)\le k$. {Let} $c:= \min_{x,y\in V,\, |\Omega_{pq}| \ge 1 } \E_{\x\sim \pi}[N_{xy}(\x)]$. We then have 
								\begin{align} 
									c \, \lVert A - \hat{A} \rVert_{1}   &=  	\sum_{x,y\in V,\, d_{\G}(x,y)\le k }   \left|  A(x,y) - \hat{A}(x,y)  \right| c  \\ 
									&\le  	\sum_{x,y\in V,\, d_{\G}(x,y)\le k } \left|  A(x,y) - \hat{A}(x,y)  \right| \, \E_{\x\sim \pi}[N_{xy}(\x)]  \\ 
									&=  	\sum_{x,y\in V } \left|  A(x,y) - \hat{A}(x,y)  \right| \, \E_{\x\sim \pi}[N_{xy}(\x)]  = 	\lVert A - \hat{A} \rVert_{1,\pi} \,.
								\end{align}
								If we instead take $C:= \max_{x,y\in V,\, |\Omega_{pq}| \ge 1 } \E_{\x\sim \pi}[N_{xy}(\x)]$, we obtain  
								\begin{align} 
									C \lVert A - \hat{A} \rVert_{1}   &=  	\sum_{x,y\in V,\, d_{\G}(x,y)\le k }   \left|  A(x,y) - \hat{A}(x,y)  \right|  C  \\ 
									&\ge  	\sum_{x,y\in V,\, d_{\G}(x,y)\le k } \left|  A(x,y) - \hat{A}(x,y)  \right| \, \E_{\x\sim \pi}[N_{xy}(\x)]  = \lVert A - \hat{A} \rVert_{1,\pi} \,.
								\end{align}
								Therefore, $C \lVert A - \hat{A} \rVert_{1} \ge 	\lVert A - \hat{A} \rVert_{1,\pi}$.
								Moreover, by using a similar argument, we obtain $c	\lVert A \lor \hat{A} \rVert_{1} \le \lVert A - \hat{A} \rVert_{1,\pi} \le C \lVert A \lor \hat{A} \rVert_{1}$. We then let $\rho=C/c$ to complete the proof.
							\end{proof}


							{
								\begin{prop}\label{prop:inj_hom_coverage}
									Let $G=(V,A)$ be a symmetric and connected network, and let $F=([k],A_{F})$ be a
									$k$-chain motif. Suppose that $2(k-1)\le \textup{diam}(\G)$. For each node $x\in V$, it then follows that there exists an injective homomorphism $\x:F\hookrightarrow \G$ {with} $x\in \{ \x(1),\dots,\x(k) \}$. 
								\end{prop}
								
								\begin{proof}
									Our proof proceeds by contradiction.
									Suppose that there is a node $x \in V$ 
									for which there does not {exist} an injective homomorphism $\x:F\hookrightarrow \G$ with
									$x\in \{ \x(1),\dots,\x(k) \}$. We will show that this implies that $\textup{diam}(\G) < 2(k-1)$, which contradicts the hypothesis.

									Fix two nodes $y, y' \in V$. Consider a walk $(y_{1},\dots,y_{\bar{a}})$ from $x$ to $y$ (with $y_{1} = x$ and $y_{\bar{a}} = y$) and another walk $(z_{1},\dots,z_{\bar{b}})$ from $z$ to $x$  (with $z_{1} = y'$ and $z_{\bar{b}} = x$). We choose these walks so that the integers $\bar{a}, \bar{b} \ge 1$ are as small as possible. By minimality, we 
									 know that these walks are paths (i.e., all $y_{i}$ are distinct and all $z_{i}$ are distinct). If $\bar{a} \ge k$, then $(y_{1},\ldots,y_{k})$ gives a $k$-path, so $\x_{y}:[k]\rightarrow V$ with $\x_{y}(i)=y_{i}$ for $i\in \{1,\ldots,k\}$ is an injective homomorphism $F \hookrightarrow \G$ and $\x_{y}(1) = y_{1} = x$. This contradicts our assumption about the node $x$ that there is no injective homomorphism $\x:F\hookrightarrow \G$ with
									$x \in \{ \x(1),\ldots,\x(k) \}$. 
									Therefore, $\bar{a} < k$. By a similar argument, we know that $\bar{b} < k$. Consequently, $(z_{1},\ldots,z_{\bar{b}},y_{1},\ldots,y_{\bar{a}})$ is a $(\bar{b} + \bar{a})$-path from $y'$ to $y$. This implies that $d_{\G}(z,y) \le \bar{b} + \bar{a} < 2(k-1)$. Because the nodes $y$ and $y'$ are arbitrary, this implies that $\textup{diam}(\G) < 2(k-1)$, which contradicts the hypothesis. This completes the proof.
								\end{proof}
							}
							
						}

						We now state and prove our main theoretical result about our NDR algorithm (see Algorithm~\ref{alg:network_reconstruction}).

						\begin{theorem}[{Guarantees of the} NDR Algorithm (see Algorithm~\ref{alg:network_reconstruction}) for Non-Bipartite Networks]\label{thm:NR}
							Let $F=([k],A_{F})$ be a $k$-chain motif, and fix a network $\G = (V,A)$ and a network dictionary $W\in \mathbb{R}^{k^{2}\times r}$. We use Algorithm~\ref{alg:network_reconstruction} with inputs $\G$, $F$, and $W$ and the parameter value $T = \infty$. Let $\hat{\G}_{t} = (V, \hat{A}_{t})$ denote the network that we reconstruct at iteration $t$, and suppose that $\G$ satisfies assumption (a) of Theorem~\ref{thm:NDL}. Let {\color{black}
								\begin{align}
									\pi := \begin{cases}
										\pi_{F\rightarrow \G} & \textup{if $\mathtt{MCMC} \in \{ \mathtt{Glauber}, \mathtt{pivot}  \}$ and $\mathtt{InjHom} = \mathtt{F}$} \\
										\pi_{F\hookrightarrow \G} & \textup{if $\mathtt{MCMC} \in \{ \mathtt{Glauber}, \mathtt{pivot}  \}$ and $\mathtt{InjHom} = \mathtt{T}$ } \\
										\hat{\pi}_{F\rightarrow \G} & \textup{if $\mathtt{MCMC} = \mathtt{PivotApprox}$ and $\mathtt{InjHom} = \mathtt{F}$} \\
										\hat{\pi}_{F\hookrightarrow \G} & \textup{if $\mathtt{MCMC} = \mathtt{PivotApprox}$ and $\mathtt{InjHom} = \mathtt{T}$} {\,.} \\
									\end{cases}
								\end{align}
							}
							The following statements hold:
							\vspace{0.1cm}
							\begin{description}[leftmargin=0.7cm, topsep=0.2pt, itemsep=0.1cm]
								\item[(i)] (Convergence of the network reconstruction)
								The network $\hat{\G}_{t}$ converges almost surely to some limiting network $\hat{\G}_{\infty}=(V, \hat{A}_{\infty})$ in the sense that 
								\begin{align}
									\lim_{t\rightarrow \infty} \hat{A}_{t}(x,y) = \hat{A}_{\infty}(x,y)\in [0,\infty) \quad \text{almost surely for all $x,y \in V$}\,. 
								\end{align}

								\item[(ii)] {(Limiting reconstructed network)} Let $\hat{A}_{\infty}$ denote the limiting matrix in \textbf{\textup{(i)}}.  For each $x,y \in V$, we then have that 
								\begin{align}\label{eq:formula_limiting_recons_network}
									\hat{A}_{\infty}(x,y) = 
									\sum_{\y\in \Omega_{xy}}\left[ \sum_{a,b \in \{1, \ldots, k\}} \hat{A}_{\y;W}(a,b) \,\one\bigg(\text{$(x,y)\overset{\y}{\hookleftarrow} (a, b)$}\bigg)\right] \,\frac{ \pi(\y)}{ \E_{\x\sim \pi}\left[ N_{xy}(\x) \right]} \,,
								\end{align}
								where we regard the right-hand side to be {$0$} when $|\Omega_{xy}| = 0$.
								
								\item[(iii)]  {(Bounds for the Jaccard reconstruction error) Suppose that the range of $A$ is contained in $\{0\}\cup [1,\infty)$ and assume that $\mathtt{denoising}=\texttt{F}$.  Let $\hat{A}_{\infty}$ denote the limiting matrix in \textbf{\textup{(i)}}. {We then have that}
										\begin{align}\label{eq:NR_Jaccard_bound}
											\rho^{-1}  {\textup{JD}(\G,\hat{\G})}  \le 	\textup{JD}_{\pi}(\G,\hat{\G}) \le \frac{\mathbb{E}_{\x\sim \pi}\left[ \lVert A_{\x} - \hat{A}_{\x;W} \rVert_{1} \right] }{ 2(k-1)} \,,
										\end{align}
										where the constant $\rho > 0$  is as in Proposition~\ref{prop:Jaccard_distance}.} 
							\end{description}
						\end{theorem}

						\begin{proof}

							Let $\x$ denote a random homomorphism $F\rightarrow \G$ with distribution $\pi$, and let $\P$ and $\mathbb{E}$ denote the associated probability measure and expectation, respectively. 
							
							We first verify \textbf{(i)} and \textbf{(ii)} simultaneously. Let $(\x_{t})_{t\ge 0}$ denote the Markov chain that we generate during the reconstruction process (see Algorithm~\ref{alg:network_reconstruction}). 
							Each $\x_{t}$ is an injective homomorphism $F\hookrightarrow \G$ if $\mathtt{InjHom} = \mathtt{T}$, and each $\x_{t}$
							is a homomorphism $F \rightarrow \G$ that may or may not be injective if $\mathtt{InjHom} = \mathtt{F}$. 
							We fix $x,y \in V$ and let 
							\begin{align}
								M_{t} = \sum_{s=1}^{t} \sum_{a,b \in \{1, \ldots, k\}}  \one\bigg(\text{$(x,y)\overset{\x_{s}}{\hookleftarrow} (a, b)$}\bigg) = \sum_{s=1}^{t} N_{xy}(\x_{s})\,,
							\end{align}
							where we defined the indicator $\one\big(\text{$(x,y)\overset{\x}{\hookleftarrow} (a, b)$}\big)$ in \eqref{eq:def_indicator_visit}. If $M_{t}=0$, then $\hat{A}_{t}(x,y)=0$. Suppose that $M_{t}\ge 1$. 
							The key observation is that 
							\begin{align}\label{eq:hatA_t_formula}
								\hat{A}_{t}(x,y) &= \frac{1}{M_{t}} \sum_{s=1}^{t} \sum_{a,b \in \{1, \ldots, k\}} \hat{A}_{\x_{s};W}(a,b)\one\bigg(\text{$(x,y)\overset{\x_{s}}{\hookleftarrow} (a, b)$}\bigg)\\
								&= \sum_{a,b \in \{1, \ldots, k\}} \frac{1}{M_{t}} \sum_{s=1}^{t} \sum_{\y\in \Omega_{xy}}  \hat{A}_{\y;W}(a,b)\one\bigg(\text{$(x,y)\overset{\x_{s}}{\hookleftarrow} (a, b)$}\bigg) \one (\x_{s}=\y) \\
								&= \sum_{\y\in \Omega_{xy}}  \sum_{a,b \in \{1, \ldots, k\}}   \hat{A}_{\y;W}(a,b)\one\bigg(\text{$(x,y)\overset{\y}{\hookleftarrow} (a, b)$}\bigg) \frac{t}{M_{t}} \,\frac{1}{t} \sum_{s=1}^{t} \one(\x_{s}=\y)\,.
							\end{align}
							With assumption (a), the Markov chain $(\x_{t})_{t\ge 0}$ of homomorphisms $F \rightarrow \G$ is irreducible and aperiodic with the unique stationary distribution $\pi$ (see \eqref{eq:def_embedding_F_N}).
							By the Markov-chain ergodic theorem (see, e.g.,~\cite[Theorem 6.2.1 and Example 6.2.4]{Durrett} or~\cite[Theorem 17.1.7]{meyn2012markov}), it follows that 
							\begin{align}
								\lim_{t\rightarrow \infty}	\frac{t}{M_{t}}  \frac{1}{t} \sum_{s=1}^{t} \one(\x_{s}=\y) =  \frac{\P\left( \x=\y \right)}{ \E\left[ N_{xy}(\x) \right] }\,.
							\end{align}
							By the definition of the probability distribution $\pi$, we have that $\pi(\x) > 0$ for all injective homomorphism $\x:F\hookrightarrow \G$. Therefore, $\E_{\x\sim \pi }[N_{xy}(\x)]>0$ if and only if $|\Omega_{xy}|\ge 1$.   This proves both \textbf{(i)} and \textbf{(ii)}.

							{\color{black}
								We now prove \textbf{(iii)}. {Conditions (a) and (b) of Proposition~\ref{prop:Jaccard_distance} are satisfied for $\hat{G}=\hat{G}_{\infty}$ 
									with the assumed choice of $\pi$.
									Therefore, the first inequality in \eqref{eq:NR_Jaccard_bound} follows immediately from Proposition~\ref{prop:Jaccard_distance}. To verify the second inequality, we prove a slightly more general result. Let $\G'=(V,B)$ be a network with the same node set $V$ as $\G=(V,A)$. Assume that $B$ is symmetric and that its range is contained in $\{0\}\cup [1,\infty)$.} For each $a,b \in \{1, \ldots, k\}$, {define} the {indicator function}
								\begin{align}\label{eq:def_1ab}
									\one_{ab}:= \one\left( \begin{matrix} \text{$\mathtt{denoising}=\texttt{F}$} \\ \text{or $A_{F}(a,b)=0$} \end{matrix}\right)\,.
								\end{align}
								We will show that 
								\begin{align}\label{eq:NDR_bound_gen_pf}
									\frac{ \lVert B - \hat{A}_{\infty} \rVert_{1, \pi } }{   \lVert B \lor \hat{A}_{\infty} \rVert_{1, \pi}  } \le \frac{1}{2(k-1)}  \sum_{\y\in \Omega} \sum_{a,b \in \{1, \ldots, k\}}  \left| B_{\y}(a,b)\mathbb{1}_{ab}-\hat{A}_{\y;W}(a,b)\mathbb{1}_{ab}\right|\pi(\y) \,.
								\end{align}
								If $\mathtt{denoising}=\texttt{F}$, the right-hand side of \eqref{eq:NDR_bound_gen_pf} reduces to $\frac{1}{2(k-1)}\E_{\x\sim \pi}[ \lVert B_{\x} - \hat{A}_{\x;W} \rVert_{1} ]$,  so \textbf{(iii)} is a special case of \eqref{eq:NDR_bound_gen_pf}.

								To verify \eqref{eq:NDR_bound_gen_pf}, we first claim that 
								\begin{align}\label{eq:NDR_bound_recons}
									\lVert B - \hat{A}_{\infty} \rVert_{1,\pi}  \le 
									\sum_{\y\in \Omega} \sum_{a,b \in \{1, \ldots, k\}}  \left| B_{\y}(a,b)-\hat{A}_{\y;W}(a,b)\right|\pi(\y) \,.
								\end{align}
								For each $a,b \in \{1, \ldots, k\}$ and $x,y \in \{1, \ldots, n\}$, let $\Omega_{ab\rightarrow xy}$ denote the set of 
								homomorphisms $\x:F\rightarrow \G$ such that $\one\big(\text{$(x,y)\overset{\x}{\hookleftarrow} (a, b)$}\big)=1$. By changing the order of the sums,
								we rewrite the formula in \eqref{eq:formula_limiting_recons_network} as 
								\begin{align}
									\hat{A}_{\infty}(x,y) = \sum_{a,b \in \{1, \ldots, k\}} \sum_{\y\in \Omega_{ab\rightarrow xy}}   \hat{A}_{\y;W}(a,b) \frac{\P \left( \x=\y \right)}{\mathbb{E}\left[ N_{xy}(\x) \right]}\,.
								\end{align}
								Additionally, observe that 
								\begin{align}
									\E\left[ N_{xy}(\x) \right] &= {\E}\left[  \sum_{a,b \in \{1, \ldots, k\}} \one(\x(a)=x, \,\x(b)=y) \, {\one\left( \begin{matrix} \text{$\mathtt{InjHom}=\texttt{F}$} \\ \text{or $\x$ is injective} \end{matrix}\right)\, }    \right] \\
									&=\sum_{a,b \in \{1, \ldots, k\}}\sum_{\y\in \Omega_{ab\rightarrow xy}} \P\left( \x=\y\right)\,. 
								\end{align}
								We now calculate 
								\begin{align}
									&\sum_{x,y\in V} \left| B(x,y) - \hat{A}_{\infty}(x,y)   \right| \mathbb{E}\left[ N_{xy}(\x) \right] \\
									&\qquad \qquad =  \sum_{x,y\in V} \left| B(x,y)\, \mathbb{E}\left[ N_{xy}(\x) \right] -  \sum_{a,b \in \{1, \ldots, k\}} \sum_{\y\in \Omega_{ab\rightarrow xy}} \hat{A}_{\y;W}(a,b) \, \P \left( \x=\y \right) \right| \\
									&\qquad\qquad=  \sum_{x,y\in V} \left|\sum_{a,b \in \{1, \ldots, k\}} \sum_{\y\in \Omega_{ab\rightarrow xy}} \left(B(x,y)-\hat{A}_{\y;W}(a,b)\right) \P\left( \x=\y \right)  \right| \\
									&\qquad\qquad\le \sum_{x,y\in V} \sum_{a,b \in \{1, \ldots, k\}}  \sum_{\y\in \Omega_{ab\rightarrow xy}} \left| B(\y(a),\y(b))-\hat{A}_{\y;W}(a,b)\right| \P\left( \x=\y \right)   \\
									&\qquad\qquad= \sum_{x,y\in V}\sum_{a,b \in \{1, \ldots, k\}} \sum_{\y\in \Omega}  \left| B(\y(a),\y(b))  -\hat{A}_{\y;W}(a,b) \right| \\
									&\hspace{7cm} \times \P\left( \x=\y \right)  \mathbb{1}(\y(a)=x, \, \y(b)=y)\\
									&\qquad\qquad= \sum_{\y\in \Omega} \P\left( \x=\y \right)  \sum_{a,b \in \{1, \ldots, k\}}  \left| B_{\y}(a,b)-\hat{A}_{\y;W}(a,b)\right| \sum_{x,y\in V} \mathbb{1}(\y(a)=x, \, \y(b)=y) \\
									&\qquad\qquad= \sum_{\y\in \Omega} \sum_{a,b \in \{1, \ldots, k\}}  \left| B_{\y}(a,b)-\hat{A}_{\y;W}(a,b)\right|\pi(\y)\,,
								\end{align}
								where 
								the indicator $\one_{ab}$ is defined in \eqref{eq:def_1ab}.
								This verifies the claim \eqref{eq:NDR_bound_recons}.

								It now suffices to show that 
								\begin{align}\label{eq:NDR_bd_bottom}
									\lVert B \lor \hat{A}_{\infty} \rVert_{1,\pi} \ge 2 \lVert A_{F}  \rVert_{1} = 2(k-1) \,,
								\end{align}
								where the equality 
								uses the fact that $F$ is a $k$-chain motif. For each $a,b\in \{1,\dots,k\}$ and a homomorphism $\x:F\rightarrow \G$, observe that
								\begin{align}\label{eq:pf_NDR_identity}
									(A_{F}(a,b)+A_{F}(b,a))  	\sum_{x,y\in V}     B(x,y)    \sum_{\y\in \Omega_{ab\rightarrow xy}} \one( \x=\y ) \ge A_{F}(a,b)+A_{F}(b,a) \,. 
								\end{align}
								The inequality \eqref{eq:pf_NDR_identity} uses the fact that $\x$ is a homomorphism. Therefore, $A_{F}(a,b)+A_{F}(b,a)>0$ and $\{ \x(a), \x(b) \} = \{x,y\}$ implies that $B(x,y)+B(y,x)>0$. Because {we assume that $B$ is symmetric} 
								and that the range of $B$ is contained in $\{0\}\cup [{1}, \infty)$, 
								it follows that $B(x,y)\ge 1$ and thus verifies the inequality \eqref{eq:pf_NDR_identity}.

								Recall the notation $\Omega_{ab\rightarrow xy}$ below the inequality \eqref{eq:NDR_bound_recons}, and observe that 
									\begin{align}
										\lVert B \lor \hat{A}_{\infty} \rVert_{1,\pi} &=\sum_{x,y\in V} \left| B(x,y) \lor \hat{A}_{\infty}(x,y)  \right| \E_{\x\sim \pi}\left[  N_{xy}(\x) \right]   \\
										& =  \E_{\x\sim \pi}\left[  	\sum_{x,y\in V}  \left| B(x,y) \lor \hat{A}_{\infty}(x,y)  \right| N_{xy}(\x) \right]   \\
										& =  \E_{\x\sim \pi}\left[  	\sum_{x,y\in V}   \left( B(x,y) \lor \hat{A}_{\infty}(x,y)  \right)  \sum_{a,b \in \{1, \ldots, k\}} \sum_{\y\in \Omega_{ab\rightarrow xy}} \one( \x=\y )  \right] \\
										& \ge  \E_{\x\sim \pi}\left[    \sum_{a,b \in \{1, \ldots, k\}} 	\sum_{x,y\in V}     B(x,y)    \sum_{\y\in \Omega_{ab\rightarrow xy}} \one( \x=\y )  \right] \\
										& \ge  \E_{\x\sim \pi}\left[    \sum_{a,b \in \{1, \ldots, k\}}  (A_{F}(a,b)+A_{F}(b,a)) 	\sum_{x,y\in V}     B(x,y)    \sum_{\y\in \Omega_{ab\rightarrow xy}} \one( \x=\y )  \right] \\
										& =   \E_{\x\sim \pi}\left[    \sum_{a,b \in \{1, \ldots, k\}} (A_{F}(a,b)+A_{F}(b,a))  \right] =2 \lVert A_{F} \rVert_{1}\,. 
									\end{align}
									For the last inequality, we have used the fact that $(A_{F}(a,b)+A_{F}(b,a)) \in \{ 0,1 \}$. This proves \eqref{eq:NDR_bd_bottom}, as desired. } 
						\end{proof}

						{
							\begin{remark}
								\normalfont 
									Suppose that the original network $\G=(V,A)$ is binary. The bound on the Jaccard reconstruction error in Theorem~\ref{thm:NR}\textup{\textbf{(iii)}} is for a direct comparison between the weighted reconstructed network $\hat{G}=(V,\hat{A})$ and the original binary network $\G=(V,A)$. In Figure~\ref{fig:network_recons} of the main manuscript, we instead used binary reconstructed networks $\hat{G}(\theta):=(V, \one(\hat{A} > \theta))$ that we obtain by thresholding the edge weights $\hat{A}(x,y)$ at some threshold $\theta\in (0,1)$. By modifying the argument in the proof of Theorem~\ref{thm:NR}\textup{\textbf{(iii)}, we obtain a similar bound for} the Jaccard reconstruction error for the thresholded reconstructed network $\hat{G}(\theta):=(V, \one(\hat{A} > \theta)  )$. We now sketch the argument.

									Let $\theta' := \min\{1-\theta, \, \theta\}$. For $\tilde{a} \in \{0,1\}$ and $\hat{a} \in [0,1]$, we obtain the inequalities 
									\begin{align}
										\theta' | \tilde{a} - \one(\hat{a}>\theta)  | \le | ( \tilde{a} - \hat{a}) \one(|\tilde{a}-\hat{a} |>  \theta' )  |  \le |  \tilde{a} - \hat{a} | \,.
									\end{align}
									Because $A: V^{2}\rightarrow \{0,1\}$ and $\hat{A}:V^{2}\rightarrow [0,1]$, it follows that 
									\begin{align}
										\theta' \lVert A - \one(\hat{A}>\theta)  \rVert_{1,\pi} \le \lVert ( A - \hat{A}) \one(|A-\hat{A} |>  \theta' )   \rVert_{1,\pi} \le \lVert A - \hat{A} \rVert_{1,\pi} \,.
									\end{align}
									By modifying the argument in the proof of Theorem~\ref{thm:NR}\textup{\textbf{(iii)}}, one can show that 
									\begin{align}\label{eq:NDR_recons_bound_binary_rmk}
										\theta'  \JD_{\pi}(\G, \hat{\G}(\theta)) \le \frac{ \lVert ( A - \hat{A}) \one(|A-\hat{A} |>  \theta' )   \rVert_{1}  }{ \Vert A \rVert_{1,\pi}}  \overset{(*)}{\le}  \frac{ \lVert  A - \hat{A}     \rVert_{1}  }{ \Vert A \rVert_{1,\pi}} \le \frac{	\mathbb{E}_{\x\sim \pi}\left[ \lVert A_{\x} - \hat{A}_{\x;W} \rVert_{1} \right] }{ 2(k-1)} \,.
									\end{align}
									By using Proposition~\ref{prop:Jaccard_distance}, one can also deduce that
									\begin{align}\label{eq:NR_Jaccard_bound_binary}
										{\JD(\G, \hat{\G}(\theta) )} \le (\rho/\theta') \, \frac{	\mathbb{E}_{\x\sim \pi}\left[ \lVert A_{\x} - \hat{A}_{\x;W} \rVert_{1} \right] }{ 2(k-1)} \,.
									\end{align}
									The inequality \eqref{eq:NR_Jaccard_bound_binary} gives a bound for the unweighted Jaccard distance between the original binary network $\G$ and the thresholded binary reconstructed network $\hat{\G}(\theta)$. 
										However, the bound \eqref{eq:NR_Jaccard_bound_binary} 
										is not sharp because of
											the possibly large constant $\rho/\theta'$ (which is at least $2.5$ for $\theta = 0.4$). 
											For instance, in Figure~\ref{fig:recons_bd_plot}\textbf{c}, we see that if $\G$ is \dataset{UCLA} and $W$ is the network dictionary of $r = 9$ latent motifs of \dataset{UCLA}, then the right-hand side of \eqref{eq:NR_Jaccard_bound_binary} is at least $2.5 \times 0.2 = 0.5$ for $\theta = 0.4$. However, in Figure~\ref{fig:network_recons}\textbf{c}, we see the empirical value of the left-hand side of \eqref{eq:NR_Jaccard_bound_binary} is about $0.05$. 
											To obtain some insight into this discrepancy, observe that the second inequality in \eqref{eq:NDR_recons_bound_binary_rmk} (which we mark with $(*)$) becomes very crude if many entries of $A$ and $\hat{A}$ do not differ by more than $\theta$, which appears to be the case in our numerical computations.

							\end{remark}
						}

						\begin{remark}
							\normalfont 
							We discuss two implications of Theorem~\ref{thm:NR}\textbf{(iii)}.

							The first implication is that we expect a network dictionary that tends to be effective at approximating the mesoscale patches of a network 
							to also be effective at approximating the entire network. 
							Suppose that $\mathtt{denoising} = \texttt{F}$ in Theorem~\ref{thm:NR}\textbf{(iii)}. 
							Recall that the columns of $W$ encode $r$ latent motifs $\mathcal{L}_{1},\ldots,\mathcal{L}_{r}\in \R_{\ge 0}^{k\times k}$ (see {Appendix}~\ref{subsection:NDL_formulation}). According to {\eqref{eq:NR_Jaccard_bound}, we have a perfect reconstruction} $\G = \hat{\G}_{\infty}$ if the right-hand side of {\eqref{eq:NR_Jaccard_bound} is} $0$. This is the case if $\sup_{\x:F\rightarrow \G} \ell(\mathtt{vec}(A_{\x}),W)=0$, {which implies} that $W$ can perfectly approximate all mesoscale patches $A_{\x}$ of $\G$. However, the right-hand side of \eqref{eq:NR_Jaccard_bound} can still be small if the worst-case approximation error $\sup_{\x:F\rightarrow \G} \ell(\mathtt{vec}(A_{\x}),W)$ is large but the expected approximation error $\E_{\x\sim \pi}\left[ \ell( \mathtt{vec}(A_{\x}), W)\right]$ is small (i.e., when $W$ is effective at approximating most of the mesoscale patches).

							How can we find a network dictionary $W$ that minimizes the right-hand side of \eqref{eq:NR_Jaccard_bound}? Although it is difficult to find a globally optimal network dictionary $W$ that minimizes the non-convex objective function on the right-hand side of \eqref{eq:NR_Jaccard_bound}, Theorems~\ref{thm:NDL} and~\ref{thm:NDL2} guarantee that our NDL algorithm (see Algorithm~\ref{alg:NDL}) always finds a locally optimal network dictionary. From these theorems, we know that the NDL algorithm with $N = 1$ computes a network dictionary $W$ that is approximately a local optimum of the expected loss function 
							\begin{align}\label{eq:NDR_bound_remark_4}
								f(W) = \E_{\x\sim \pi}\left[ \ell( \mathtt{vec}(A_{\x}), W)\right]\,,
							\end{align}
							where {$\pi=\hat{\pi}_{F\hookrightarrow \G}$} if $\mathtt{MCMC}=\mathtt{PivotApprox}$ and {$\pi=\pi_{F\hookrightarrow \G}$} {if $\mathtt{MCMC}\in \{ \mathtt{Glauber}, \mathtt{pivot}  \}$.}
							The function $f$ in \eqref{eq:NDR_bound_remark_4} is {similar to the one in the upper bound in} \eqref{eq:NR_Jaccard_bound}. 
							In our experiments, we find
							that our NDL algorithm produces network dictionaries that are efficient at minimizing the reconstruction error. {See Figure~\ref{fig:network_recons} and the left-hand sides of \eqref{eq:NR_Jaccard_bound} and \eqref{eq:NR_Jaccard_bound_binary}.}

							The second implication is that reconstructing a corrupted network using a network dictionary that is trained from an uncorrupted network will yield a network that is similar to the uncorrupted network. Consider an uncorrupted network $\G'=(V,B)$ and a corrupted network $\G=(V,A)$. Additionally, suppose that we have trained the network dictionary $W$ for the uncorrupted network $\G'$, but that we use it to reconstruct the corrupted network $\G$. Even if $\hat{A}_{\x;W}$ is a nonnegative linear approximation of the $k \times k$ matrix $A_{\x}$ of a mesoscale patch of the corrupted network $\G$, it may be close to the corresponding mesoscale patch $B_{\x}$ of the uncorrupted network $\G'$ because we use the network dictionary $W$ that we learned from the uncorrupted network $\G'$.  Theorem~\ref{thm:NR}\textbf{(iii)} guarantees that the network $\hat{\G}_{\infty}$ that we reconstruct for the corrupted network $\G$ using the uncorrupted-network dictionary $W$ is close to the uncorrupted network $\G'$.    
						\end{remark}

						\begin{remark}
							\normalfont

							The update step (see line 17) for 
							the reconstruction in Algorithm~\ref{alg:network_reconstruction} indicates that we loop over all node pairs $(a,b)$ in a $k$-chain motif  
							and that we update the weight of the edge {$\{ \x_{t}(a), \x_{t}(b) \}$} in the reconstructed network using the homomorphism $\x:F\rightarrow \G$.

							There may be multiple node pairs $(a,b)$ in $F$ that contribute to the edge {$\{x,y\}$} in the reconstructed network because $\x_{t}(a) = x$ and $\x_{t}(b) = y$ can occur for multiple choices of $(a,b)$. 
							The output of this update step does not depend on the ordering of $a,b \in \{1,\ldots,k\}$, as one can see from the expressions in \eqref{eq:hatA_t_formula}.

							One can also consider the following alternative update step for the reconstruction. 
							In this alternative, we first choose two nodes, $x$ and $y$, of the reconstructed network in the image $\{\x_{t}(j)\,|\, j \in \{1,\ldots,k\}\}$ of the homomorphism $\x_{t}$ and average over all pairs $(a,b)\in \{1, \ldots, k\}^{2}$ such that $(x,y)$ is visited by $(a,b)$ through $\x_{t}$. We then update the weight of $(x,y)$ in the reconstructed network with this mean contribution from $\x_{t}$. Specifically, for each $a,b \in \{1, \ldots, k\}$, let $\mathbb{1}\big(\text{$(x,y)\overset{\x_{t}}{\hookleftarrow} (a, b)$}\big)$ denote the indicator in \eqref{eq:def_indicator_visit} and let $N_{xy}(\x_{t})\ge 0$ denote the number of visits of $\x_{t}$ to $(x,y)$ (see \eqref{eq:def_N_pq}). 
							We can then replace lines 17--19 in Algorithm~\ref{alg:network_reconstruction} with the following lines: 
							\begin{description}[leftmargin=0.7cm, topsep=0.2pt, itemsep=0.1cm]
								\item{\textit{Alternative update for reconstruction}:}
								\vspace{0.1cm}
								\item{} \quad \textbf{For} $x,y\in V$ such that $N_{xy}(\x_{t})>0$:
								\item{} \qquad $\displaystyle \widetilde{A}_{\x_{t};W}(x,y)\leftarrow \frac{\sum_{1\le a,b\le k} \hat{A}_{\mathbf{x}_{t};W}(a,b) \mathbb{1}\big(\text{$(x,y)\overset{\x}\hookleftarrow (a, b)$}\big) }{\sum_{1\le a,b\le k} \mathbb{1}\big(\text{$(x,y)\overset{\x}{\hookleftarrow} (a, b)$}\big)}$\,, \quad $j\leftarrow A_{\textup{count}}(x,y)+1$  
								\vspace{0.1cm}
								\item{} \qquad $A_{\textup{recons}}(x,y)\leftarrow (1 - j^{-1})A_{\textup{recons}}(x,y) + j^{-1} \widetilde{A}_{\x_{t};W}(x,y)$\,.
							\end{description}
							
							\medskip

							For the alternative NDR algorithm that we just described, we can establish a convergence result that is similar to Theorem~\ref{thm:NR} using a similar argument as the one in our proof of Theorem~\ref{thm:NR}. 
							Specifically, \textbf{(i)} holds for the alternative NDR algorithm, so there exists a limiting reconstructed network. In 
							\textbf{(ii)}, the formula for the limiting reconstructed network is now
							\begin{align}
								\hat{A}_{\infty}(x,y) = \sum_{\y\in \Omega_{xy}} \widetilde{A}_{\y;W}(x,y)\, \P_{\x\sim \pi}\left( \x=\y \,\big|\, \x\in \Omega_{xy} \right) \quad \text{for all} \quad \text{$x,y \in V$}\,,
							\end{align}
							where $\Omega_{xy}$ is the set of all homomorphisms that visit $(x,y)$ (see \eqref{eq:def_Omega_pq}). In particular, if $\G$ is an undirected and unweighted graph, then 
							\begin{align}
								\hat{A}_{\infty}(x,y) = \frac{1}{|\Omega_{xy}|}\sum_{\y\in \Omega_{xy}} \widetilde{A}_{\y;W}(x,y)  
								\quad \text{for all} \quad \text{$x,y \in V$}\,.
							\end{align}
							In the proof of \textbf{(iii)}, 
							the same error bounds hold with $\E_{\x\sim \pi}[N_{xy}(\x)]$ replaced by $\P_{\x\sim \pi}\left( \x\in  \Omega_{xy} \right)$. We omit the details of the proofs of the above statements for this alternative NDR algorithm. 
						\end{remark}

						We now discuss the convergence results of Algorithm~\ref{alg:network_reconstruction} for a bipartite network $\G$. Recall our notation and our discussion of bipartite networks above Theorem~\ref{thm:NDL2}. Additionally, for  bipartite networks, recall that there exist disjoint subsets $\Omega_{1}$ and $\Omega_{2}$ of the set $\Omega$ of all homomorphisms $F\rightarrow \G$ such that (1) $\Omega = \Omega_{1} \cup \Omega_{2}$ and (2) the Markov chain $(\x)_{t\ge 0}$ restricted to each $\Omega_{i}$ (with $i \in \{1,2\}$) is irreducible but is not irreducible on the set $\Omega$.

						\begin{theorem}[Convergence of our NDR Algorithm (see Algorithm~\ref{alg:network_reconstruction}) for Bipartite Networks]\label{thm:NR2}
							
							Let $F=([k],A_{F})$ be {a} $k$-chain motif, and let $\G=(V,A)$ be a network that satisfies assumption
							(a') in Theorem~\ref{thm:NDL2}. Let $\hat{\G}_{t}=(V, \hat{A}_{t})$ denote the network that we reconstruct using Algorithm~\ref{alg:network_reconstruction} at iteration $t$ {with} a fixed network dictionary $W\in \mathbb{R}^{k^{2}\times r}$. Fix $i \in \{1,2\}$ and an initial (not necessarily injective) homomorphism $\x_{0}\in \Omega_{i}$. Let $\pi=\hat{\pi}_{F\rightarrow \G}$ if \,$\mathtt{MCMC} = \mathtt{PivotApprox}$ and $\pi=\pi_{F\rightarrow \G}$ for $\mathtt{MCMC}\in \{ \mathtt{Glauber}, \mathtt{pivot}  \}$. The following properties hold:
							\begin{description}[leftmargin=0.7cm, topsep=0.2pt, itemsep=0.1cm]
								\item[(i)] (Convergence of the network reconstruction)
								The network $\hat{\G}_{t}$ converges almost surely to some limiting network $\hat{\G}_{\infty}=(V, \hat{A}_{\infty})$ in the sense that 
								\begin{align}
									\lim_{t\rightarrow \infty} \hat{A}_{t}(x,y) = \hat{A}_{\infty}(x,y) \quad \text{almost surely for all $x,y \in V$}\,. 
								\end{align}
								
								\item[(ii)--{(iii)}] The same statements as in statements \textup{\textbf{(ii)}--{\textbf{(iii)}}} of Theorem~\ref{thm:NR} hold with the expectation $\E_{\x\sim \pi}$ replaced by the conditional expectation $\E_{\x\sim \pi}[ \cdot \,|\, \x\in \Omega_{i}]$.
								
								\vspace{0.1cm}
								\item[{(iv)}] The results in \textup{{\textbf{(i)}}--{\textbf{(iii)}}} do not depend on $i \in \{1,2\}$ if $k$ is even. 
								
							\end{description}
						\end{theorem}
						
						\begin{proof}
							The proofs of statements \textbf{(i)}--{\textbf{(iii)}} are identical to those for Theorem~\ref{thm:NR}. Statement {\textbf{(iv)}} follows from a similar argument as in the proof of Theorem~\ref{thm:NDL2}\textbf{(ii)} by constructing coupled Markov chains $(\x_{t})_{t\ge 0}$ and $(\x_{t}')_{t\ge 0}$ such that $\x_{t}'=\overline{\x_{t}}$ for all $t\ge 0$. 
						\end{proof}

						\begin{remark}
							\normalfont
							In Theorem~\ref{thm:NR2}\textbf{(ii)}, let $\hat{G}_{\infty}^{(i)}=(V, \hat{A}_{\infty}^{(i)})$, with $i \in \{1,2\}$, denote the limiting reconstructed network for $\G$ conditional {on} {initializing the Markov chain} in the subset $\Omega_{i}${.} When $k$ is even, Theorem~\ref{thm:NR2}\textbf{(iv)} implies that $\hat{\G}_{\infty}^{(1)}=\hat{\G}_{\infty}^{(2)}$. When $k$ is odd, we run the NDR algorithm (see Algorithm~\ref{alg:network_reconstruction}) twice with the Markov chain initialized {once in $\Omega_{1}$ and {once} in $\Omega_{2}$}. We then define the network $\hat{\G}_{\infty} = (V, (\hat{A}_{\infty}^{(1)}+\hat{A}_{\infty}^{(2)})/2)$ whose weight matrix is the mean of those of the two limiting reconstructed networks $\hat{\G}^{(i)}_{\infty}$ for $i \in \{1,2\}$. We obtain a similar error bound as in Theorem~\ref{thm:NR2}\textbf{(iii)} for the mean limiting reconstructed network $\hat{\G}_{\infty}$. In practice, one can obtain a sequence of reconstructed networks that converges to $\hat{\G}_{\infty}$ by reinitializing the Markov chain every $\tau$ iterations of the reconstruction procedure for any fixed $\tau$. 
						\end{remark}
						
						

						\section{Auxiliary Algorithms}
						\label{subsection:aux_alg}

						We now present auxiliary algorithms for solving various subproblems of Algorithms~\ref{alg:NDL} and~\ref{alg:network_reconstruction}. Let $\Pi_{S}$ denote the projection operator onto a subset $S$ of a 
						space. For each matrix $A$, let $[A]_{\bullet i}$ (respectively, $[A]_{i\bullet}$) denote the $i^{\textup{th}}$ column (respectively, $i^{\textup{th}}$ row) of $A$.

						\begin{algorithm}
							\renewcommand{\thealgorithm}{A1}
							\caption{\!\!. Coding}
							\label{algorithm:spaser_coding}
							\begin{algorithmic}[1]
								\State \textbf{Input:} Data matrix $X \in  \mathbb{R}^{d\times d'}$, dictionary matrix $W\in \mathbb{R}^{d\times r}$
								\State \textbf{Parameters:} $T\in \mathbb{N}$ (the number of iterations), $\lambda>0$ (the coefficient of an $L_{1}$-regularizer),
								\Statex \hspace{2.3cm} $\mathcal{C}^{\textup{code}}\subseteq \mathbb{R}^{r\times d'}$ (a convex constraint set of codes)
								
								\State \textbf{For $t=1,\ldots,T$:}
								\State \qquad \textbf{Do:} 
								\begin{align}\label{eq:algorithm_H}	
									H \leftarrow \Pi_{\mathcal{C}^{\textup{code}}}\left( H - \frac{1}{\tr(W^{T}W)}(W^{T}W H - W^{T}X + \lambda \mathbf{1}_{d\times d'})  \right)\,,
								\end{align}			
								\qquad where ${\mathbf{1}_{d\times d'}}\subseteq \mathbb{R}^{d\times d'}$ is the matrix with all $1$ entries

								\State \textbf{Output:} $H\in \mathcal{C}^{\textup{code}}\subseteq \R^{r\times d'}$
							\end{algorithmic}
						\end{algorithm}

						\begin{algorithm}
							\renewcommand{\thealgorithm}{A2}
							\caption{\!\!. Dictionary-Matrix Update}
							\label{algorithm:dictionary_update}
							
							\begin{algorithmic}[1]
								\State {\textbf{Input:} Previous dictionary matrix $W_{t-1} \in \R^{k^{2}\times r}$, previous aggregate matrices $(P_{t}, Q_{t})\in \R^{r\times r}\times \R^{r\times N}$}
								\State {\textbf{Parameters:} $\mathcal{C}^{\textup{dict}}\subseteq \R^{k^{2}\times r}$ (compactness and convexity constraint for dictionary matrices)},
								\Statex \hspace{2.3cm} {$T\in \mathbb{N}$ (the number of iterations)}
								
								\State \textbf{For $t=1,\ldots,T$:} 
								\State \qquad $W\leftarrow W_{t-1}$
								\State \qquad \textbf{For $j=1,2,\ldots,N$:}
								\begin{align}\label{eq:dictionary_column_update}
									W(:,j) \leftarrow \Pi_{\mathcal{C}^{\textup{dict}}} \left( W(:,j)  - \frac{1}{A_{t}(j,j)+1} (W P_{t}(:,j) - Q_{t}^{T}(:,j) ) \right)
								\end{align}
								\State \textbf{Output:} $W_{t}=W\in \mathcal{C}^{\textup{dict}}\subseteq \R^{k^{2}\times r}_{\ge 0}$
							\end{algorithmic}
						\end{algorithm}

						\begin{algorithm}
							\renewcommand{\thealgorithm}{A3}
							\begin{algorithmic}[1]
								\caption{\!\!. Rejection Sampling of Homomorphisms}\label{alg:rejection_motif}
								\State \textbf{Input:} Network $\G=(V,A)$, a $k$-chain motif $F=([k],A_{F})$
								
								{\qquad ($\triangleright$ This algorithm works for all motifs, but we specialize it to $k$-chain motifs.)}
								
								\vspace{0.1cm}
								\State \textbf{Requirement:} There exists at least one homomorphism $F\rightarrow \G$
								
								\vspace{0.1cm}
								\State \textbf{Repeat:} Sample a sequence $\x = [\x(1),\x(2),\ldots,\x(k)]\in V^{[k]}$ 
								such that $\x(1),\ldots,\x(k)$ are independent and they each are sampled uniformly from $V$

								\State \qquad \textbf{If}  $\prod_{i,j \in \{1, \ldots, k\}} A(\mathbf{x}(i), \mathbf{x}(j))^{A_{F}(i,j)} > 0$
								\State \qquad \qquad \textbf{Return} $\x:F\rightarrow \G$ and \textbf{Terminate}
								\State \textbf{Output:} Homomorphism $\x:F\rightarrow \G$
							\end{algorithmic}
						\end{algorithm}

						\begin{algorithm}
							\renewcommand{\thealgorithm}{A4}
							\begin{algorithmic}[1]
								\caption{\!\!. Vectorization}\label{alg:vectorize}
								\State \textbf{Input:} Matrix $X \in \R^{k_{1} \times k_{2} }$
								\vspace{0.1cm}
								
								\State \textbf{Output:} Matrix $\texttt{vec}(X) := Y\in \R^{k_{1}k_{2}\times 1}$\,, where 
								\begin{align}
									Y(k_{2}(j-1)+i,1) = X(i,j)\quad  \text{for all $i \in \{1, \ldots, k_{1}\}$ and $j \in \{1, \ldots, k_{2}\}$}
								\end{align}
								
							\end{algorithmic}
						\end{algorithm}

						\begin{algorithm}
							\renewcommand{\thealgorithm}{A5}
							\begin{algorithmic}[1]
								\caption{\!\!. Reshaping}\label{alg:reshape}
								\State \textbf{Input:} Matrix $X \in \R^{k_{1}k_{2} \times 1}$, a pair $(k_{1},k_{2})$ of integers
								\vspace{0.1cm}
								
								\State \textbf{Output:} Matrix $\texttt{reshape}(X) := Y\in \R^{k_{1}\times k_{2}}$\,, where 
								\begin{align}
									Y(i,j) = X(k_{2}(j-1)+i,1)\quad  \text{for all $i \in \{1, \ldots, k_{1}\}$   and   $j \in \{1, \ldots, k_{2}\}$}
								\end{align}
								
							\end{algorithmic}
						\end{algorithm}

						
						\newpage

						\hspace{3cm}


						\section{Additional Figures}
						\label{subsection:additional_figures}

						In Figures \ref{fig:Figure4__PRC_REC}, \ref{fig:Figure_PR_curve}, and \ref{fig:Figure_PR_curve_flipped}, we show additional binary-classification measures
						for the network-denoising experiments in Figure \ref{fig:Figure4}. In Figures \ref{fig:all_dictionaries_1}--\ref{fig:all_dictionaries_rank_4}, we show latent motifs of 
						all networks we consider in this work at various choices of parameters.

						\begin{figure*}[h]
							\centering
							\includegraphics[width=1\linewidth]{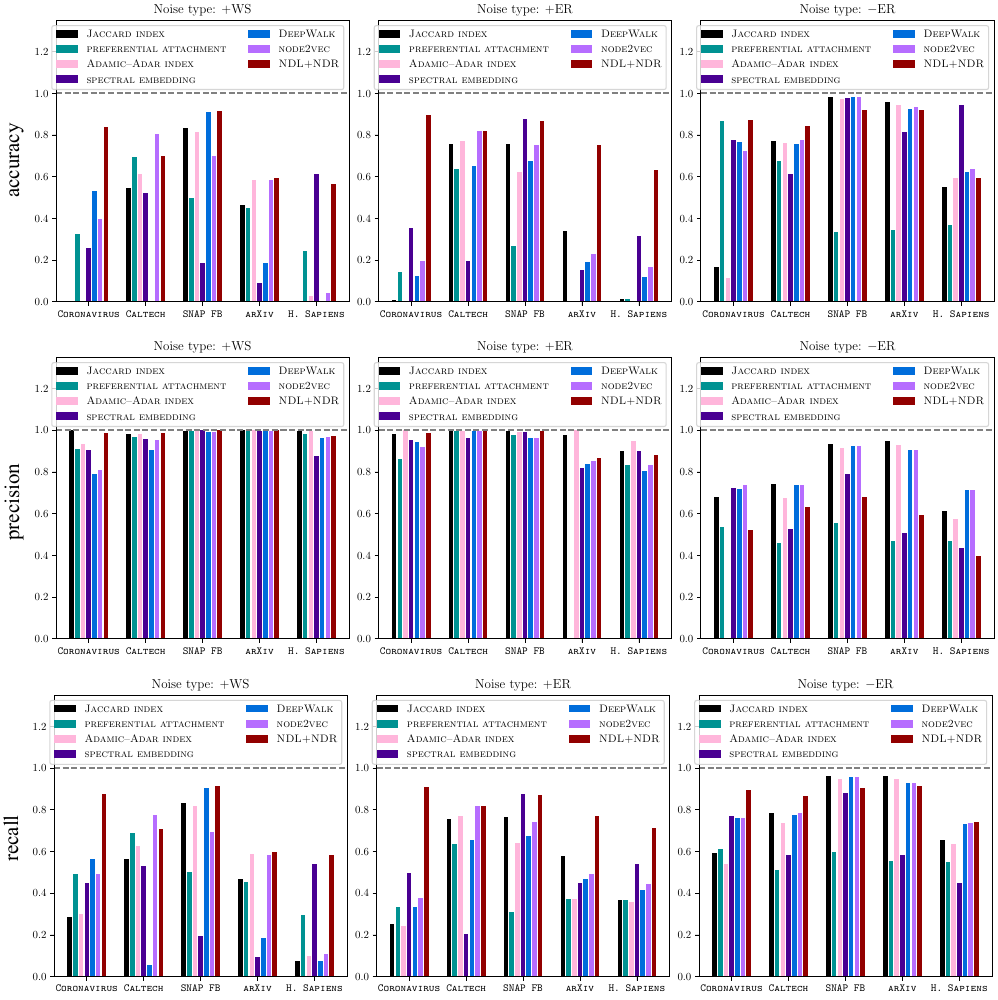}
							\caption{The accuracy, precision, and recall scores for the network-denoising experiments in Figure~\ref{fig:Figure4} of the main manuscript. 
									{For convenience, we recall the definitions of these quantities, which one can use for binary classification. 
										One summarizes the result of a binary classification using combinations of four quantities: 
										$\text{TP}$ (true positives), which is the number of positives that are classified as positive; $\text{TN}$ (true negatives), which is the number of {negatives} that are classified as negative; $\text{FP}$ (false positives), which is the number of positives that are classified as negative; and $\text{FN}$ (false negatives), which is the number of {negatives} that are classified as positive. The total number of examples is the sum of these four quantities. Accuracy is $\frac{\text{TP} + \text{TN}}{\text{TP} + \text{TN} + \text{FP} + \text{FN}}$, precision is $\frac{\text{TP}}{\text{TP} + \text{FP}}$, and recall is $\frac{\text{TP}}{\text{TP} + \text{FN}}$.}  
							}
							\label{fig:Figure4__PRC_REC}
						\end{figure*}

						\newpage

						\begin{figure*}[h]
							\centering
							\includegraphics[width=1\linewidth]{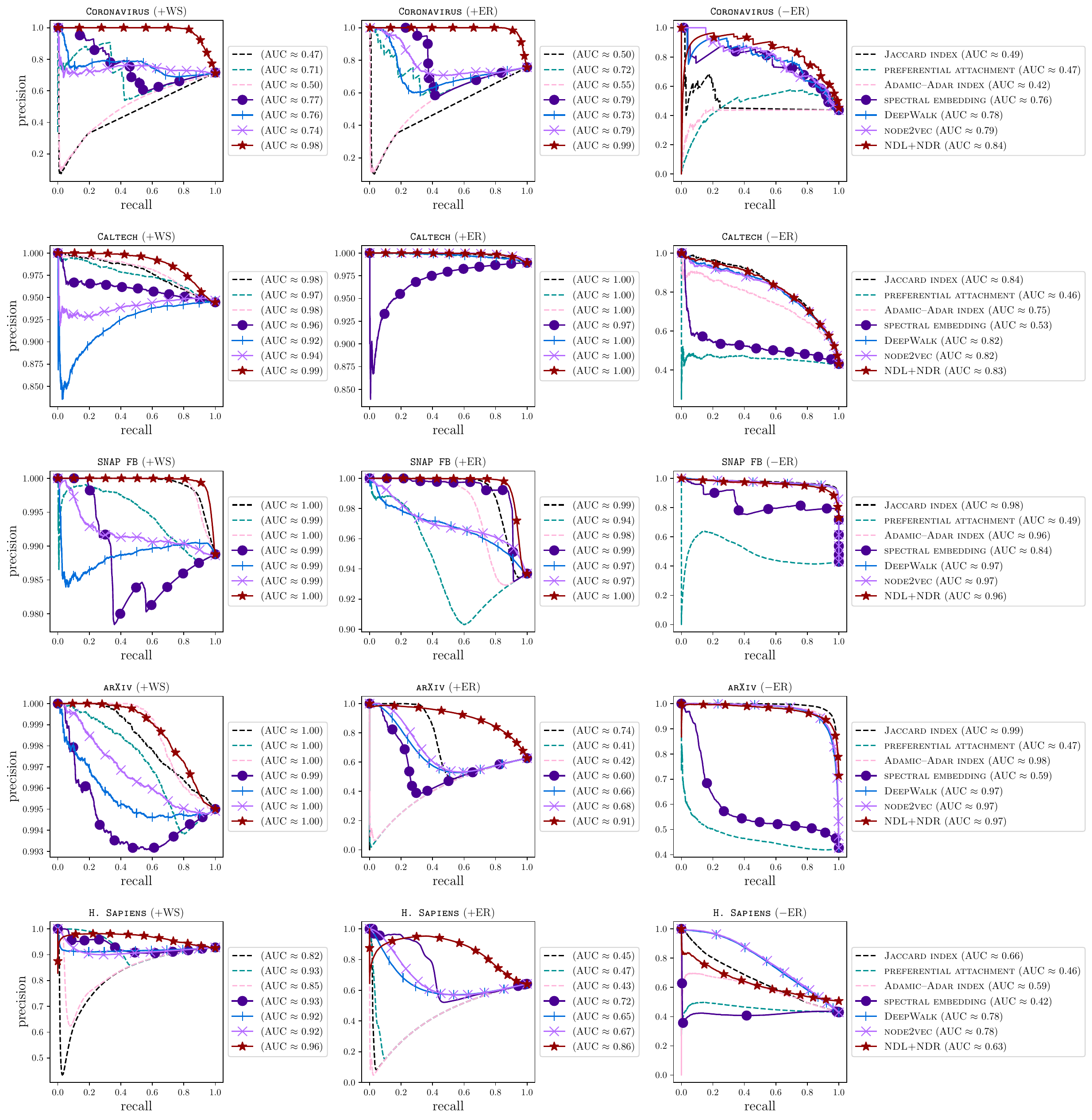}
							\caption{The  curves of precision versus recall for the network-denoising experiments in Figure~\ref{fig:Figure4} of the main manuscript. See the caption of Figure \ref{fig:Figure4__PRC_REC} for the definition of precision and recall. 
							}
							\label{fig:Figure_PR_curve}
						\end{figure*}

						\newpage
						
							\begin{figure*}[h]
							\centering
							\includegraphics[width=1\linewidth]{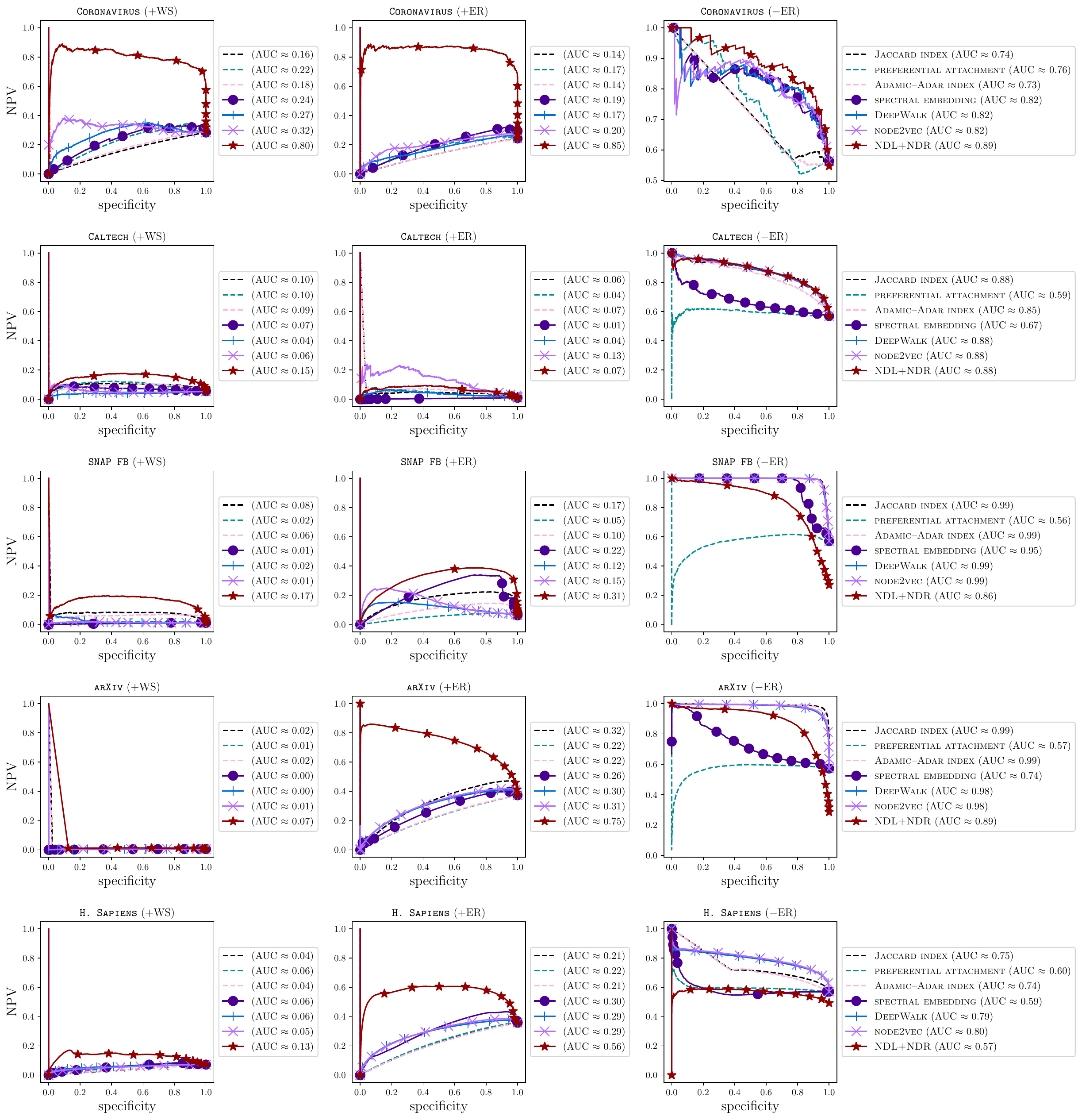}
							\caption{The  curves of negative predictive value (NPV) versus specificity
							  for the network-denoising experiments in Figure~\ref{fig:Figure4} of the main manuscript. The NPV is $\frac{\textup{TN}}{\textup{TN}+\textup{FN}}$ and specificity is {$\frac{\textup{TN}}{\textup{TN} + \textup{FP}}$}. See \cite{parikh2008understanding} for a discussion of NPV and specificity. 
							}
							\label{fig:Figure_PR_curve_flipped}
						\end{figure*}

						\newpage

						\begin{figure*}[h]
							\centering
							\includegraphics[width=1 \linewidth]{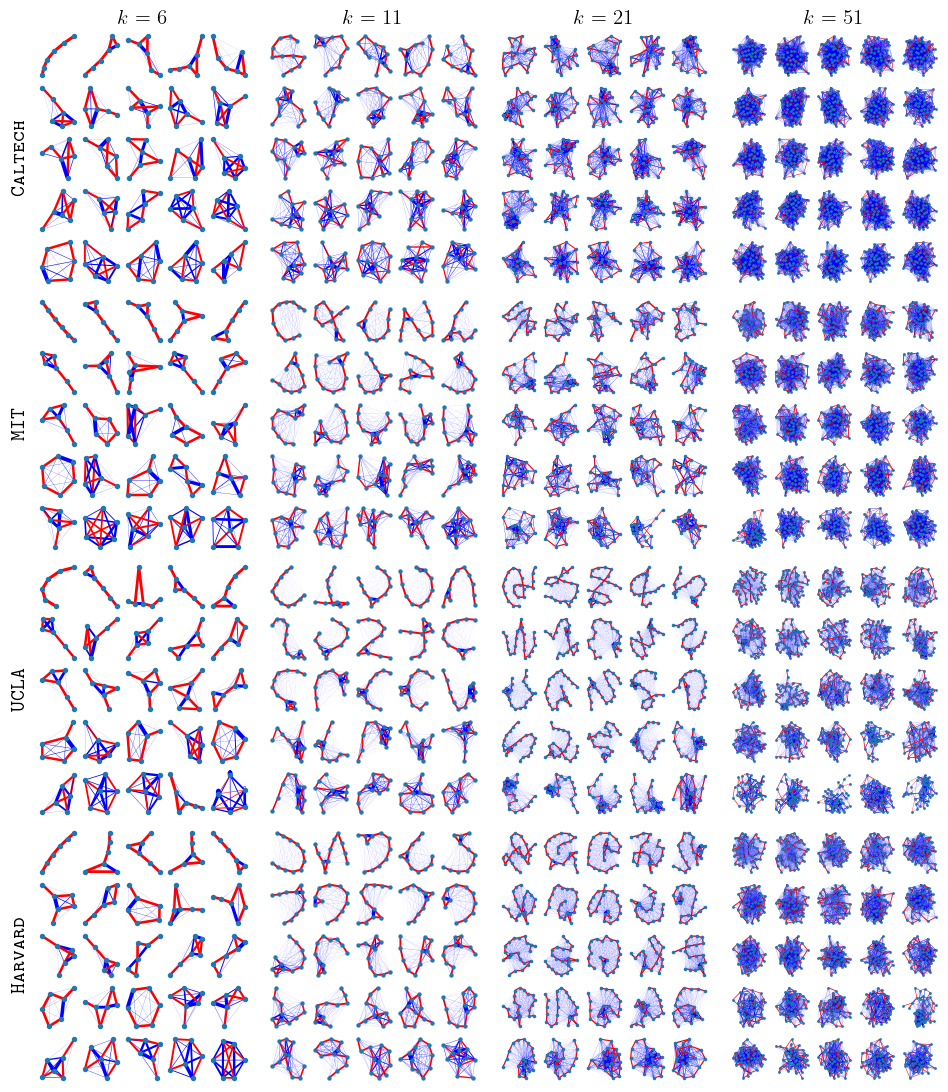}
							\caption{The $r = 25$ latent motifs at scales $k = 6$, $k = 11$, $k = 21$, and $k = 51$ that we learn from the networks \dataset{Caltech}, \dataset{MIT}, \dataset{UCLA}, and \dataset{Harvard}. See {Appendix}~\ref{section:experimental_details} for the details of these experiments. 
							}
							\label{fig:all_dictionaries_1}
						\end{figure*}

						\newpage

						\begin{figure*}[h]
							\centering
							\includegraphics[width=1 \linewidth]{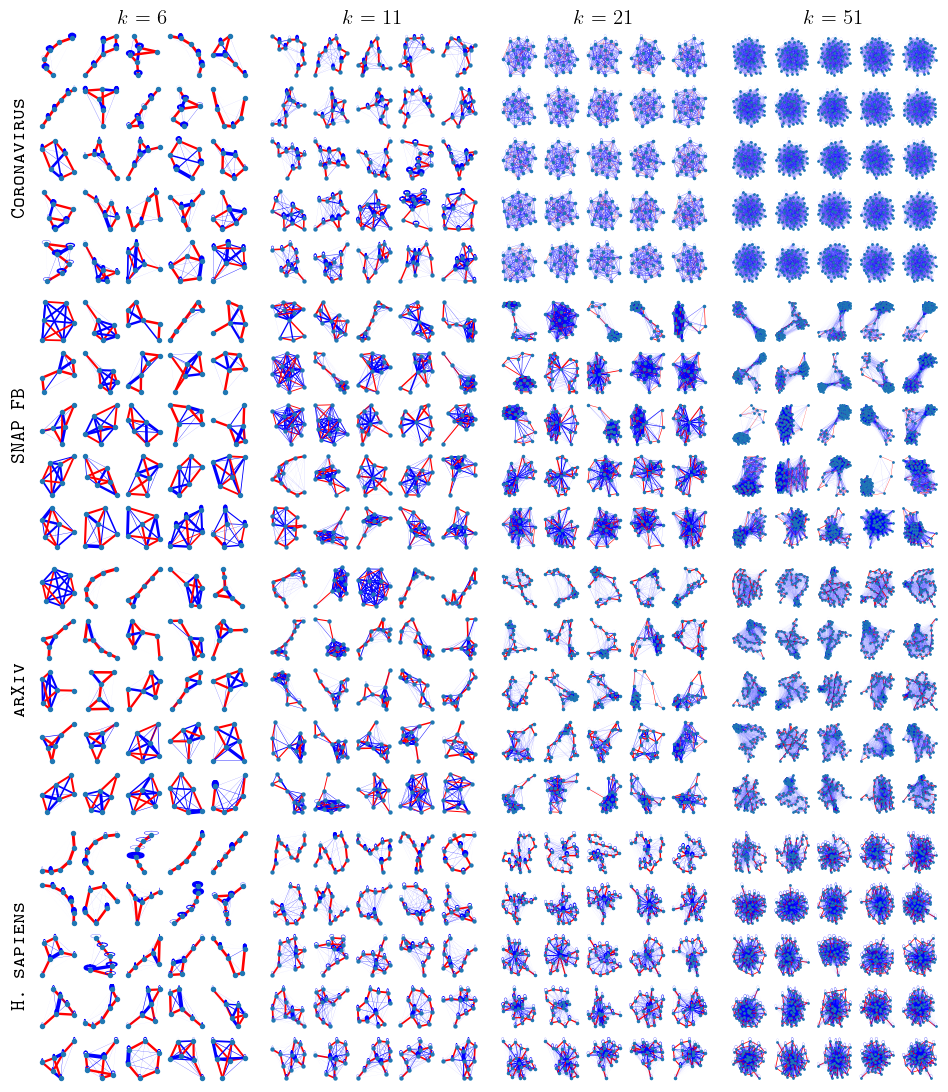}
							\caption{The $r = 25$ latent motifs at scales $k = 6$, $k = 11$, $k = 21$, and $k = 51$ that we learn from the networks \dataset{Coronavirus PPI}, \dataset{SNAP Facebook}, \dataset{arXiv ASTRO-PH}, and \dataset{Homo sapiens PPI}. See {Appendix}~\ref{section:experimental_details} for the details of these experiments.} 
							\label{fig:all_dictionaries_2}
						\end{figure*}

						\newpage
						
						\begin{figure*}[h]
							\centering
							\includegraphics[width=1 \linewidth]{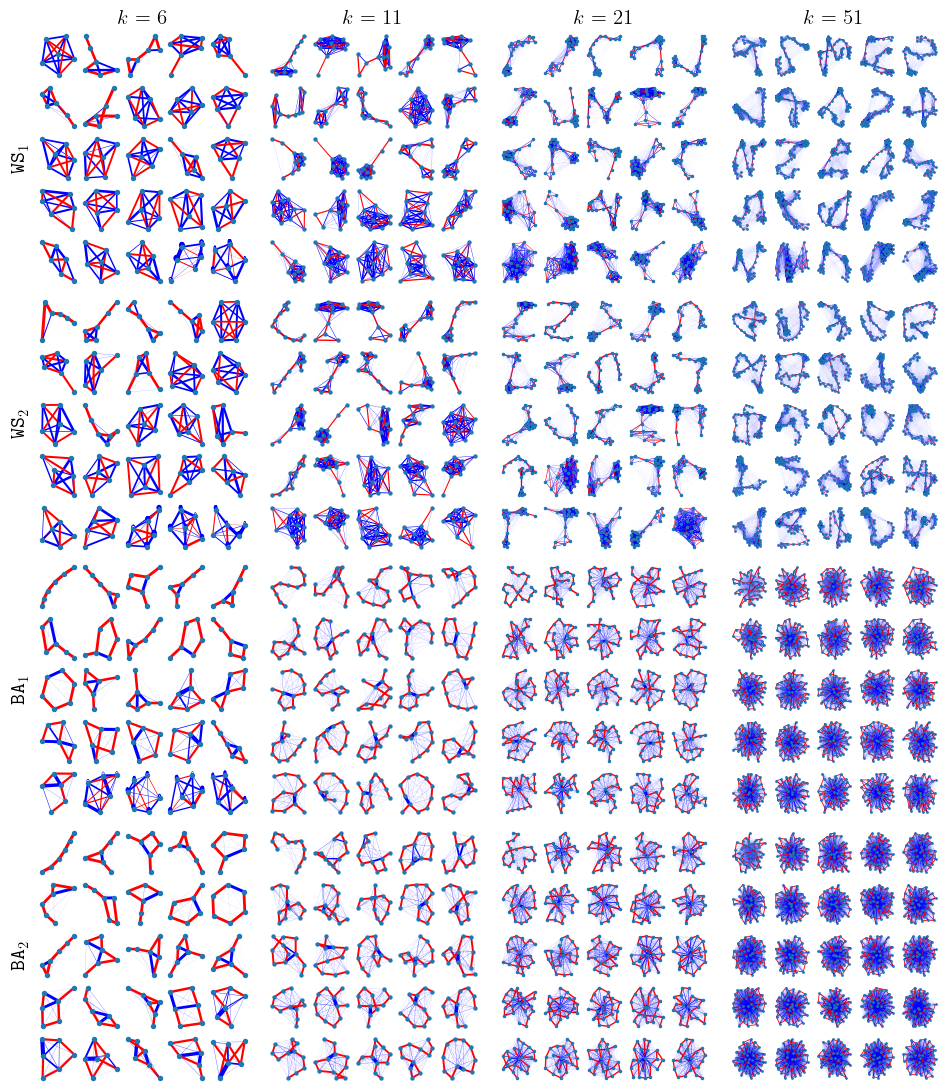}
							\caption{The $r = 25$ latent motifs at scales $k = 6$, $k = 11$, $k = 21$, and $k = 51$ that we learn from the networks \dataset{WS$_1$}, \dataset{WS$_2$}, \dataset{BA$_1$}, and \dataset{BA$_2$}. See {Appendix}~\ref{section:experimental_details} for the details of these experiments.} 
							\label{fig:all_dictionaries_3}
						\end{figure*}

						\newpage
						
						\begin{figure*}[h]
							\centering
							
							\includegraphics[width=1 \linewidth]{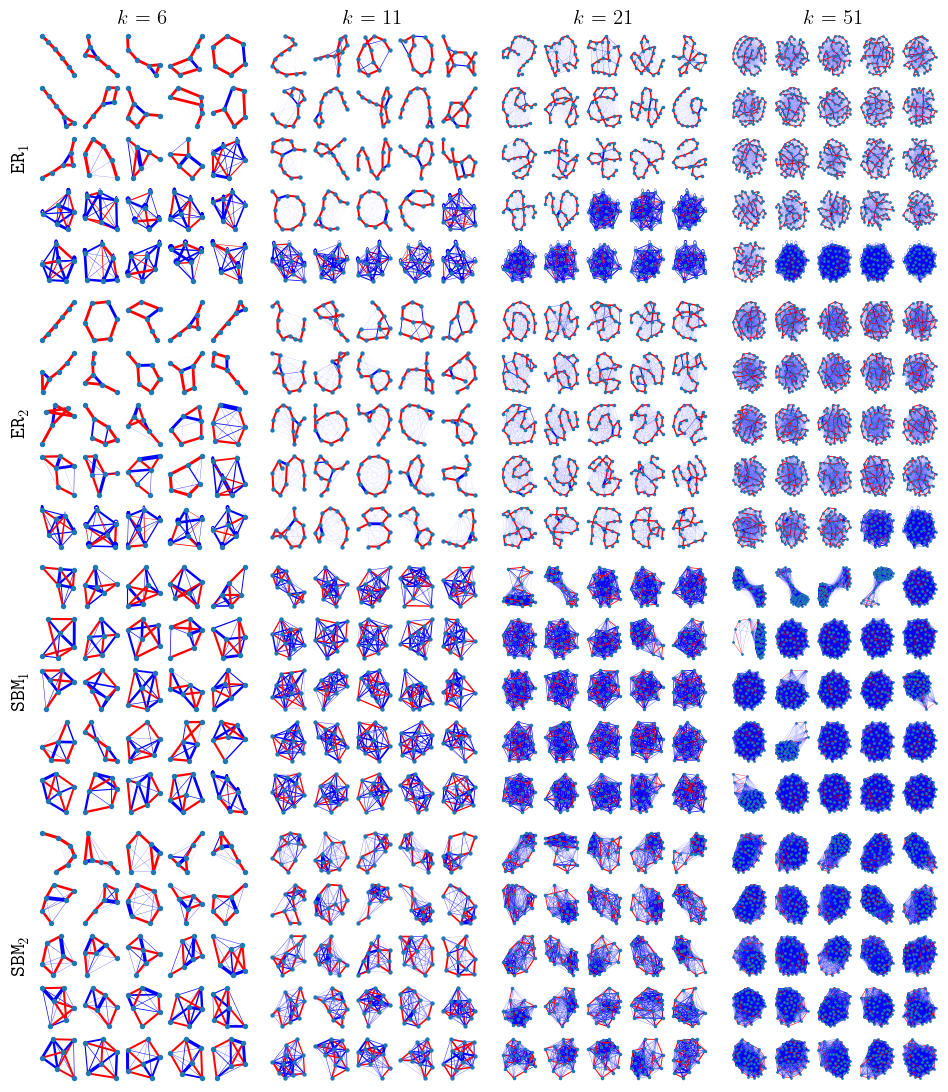}
							\caption{The $r = 25$ latent motifs at scales $k = 6$, $k = 11$, $k = 21$, and $k = 51$ that we learn from the networks \dataset{ER$_1$}, \dataset{ER$_2$}, \dataset{SBM$_1$}, and \dataset{SBM$_2$}. See {Appendix}~\ref{section:experimental_details} for the details of these experiments.}
							\label{fig:all_dictionaries_4}
						\end{figure*}

						\newpage
						
						\begin{figure*}[h]
							\centering
							\includegraphics[width=1 \linewidth]{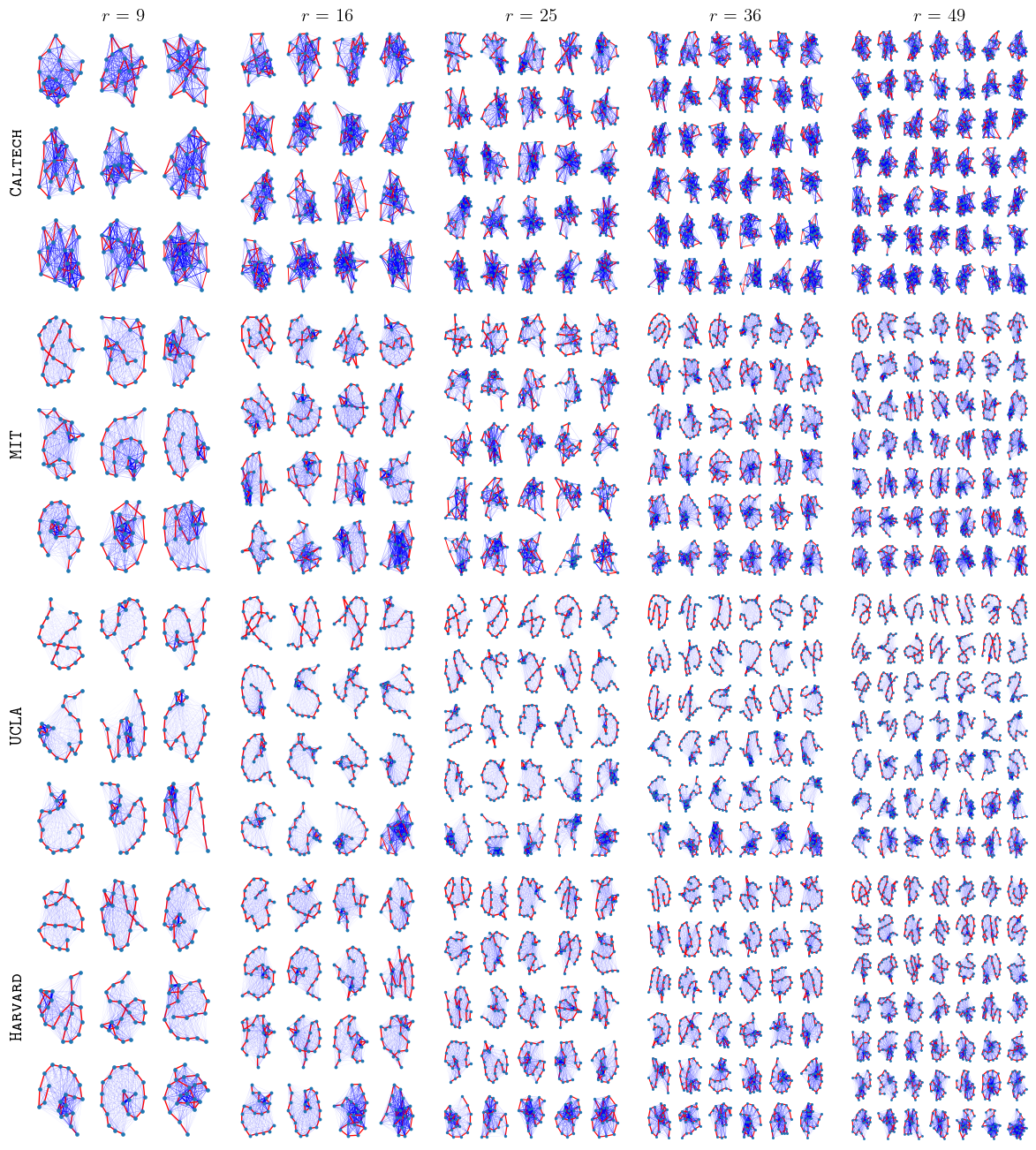}
							\caption{The $r \in \{9,16,25,36,49\}$ latent motifs at scale $k = 21$ that we learn from the networks \dataset{Caltech}, \dataset{MIT}, and \dataset{UCLA}. The $r = 25$ column is identical to the $k=21$ column in Figure~\ref{fig:all_dictionaries_1}. See {Appendix}~\ref{section:experimental_details} for the details of these experiments.
							}
							
							\label{fig:all_dictionaries_rank_1}
						\end{figure*}

						\newpage
						
						\begin{figure*}[h]
							\centering
							\includegraphics[width=1 \linewidth]{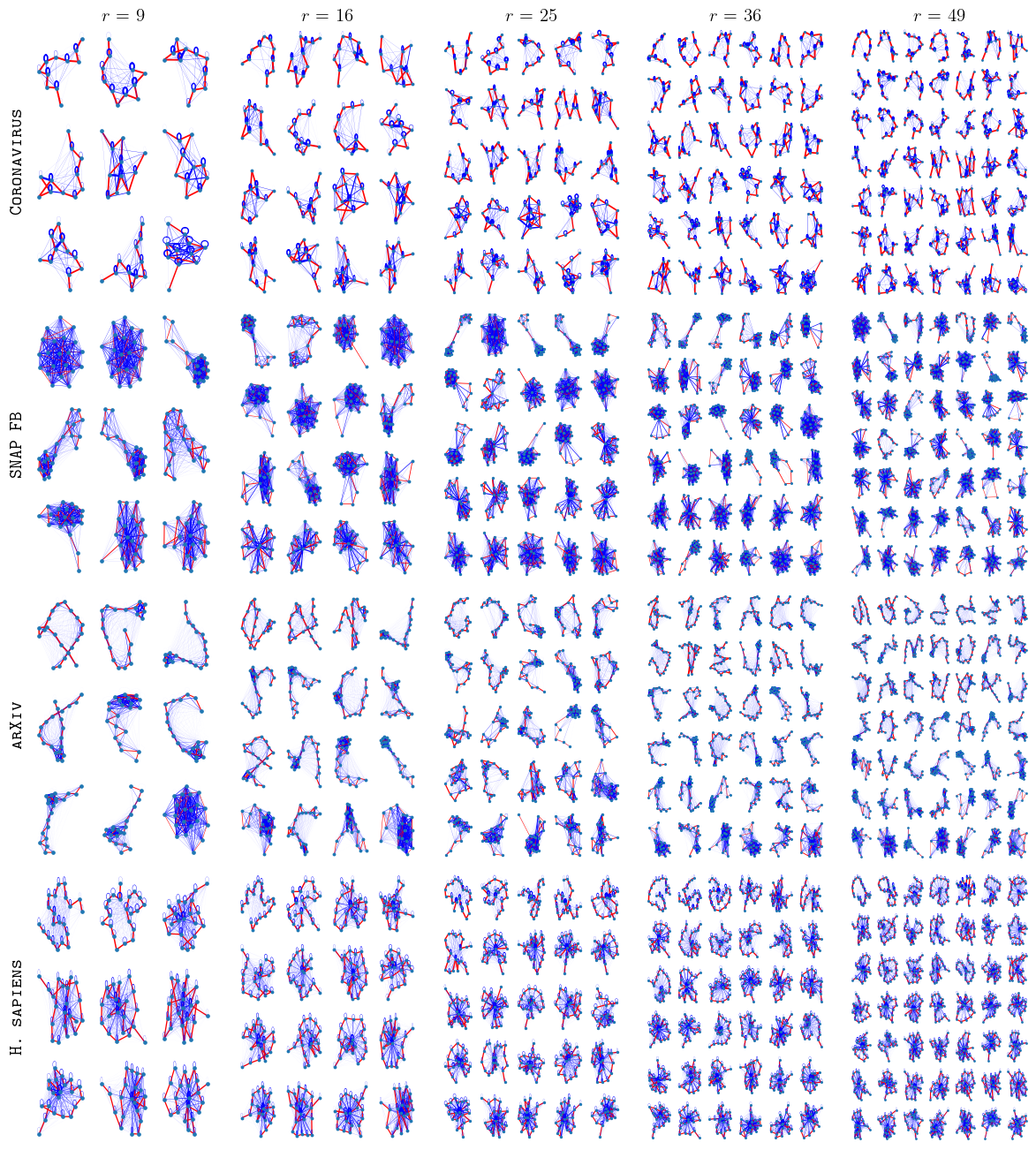}
							\caption{The $r \in \{9,16,25,36,49\}$ latent motifs 
							that we learn from the networks \dataset{SNAP Facebook}, \dataset{arXiv ASTRO-PH}, and \dataset{Homo sapiens PPI} at scale $k = 21$ and from \dataset{Coronavirus PPI} at scale $k = 11$. The $r = 25$ for column is identical to the $k = 21$ column in Figure~\ref{fig:all_dictionaries_2}, except for \dataset{Coronavirus PPI}.
									See {Appendix}~\ref{section:experimental_details} for the details of these experiments. 
							}
							
							\label{fig:all_dictionaries_rank_2}
						\end{figure*}
						
						\newpage
						
						\begin{figure*}[h]
							\centering
							\includegraphics[width=1 \linewidth]{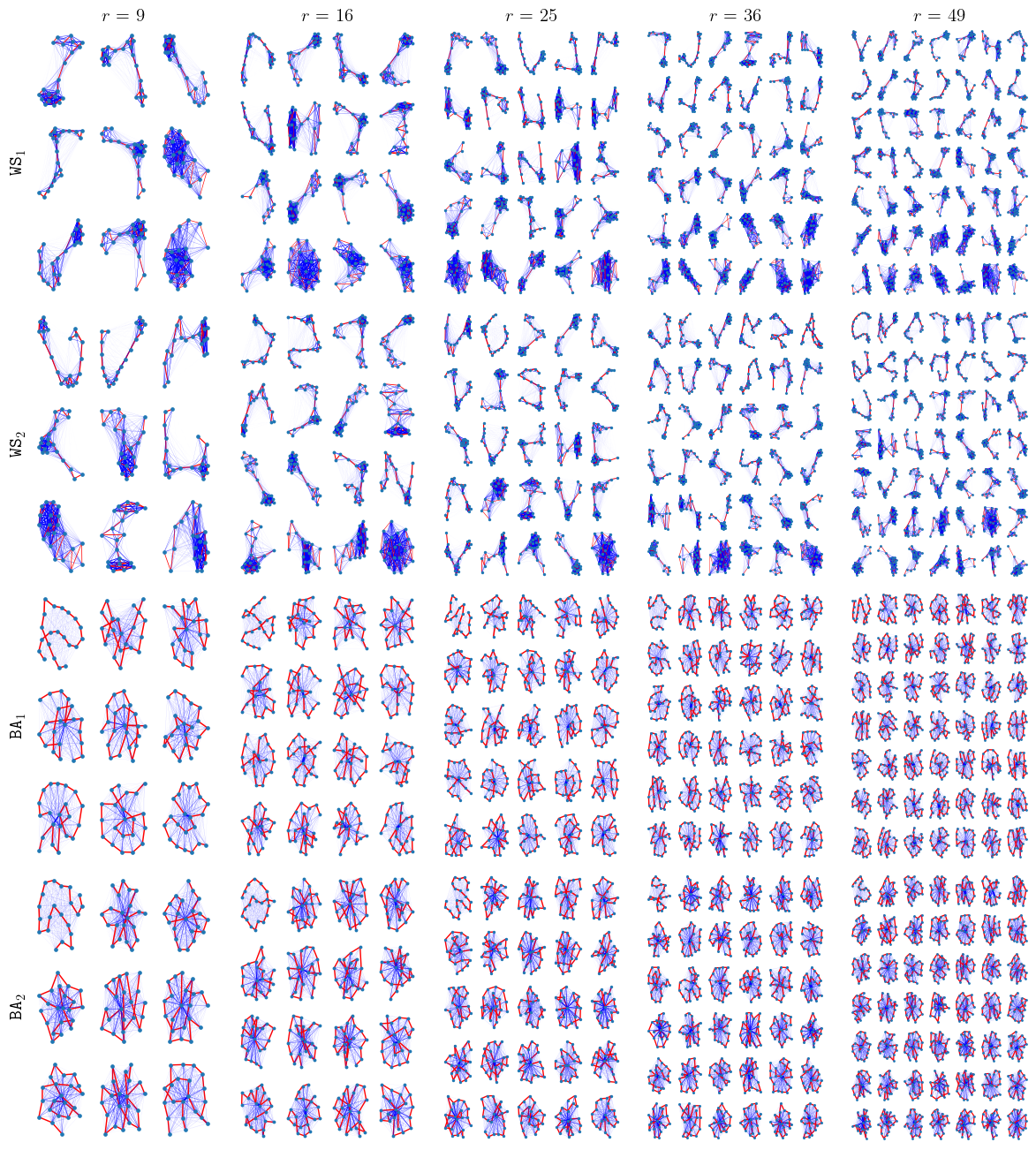}
							\caption{The $r \in \{9,16,25,36,49\}$ latent motifs 
							at scale $k = 21$ that we learn from the networks \dataset{WS$_1$}, \dataset{WS$_2$}, \dataset{BA$_1$}, and \dataset{BA$_2$}. The $r = 25$ column is identical to the $k=21$ column in Figure~\ref{fig:all_dictionaries_3}. See {Appendix}~\ref{section:experimental_details} for the details of these experiments.} 
							\label{fig:all_dictionaries_rank_3}
						\end{figure*}

						\newpage
						\begin{figure*}[h]
							\centering
							
							\includegraphics[width=1 \linewidth]{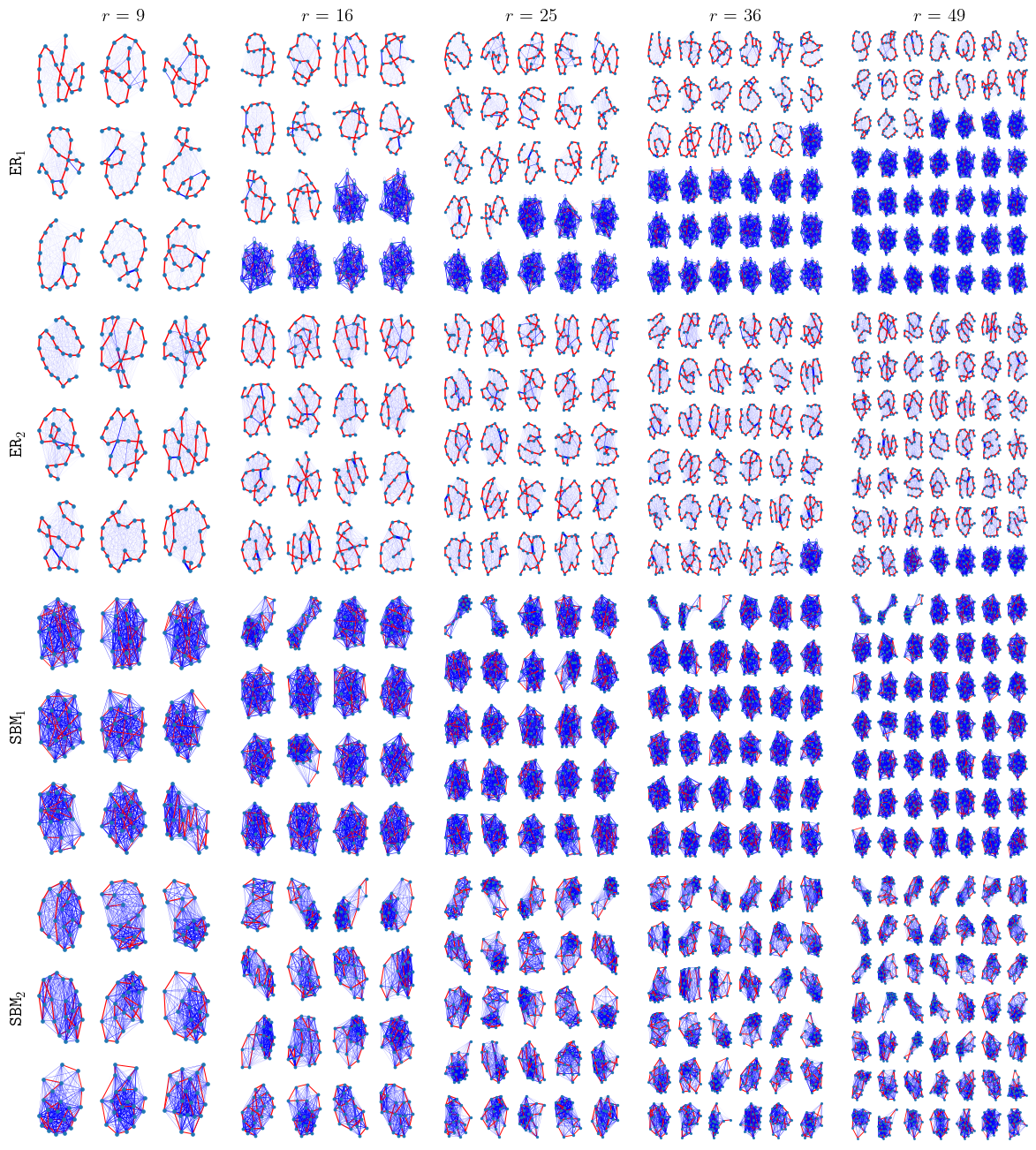}
							\caption{The $r \in \{9,16,25,36,49\}$ latent motifs
							at scale $k = 21$ that we learn from the networks \dataset{ER$_1$}, \dataset{ER$_2$}, \dataset{SBM$_1$}, and \dataset{SBM$_2$}. The $r = 25$ column is identical to the $k = 21$ column in Figure~\ref{fig:all_dictionaries_3}. See {Appendix}~\ref{section:experimental_details} for the details of these experiments.}
							\label{fig:all_dictionaries_rank_4}
						\end{figure*}

						\newpage 
					
						\clearpage
						
						\printbibliography[keyword = {appendix}]

					\end{document}